\documentclass{article}

\PassOptionsToPackage{numbers,compress}{natbib}
\usepackage[preprint]{neurips_2026}

\usepackage[utf8]{inputenc}
\usepackage[T1]{fontenc}
\usepackage{microtype}
\usepackage{hyperref}
\usepackage{url}
\usepackage{xcolor}

\usepackage{graphicx}
\usepackage{subcaption}
\usepackage{booktabs}
\usepackage{nicefrac}

\usepackage{amsmath}
\usepackage{amssymb}
\usepackage{amsfonts}
\usepackage{mathtools}
\usepackage{amsthm}
\usepackage{siunitx}

\usepackage{enumitem}
\usepackage[capitalize,noabbrev]{cleveref}

\DeclareMathOperator{\LinearOp}{Linear}

\DeclareMathOperator{\GLUOp}{GLU}
\DeclareMathOperator{\MLPOp}{MLP}
\DeclareMathOperator{\DropoutOp}{Dropout}
\DeclareMathOperator{\NormOp}{Norm}
\DeclareMathOperator{\SubLayerOp}{SubLayer}
\DeclareMathOperator{\CellOp}{Cell}
\DeclareMathOperator{\CellBOp}{\widetilde{\mathrm{Cell}}}

\newcommand{\Linear}[1]{\LinearOp\!\left(#1\right)}

\newcommand{\GLU}[1]{\GLUOp\!\left(#1\right)}
\newcommand{\MLP}[1]{\MLPOp\!\left(#1\right)}
\newcommand{\Dropout}[1]{\DropoutOp\!\left(#1\right)}
\newcommand{\Norm}[1]{\NormOp\!\left(#1\right)}
\newcommand{\SubLayer}[1]{\SubLayerOp\!\left(#1\right)}
\newcommand{\Cell}[1]{\CellOp\!\left(#1\right)}
\newcommand{\CellB}[1]{\CellBOp\!\left(#1\right)}

\theoremstyle{plain}
\newtheorem{theorem}{Theorem}[section]
\newtheorem{proposition}[theorem]{Proposition}

\theoremstyle{definition}
\newtheorem{definition}[theorem]{Definition}

\theoremstyle{remark}

\usepackage{tikz}
\usetikzlibrary{arrows.meta,positioning,calc,intersections,shapes,angles,patterns,quotes,backgrounds,fit}
\usepackage{circuitikz}
\ctikzset{tripoles/pmos style/emptycircle}
\ctikzset{tripoles/mos style/arrows}

\definecolor{mycolor1}{RGB}{31,119,180}
\definecolor{mycolor2}{RGB}{255,127,14}
\definecolor{mycolor3}{RGB}{44,160,44}
\definecolor{mycolorlight}{RGB}{204,204,204}

\tikzset{
  sum/.style={draw, circle, inner sep=1pt, minimum size=10pt, thick},
  signal/.style={-{Latex[length=3mm, width=2mm]}, thick},
  modsignal/.style={-{Latex[open, length=3mm, width=2mm]}, thick, dashed},
  fblock/.style={draw, thick, minimum height=35, minimum width=40, align=center, rounded corners=5pt},
  sig pic/.pic={
    \draw[thick, black] (-0.5,-0.5) .. controls (0.3,-0.5) and (-0.3,0.5) .. (0.5,0.5);
  },
  heaviside pic/.pic={
      \draw[thick, black]
        (-0.6,-0.4) -- (0, -0.4)   
        -- (0, 0.4)               
        -- (0.6, 0.4);            
    },
  sblock/.style={draw, thick, minimum height=55, minimum width=60, align=center, rounded corners=2pt,
  path picture={
  \pic at (path picture bounding box.center) {heaviside pic};
  }
  }
}

\title{Hardware-Software Co-Design of Scalable, Energy-Efficient Analog Recurrent Computations}

\author{%
  Arthur Fyon \\
  University of Liège \\
  Liège, Belgium \\
  \texttt{afyon@uliege.be} \\
  \And
  Julien Brandoit \\
  University of Liège \\
  Liège, Belgium \\
  \texttt{jbrandoit@uliege.be} \\
  \And
  Loris Mendolia \\
  University of Liège \\
  Liège, Belgium \\
  \texttt{lmendolia@uliege.be} \\
  \And
  Damien Ernst \\
  University of Liège \\
  Liège, Belgium \\
  \texttt{dernst@uliege.be} \\
  \AND
  Jean-Michel Redouté \\
  University of Liège \\
  Liège, Belgium \\
  \texttt{jean-michel.redoute@uliege.be} \\
  \And
  Guillaume Drion \\
  University of Liège \\
  Liège, Belgium \\
  \texttt{gdrion@uliege.be} \\
}

\begin{document}

\maketitle

\begin{abstract}
Always-on AI applications, from environmental sensors to biomedical implants, require ultra-low power consumption. Analog circuits offer a path to sub-microwatt inference, yet existing analog implementations are limited to feedforward architectures: extending them to recurrent dynamics has been considered impractical due to noise accumulation through temporal feedback. We demonstrate that this barrier can be overcome through hardware-software co-design. Specifically, we identify that Bistable Memory Recurrent Units (BMRUs), a class of Recurrent Neural Networks (RNNs) with discrete-valued outputs and hysteretic dynamics, admit an ultra-low power current-mode analog implementation which we design from first principles. The resulting circuit establishes a one-to-one correspondence between each learned parameter and a circuit element. The discrete outputs suppress analog noise by at least 20-fold at each cell boundary, breaking the noise accumulation that prevents analog recurrence. We reformulate BMRUs for first-quadrant operation with fixed thresholds, enabling the direct correspondence while preserving expressivity and trainability. Transistor-level simulations in \SI{180}{\nano\meter} Complementary Metal-Oxide-Semiconductor (CMOS) show near-perfect agreement between software predictions and circuit-level behavior, with the software model thereby serving as a high-fidelity simulator of the physical hardware at low computational cost. We leverage this fidelity to conduct large-scale noise immunity and power scaling analyses: the power cost of adding recurrence scales linearly with state dimension, while the feedforward layers dominating total power scale quadratically, meaning recurrence is added at linear marginal cost relative to the feedforward backbone. End-to-end keyword spotting achieves sub-microwatt inference at the RNN core. To our knowledge, this is the first demonstration of an analog recurrent neural network with structural noise immunity.
\end{abstract}

\section{Introduction}
The computational demands of modern Artificial Intelligence (AI) have grown exponentially, now consuming significant portions of global energy resources~\cite{faiz2024llmcarbon, barbierato2024toward, horowitz2014computing, sze2017efficient}. This challenge becomes particularly acute for always-on applications where battery life is paramount: environmental sensor networks, biomedical monitoring devices, and edge devices processing real-world sequential data streams~\cite{lee2021neural, seo2013neural, sharifshazileh2021electronic}. Current systems rely on general-purpose digital hardware inherited from Von Neumann architectures that separate memory and computation~\cite{sarpeshkar1998analog}, necessitating constant data movement and substantial energy overhead~\cite{sebastian2020memory, ielmini2018memory}.

Modern sequence-based AI is dominated by Transformers, whose attention mechanisms achieve state-of-the-art performance across domains~\cite{vaswani2017attention}. However, their quadratic complexity in sequence length and large memory footprint make them unsuitable for resource-constrained edge devices~\cite{fichtl2025end, tay2022efficient}. State-Space Models (SSMs) and RNNs offer more efficient alternatives for sequential processing. SSMs achieve parallelizable training through linear recurrence~\cite{gu2021efficiently, gu2024mamba, dao2024transformers, yue2024linear, wang2024mamba, gu2022parameterization}, yet their transient dynamics, states decaying exponentially toward equilibrium, prevent multistability~\cite{boyd1985fading}, limiting performance on tasks requiring persistent memory~\cite{vecoven2021bio, lambrechts2023warming}. Nonlinear RNNs provide multistability, essential for long-term memory~\cite{marder1996memory}, through nonlinear recurrent dynamics~\cite{hochreiter1997long, cho2014learning}, but their inherently sequential nature creates training bottlenecks~\cite{bengio1994learning, pascanu2013difficulty}.

Orthogonal to this algorithmic divide lies the hardware implementation challenge. Application-Specific Integrated Circuits (ASICs) achieve efficiency gains through parallelized operations and quantization~\cite{machupalli2022review, chen2019eyeriss}, with the TinyML movement extending this to microcontrollers, typically enabling milliwatt-level inference~\cite{heydari2025tiny, banbury2021mlperf, warden2019tinyml, lin2020mcunet, fedorov2019sparse}. Yet these systems remain fundamentally digital, constrained by sequential execution and separated memory~\cite{maass2002real}. Neuromorphic approaches offer alternatives: digital chips like Intel Loihi 2 implement Spiking Neural Networks (SNNs) with event-driven efficiency~\cite{davies2018loihi, orchard2021efficient, mead1990neuromorphic}, while analog crossbar arrays achieve in-memory computation through Kirchhoff's law summation~\cite{xia2019memristive, prezioso2015training, azghadi2020hardware, li2019long}. However, SNN performance still lags behind conventional RNNs on many sequential tasks~\cite{wu2024review, roy2019towards, lebow2021real, schuman2022opportunities}, and while digital FPGA implementations of recurrent architectures exist~\cite{chang2015recurrent}, analog implementations focus predominantly on feedforward networks~\cite{semenova2022noise, kendall2020training}: extending them to recurrent dynamics introduces noise accumulation through temporal feedback across timesteps~\cite{liu2002analog}. Programmable analog processors such as the Aspinity AML100~\cite{aspinity2022aml100} perform always-on event detection entirely in the analog domain at tens of microwatts, and dedicated digital neural decision processors such as the Syntiant NDP120~\cite{syntiant2021ndp120} achieve inference at hundreds of microwatts with support for Convolutional Neural Networks (CNNs), RNNs, and Fully Connected (FC) networks. Yet no existing fully analog approach implements recurrent dynamics: the AML100 is restricted to feedforward signal processing, and all sub-milliwatt digital processors remain orders of magnitude above what subthreshold analog circuits can achieve. \emph{No existing approach combines parallelizable training, multistable memory, analog computation, and sub-microwatt power consumption.}

In this paper, we focus on a specific class of recurrent architectures that addresses both the algorithmic and hardware challenges. On the algorithmic side, we seek parallelizable training while maintaining multistability for persistent memory. On the hardware side, we seek architectures whose computational primitives align naturally with ultra-low power analog circuits. Memory Recurrent Units (MRUs) satisfy the first requirement~\cite{BMRUref}. Unlike traditional RNNs requiring iterative settling~\cite{danieli2025pararnn}, MRUs combine instantaneous convergence with multistability, enabling parallel training via associative scans~\cite{martin2017parallelizing} while maintaining persistent memory without transient dynamics. The Bistable MRU (BMRU) variant achieves this through hysteresis, producing discrete-valued outputs. As we demonstrate in this paper, these discrete outputs and hysteretic dynamics can be matched by an ultra-low power current-mode analog circuit that we design from first principles, enabling genuine hardware-software co-design. We leverage this co-design to create the first fully analog RNN with four important characteristics:

\textbf{One-to-one hardware-software correspondence.} Rather than a loose algorithmic inspiration, each learned parameter in the reformulated BMRU maps to a single, programmable circuit element: thresholds become bias currents and weights become transistor width ratios. Cadence Spectre simulations confirm near-perfect agreement between the software model and transistor-level dynamics, not only on final predictions but at every intermediate signal throughout the network (Section~\ref{sec:hardware_validation}). This fidelity turns the software model into a proxy for the physical circuit, enabling large-scale analyses that would be prohibitively expensive in transistor-level simulation.

\textbf{Structural noise immunity.} Discrete-valued BMRU outputs block analog noise at each cell boundary, enabling robust recurrent architectures despite device variability. Leveraging the high-fidelity software simulator, we conduct a large-scale noise analysis across multiple benchmarks (Section~\ref{sec:robustness_power}) that confirms that this immunity holds under realistic operating conditions.

\textbf{Recurrence at low marginal cost.} The BMRUs scale linearly in power with state dimension, while the FC layers scale quadratically. Component-level power breakdown (Section~\ref{sec:robustness_power}) shows that at realistic network sizes, recurrence adds a small, linearly-scaling overhead on top of the feedforward cost. This means any analog feedforward platform could, in principle, gain temporal processing capabilities by integrating BMRU cells at marginal additional cost.

\textbf{Sub-microwatt operation.} Subthreshold operation with picoampere currents achieves approximately \SI{100}{\nano\watt} for the RNN inference core of a minimal proof-of-concept network on Keyword Spotting (KWS). Additional overhead from bias generation, shift registers, and routing remains well within the sub-microwatt envelope based on component-level estimates (Appendix~\ref{sec:programmable_overhead}). Clockless continuous-time analog dynamics eliminate clock overhead entirely, with latency reducing to propagation delay.

We validate our approach through Cadence transistor-level simulations in X-FAB \SI{180}{\nano\meter} CMOS, demonstrating end-to-end feasibility on KWS using the Google Speech Commands dataset, a standard benchmark for always-on edge applications~\cite{warden2018speech} (Section~\ref{sec:kws_proof_of_concept}). We also benchmark the proposed cell against LRU~\cite{orvieto2023resurrecting} and minGRU~\cite{feng2024minGRU} on standard sequence modeling tasks and on character-level language modeling to highlight the scalability of the proposed approach. A comprehensive review of related work is provided in Appendix~\ref{sec:related_work}.


\section{Hardware-software co-design of an analog RNN}\label{sec:codesign}
We build the hardware-software co-design in three steps: identifying a recurrent cell that is amenable to monolithic analog implementation (Section~\ref{sec:cell_choice}), designing a matching ultra-low power current-mode circuit (Section~\ref{sec:circuit}), and adapting the cell equations to achieve one-to-one mapping between learned parameter maps and circuit elements (Section~\ref{sec:FQ_BMRU}).

\subsection{Recurrent cell selection}\label{sec:cell_choice}

Three properties are essential for an analog RNN target. First, \emph{competitive performance} on standard sequence benchmarks. Second, \emph{parallelizable training} through associative scans~\cite{martin2017parallelizing}, so that the model can be trained at scale on GPUs. Third, and most importantly for monolithic analog integration, \emph{persistent memory without transient dynamics}: cells that converge instantaneously rather than through exponential settling eliminate the need for on-chip capacitors. This is what distinguishes BMRUs~\cite{BMRUref} from SSMs and classical RNNs. Together with the discrete-valued outputs that we exploit for noise immunity (Section~\ref{sec:robustness_power}), these properties make the BMRU a natural candidate for monolithic, capacitor-free, transistor-only analog implementation.

The BMRU combines these properties through a formulation based on bistability~\cite{BMRUref}:
\begin{align}
    \hat{h}_t &= W_x x_t + b_x, \label{eq:bmru_candidate}\\
    \beta_t &= \left|W_\beta x_t + b_\beta\right|, \label{eq:bmru_threshold}\\
    z_t &= \mathcal{H}\bigl(|\hat{h}_t| - \beta_t\bigr), \label{eq:bmru_gate}\\
    h_t &= z_t \odot S(\hat{h}_t) \odot \alpha + (1 - z_t) \odot h_{t-1}, \label{eq:bmru_update}
\end{align}
where $x_t$ and $h_t$ are the current input and hidden state, respectively, and $\beta_t$ is an input-dependent threshold. $\mathcal{H}(\cdot)$ denotes the Heaviside function, $S(\cdot)$ is the sign function, $\odot$ denotes the Hadamard product, $W_x$ and $W_\beta$ are learnable parameter matrices, and $b_x$, $b_\beta$ and $\alpha$ are learnable parameters.

This formulation implements adaptive thresholding with hysteresis. When the candidate magnitude $|\hat{h}_t|$ exceeds the input-dependent threshold $\beta_t$, the gate $z_t$ opens and the hidden state is mapped to $S(\hat{h}_t) \cdot \alpha \in \{-\alpha, +\alpha\}$. A state change only occurs when $S(\hat{h}_t) \neq S(h_{t-1})$, that is, when the candidate and current state have opposite signs. When $|\hat{h}_t| < \beta_t$, the gate remains closed ($z_t = 0$) and the hidden state persists unchanged. Two properties then characterize the BMRU: \emph{instantaneous convergence}, since the piecewise-constant activation removes iterative settling dynamics and enables parallel computation across timesteps during training; and \emph{discrete-valued outputs}, since hidden states are restricted to $\{-\alpha, +\alpha\}$. These two properties are what monolithic analog recurrence has been missing. Instantaneous convergence eliminates capacitive settling, enabling capacitor-free integration; discrete outputs introduce a thresholding nonlinearity at every cell boundary, blocking the noise propagation that has historically prevented analog RNN implementations.

Beyond these implementation-friendly properties, recent work in reinforcement learning demonstrates that BMRU-based agents trained on short-horizon environments generalize to substantially longer horizons at deployment~\cite{BMRURL}, providing evidence that the cell captures transferable temporal structure.

\subsection{Analog circuit implementation}\label{sec:circuit}
The properties described above allow for a simple implementation following three guiding principles (full circuit details appear in Appendix~\ref{sec:circuits}).

\textbf{Current-mode operation.} All signals are represented as currents, exploiting the exponential $I$-$V$ relationship of subthreshold MOS transistors for natural dynamic range compression and low power operation at picoampere signal levels. Addition and subtraction become trivial through Kirchhoff's Current Law (KCL).

\textbf{Transistor-only design.} By eliminating resistors, capacitors, and inductors, the circuit achieves true monolithic integration without external passives, minimizing area and enabling full CMOS compatibility. The BMRU instantaneous convergence eliminates the need for capacitors that would be required to implement the transient dynamics of classical RNNs or SSMs.

\textbf{Parameter mapping.} All learned parameters map directly to circuit elements: BMRU thresholds and amplitudes correspond to bias currents, while FC layer weights correspond to transistor width ratios in current mirrors. In our proof of concept, these are fixed at design time; Section~\ref{sec:perspandlimit} discusses a path toward post-fabrication programmability through binary-weighted current mirror banks.

\textbf{Bistable cell.}
We design a current-mode bistable cell based on a dual-Heaviside feedback structure (Figure~\ref{fig:fig1}A). It achieves independent programmability of three key parameters through bias currents: the threshold current $I_\text{thresh}$ determines the high switching point, the output current $I_\text{gain}$ sets the high-state amplitude, and the hysteresis width $I_\text{width}$ tunes the memory region. All three parameters are set by external bias currents independently, which can be tuned by adapting their respective gate voltage. Each cell requires 9 transistors, all operating in subthreshold, achieving per-cell nanowatt power consumption.

\begin{figure}[t!]
  \centering
  \includegraphics[width=\linewidth]{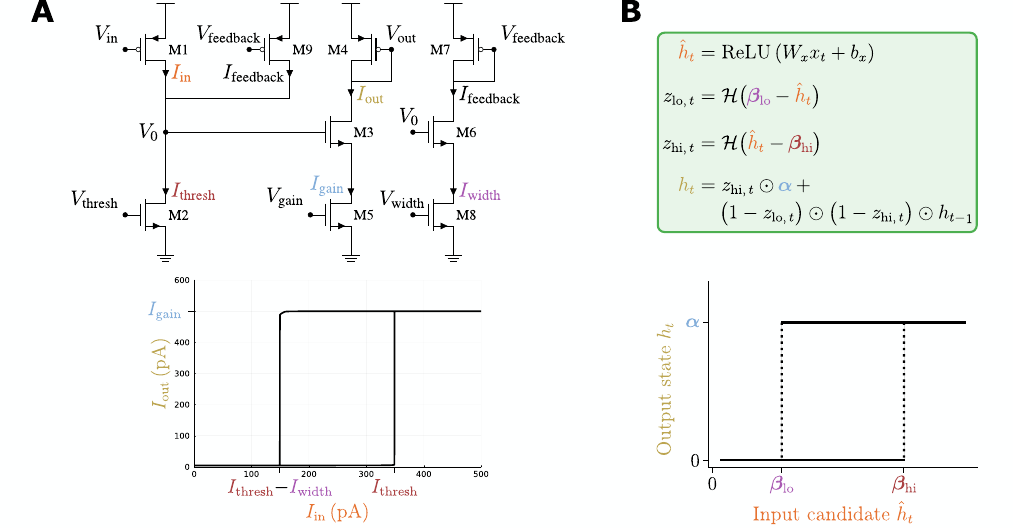}
  \caption{%
    \textbf{Current-mode analog implementation and FQ BMRU formulation.}
    \textbf{A.} Schematic of the ultra-low power current-mode bistable cell (top) and associated input-output current relationship (bottom). All thresholds and output gain are independently tunable via bias currents.
    \textbf{B.} FQ BMRU equations (top) and input candidate versus state relationship (bottom), with $\alpha$, $\beta_\text{lo}$ and $\beta_\text{hi}$ as learnable parameters. The correspondence between panels A and B illustrates the one-to-one mapping: each learned parameter maps directly to a tunable circuit parameter.
  }
  \label{fig:fig1}
\end{figure}

\textbf{Fully connected layers.}
Matrix-vector multiplication ($y = W\,x$) exploits KCL for analog summation using standard current-mirror techniques. Each input current $x_i$ feeds a bank of current mirrors producing partial currents $y_{ij} = w_{ij} x_i$, which sum at each output node: $y_j = \sum_i y_{ij}$. Transistor width ratios encode weight magnitudes:
\begin{equation}
    I_{\text{out},ij} = \frac{W_{\text{out}, ij}}{W_{\text{in}, i}}\cdot I_{\text{in},i}, \qquad I_{\text{out},j} = \sum_i I_{\text{out},ij},
\end{equation}
where $W_{\text{out}, ij}/W_{\text{in}, i}$ represents the width ratio implementing weight $w_{ij}$. Bias terms add through dedicated current sources. Weight precision is determined by layout matching~\cite{pelgrom1989matching, kinget2005device}, typically achieving 6--8 bit equivalent resolution in standard CMOS processes.

\textbf{Rectified linear units.}
The ReLU activation, placed after FC layers, emerges naturally from diode-connected NMOS/PMOS transistors, which conduct current only in one direction~\cite{mansoor2015silicon}. This implements $\text{ReLU}(I_{\text{out},j}) = \max(0, I_{\text{out},j})$ with no additional power overhead.

\textbf{Complete architecture.}
The full network comprises an input projection layer, $N$ stacked BMRU layers with inter-layer FC transformations, and an output classifier producing class currents. All components operate in the current domain with subthreshold biasing. All parameters, FC weights, BMRU thresholds, gains, and hysteresis widths, are set by bias currents, enabling the same fabricated chip to implement different trained networks.

\subsection{First-quadrant BMRU}\label{sec:FQ_BMRU}
The current-mode implementation described in Section~\ref{sec:circuit} operates with unipolar (positive-only) currents and fixed threshold parameters set by DC biases. The standard BMRU formulation presents two mismatches with these constraints. First, the sign function in Equation~\eqref{eq:bmru_update} produces bipolar outputs $\{-\alpha, +\alpha\}$, requiring differential signaling or dual power supplies, which increases circuit complexity and power consumption. Second, computing $\beta_t$ implies a linear transformation in Equation~\eqref{eq:bmru_threshold} and demands dedicated multiplication and addition circuitry evaluated at each timestep, adding substantial overhead to the fixed-bias tunability of the analog implementation. These mismatches motivate a reformulation that preserves the essential properties of the BMRU while aligning with current-mode circuit constraints. We propose a First-Quadrant BMRU (FQ BMRU) variant operating entirely with non-negative values (Figure~\ref{fig:fig1}B) that preserves expressivity (Appendix~\ref{sec:proofs}):
\begin{align}
    \hat{h}_t &= \text{ReLU}\left(W_x x_t + b_x\right), \label{eq:fq_candidate}\\
    z_{\text{lo},\, t} &= \mathcal{H}\bigl(\beta_\text{lo} - \hat{h}_t\bigr), \label{eq:fq_gate_lo}\\
    z_{\text{hi},\, t} &= \mathcal{H}\bigl(\hat{h}_t - \beta_\text{hi}\bigr), \label{eq:fq_gate_hi}\\
    h_t &= z_{\text{hi},\, t} \odot \alpha + \bigl(1 - z_{\text{lo},\, t}\bigr)\odot\bigl(1 - z_{\text{hi},\, t}\bigr) \odot h_{t-1}, \label{eq:fq_update}
\end{align}
where $\beta_\text{lo}$ and $\beta_\text{hi}$ are learned positive threshold parameters with $\beta_\text{hi} > \beta_\text{lo}$ (element-wise).

This formulation implements a window comparator with three operating regions. When $\hat{h}_t < \beta_\text{lo}$, the state resets (or stays) to zero. When $\hat{h}_t > \beta_\text{hi}$, the state sets (or stays) to $\alpha$. When $\beta_\text{lo} \leq \hat{h}_t \leq \beta_\text{hi}$, the state persists from the previous timestep. The hysteresis window $[\beta_\text{lo}, \beta_\text{hi}]$ provides noise immunity while enabling bistable memory. This reformulation offers three advantages for hardware implementation. \emph{(i) Unipolar signals}: input candidate $\hat{h}_t$ and hidden state $h_t$ remain non-negative, enabling single-ended current-mode circuits without differential pairs or negative supplies. \emph{(ii) Fixed thresholds}: parameters $\beta_\text{lo}$ and $\beta_\text{hi}$ are learned during training but remain constant during inference, mapping directly to fixed DC bias currents. \emph{(iii) Direct parameter correspondence}: each learned parameter maps to a specific circuit element, with $\hat{h}_t$ corresponding to the input current $I_\text{in}$, $h_t$ to the output current $I_\text{out}$, $\beta_\text{hi}$ to $I_\text{thresh}$, $\beta_\text{lo}$ to $I_\text{thresh} - I_\text{width}$, and $\alpha$ to $I_\text{gain}$ (color coded in Figure~\ref{fig:fig1}).

\textbf{Empirical trainability.}
\emph{While expressivity is preserved in theory}, the reformulation could still affect gradient-based training. We compare the FQ BMRU against the original BMRU and two recent parallelizable RNN baselines, LRU~\cite{orvieto2023resurrecting} and minGRU~\cite{feng2024minGRU}, on a unified suite of sequence modeling benchmarks spanning pixel-level image classification (sequential and permuted MNIST), spoken-digit keyword spotting (dKWS), the ListOps task from the Long Range Arena~\cite{tay2021long}, and character-level language modeling on Shakespeare~\cite{feng2024minGRU, karpathy2022nanogpt} (Table~\ref{tab:benchmarks}). For classification tasks, loss is computed via average pooling over timesteps and results are averaged over 5 seeds. Across tasks, the FQ BMRU matches or closely approaches the baselines on most tasks, with the largest gap appearing on ListOps. We attribute this gap to the optimization cost of fixed thresholds rather than to reduced representational capacity, consistent with the expressivity-versus-trainability distinction above. Notably, no analog co-design approach exists for LRU or minGRU, making FQ BMRU the only architecture with a demonstrated path to ultra-low power analog implementation. Complete training procedures are detailed in Appendix~\ref{sec:softwareKWS}.

\begin{table}[t!]
\caption{%
  \textbf{Unified benchmark results.} Accuracy (\%) on classification tasks (higher is better); cross-entropy loss on Shakespeare character-level language modeling (lower is better). Classification models use two recurrent layers with state dimension 64; Shakespeare uses a depth-6 model with state and model dimension 256 following~\citet{feng2024minGRU}. ListOps follows the Long Range Arena protocol~\cite{tay2021long}. Results averaged over 5 seeds. $^*$LRU and minGRU use identical architecture and training hyperparameters.
}
\label{tab:benchmarks}
\centering
\small
\setlength{\tabcolsep}{5pt}
\begin{tabular}{lccccc c}
\toprule
& \multicolumn{4}{c}{Accuracy (\%) $\uparrow$} & Loss $\downarrow$ \\
\cmidrule(lr){2-5} \cmidrule(lr){6-6}
Model & sMNIST & pMNIST & dKWS & ListOps & Shakespeare \\
\midrule
BMRU       & 
\shortstack{95.4 \\ {\scriptsize [94.9; 96.3]}} & 
\shortstack{94.2 \\ {\scriptsize [92.9; 95.8]}} & 
\shortstack{94.8 \\ {\scriptsize [94.4; 95.2]}} & 
\shortstack{41.7 \\ {\scriptsize [41.0; 42.7]}} & 
\shortstack{1.439 \\ {\scriptsize [1.437; 1.446]}}  \\
FQ BMRU  & 
\shortstack{95.6 \\ {\scriptsize [95.2; 96.0]}} & 
\shortstack{94.8 \\ {\scriptsize [94.4; 95.2]}} & 
\shortstack{92.9 \\ {\scriptsize [92.7; 93.2]}} & 
\shortstack{37.5 \\ {\scriptsize [37.0; 38.3]}} & 
\shortstack{1.484 \\ {\scriptsize [1.477; 1.492]}} \\
LRU$^*$    & 
\shortstack{92.6 \\ {\scriptsize [90.4; 93.8]}} & 
\shortstack{95.9 \\ {\scriptsize [95.0; 96.7]}} & 
\shortstack{96.4 \\ {\scriptsize [95.9; 96.8]}} &  
\shortstack{41.0 \\ {\scriptsize [37.9; 49.3]}} & 
\shortstack{1.480 \\ {\scriptsize [1.474; 1.487]}} \\
minGRU$^*$ & 
\shortstack{93.4 \\ {\scriptsize [88.5; 96.5]}} & 
\shortstack{95.6 \\ {\scriptsize [94.9; 96.6]}} & 
\shortstack{96.3 \\ {\scriptsize [96.0; 96.6]}} & 
\shortstack{39.0 \\ {\scriptsize [38.4; 39.4]}} & 
\shortstack{1.451 \\ {\scriptsize [1.445; 1.455]}} \\
\bottomrule
\end{tabular}

\end{table}


\section{Proof of concept: keyword spotting}\label{sec:kws_proof_of_concept}
We validate the co-design through a complete pipeline: the network is trained entirely in software, and the learned parameters are then mapped to transistor dimensions and bias currents for inference in transistor-level Cadence simulations. The target task is single-word (``yes'') keyword spotting using the Google Speech Commands dataset~\cite{warden2018speech} (Figure~\ref{fig:fig2}A, bottom), a standard benchmark for always-on edge applications requiring continuous audio monitoring with minimal power~\cite{lopez2022deep, hinton2012deep, chen2014small, arik2017convolutional}.

Audio preprocessing extracts 13-dimensional Mel-Frequency Cepstral Coefficients (MFCCs)~\cite{davis1980comparison} from 1-second clips at 100 frames per second, yielding input sequences of 101 timesteps. For this proof-of-concept demonstration, MFCC extraction occurs off-chip. MFCC front-ends operating at comparable power levels have been demonstrated~\cite{villamizar2021switchedcap, yang2019jssc_afe}, confirming that a complete signal chain from microphone to classification at microwatt-level total power is feasible. On-chip feature extraction remains a direction for future integration. The network architecture comprises an input projection layer mapping 13-dimensional MFCC features to the recurrent state dimension $d$, $N$ stacked FQ BMRU layers, and a binary output classifier (Figure~\ref{fig:fig2}A, top). Trained parameters, FC weights and BMRU thresholds, are then mapped to transistor width ratios and bias currents respectively for Cadence transistor-level simulation.

\begin{figure}[t!]
  \centering
  \includegraphics[width=0.94\linewidth]{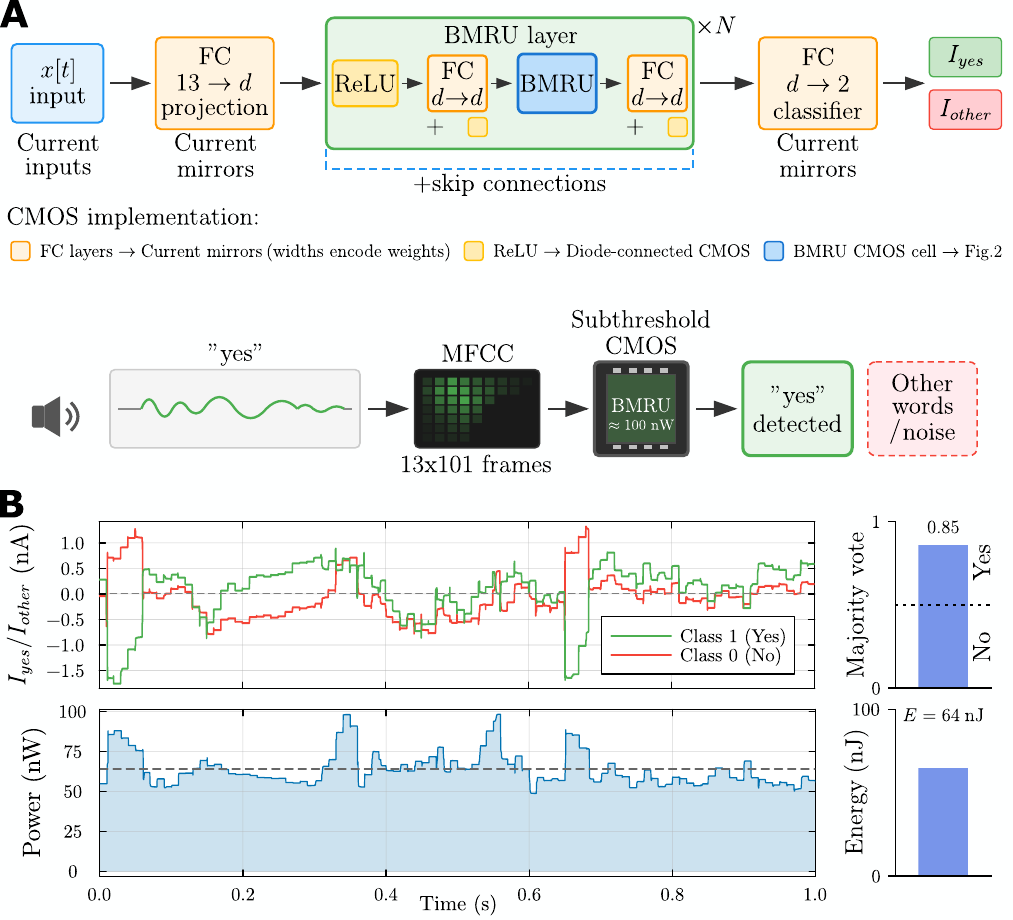}
  \caption{%
    \textbf{Analog CMOS implementation of a complete BMRU-based RNN for ``yes'' KWS.}
    \textbf{A.} Complete network architecture where all operations are computed using analog primitives whose behavior emerges from the physical properties of subthreshold transistors (top). For this proof of concept, a minimal configuration with $N=2$ layers and state dimension $d=4$ is implemented. KWS task for ``yes'' recognition. MFCC extraction occurs off-chip for simplicity; future implementations will integrate full inference from raw audio (bottom). The analog chip produces two output currents representing class logits.
    \textbf{B.} Positive inference example simulated in Cadence. Top: time evolution of output logit currents, with majority vote yielding correct classification. Bottom: instantaneous power consumption remaining below \SI{100}{\nano\watt} for the RNN core.
  }
  \label{fig:fig2}
\end{figure}

The complete analog circuit (768 transistors) is implemented in the \SI{180}{\nano\meter} CMOS X-FAB technology using Cadence Virtuoso~\cite{cadence2023virtuoso}. All simulations are conducted using Spectre with PRIMLIB \SI{180}{\nano\meter} models (ne/pe devices) at a supply voltage of \SI{1.8}{\volt}. Transient simulations validate inference on test samples, with input MFCC sequences converted to current waveforms and output class currents compared to determine classification decisions.

\subsection{Software baseline}\label{sec:software_baseline}
We evaluated FQ BMRU accuracy on balanced test sets comprising all 397 positive samples (``yes'') paired with 397 randomly selected negative samples (other words or noise) from the Google Speech Commands dataset (accuracy table and complete training procedure can be found in Appendix~\ref{sec:softwareKWS}). Even with a minimal $d=4$ configuration, the network achieves \SI{93.86}{\percent} accuracy. Performance improves steadily with increasing state dimension, reaching \SI{96.98}{\percent} at $d=8$ and \SI{98.08}{\percent} at $d=16$, before plateauing around \SI{97.5}{\percent} for $d \geq 64$. The consistently tight min-max ranges indicate stable training and good generalization despite the random sampling. The primary goal of this work is to demonstrate end-to-end co-design feasibility rather than maximizing accuracy; accordingly, hardware validation focuses on the minimal $d=4$ configuration to establish the complete pipeline from training to analog inference. Preliminary quantization experiments confirm that restricting all parameters to 4-bit precision degrades accuracy negligibly relative to full precision, even without quantization-aware training (Appendix~\ref{sec:quantization}).

\subsection{Hardware validation}\label{sec:hardware_validation}
Cadence simulations of the complete analog implementation with $d=4$ demonstrate successful inference. Figure~\ref{fig:fig2}B illustrates representative inference traces, showing the output class currents $I_\text{yes}$ and $I_\text{other}$ over time for a positive sample. The BMRU cells exhibit clean switching behavior, and the correct class consistently produces higher output current, enabling reliable discrimination. Across 50 randomly selected test samples, hardware predictions match software in 49 cases. The single discrepancy occurs on a sample where software confidence was minimal (51 vs.\ 50 timestep majority), a tie rather than a systematic error. Additional inference traces covering positive, negative, and boundary samples appear in Appendix~\ref{sec:inference_traces}.

Operating entirely in subthreshold, the circuit achieves approximately \SI{100}{\nano\watt} average power consumption during continuous inference for the RNN core at the $d=4$ proof of concept (Figure~\ref{fig:fig2}B bottom). A component-level power breakdown (Appendix~\ref{sec:power_breakdown}) confirms the approximately even split between BMRU cells and FC layers at $d=4$, with BMRU cells exhibiting substantially lower power variance across inferences, consistent with stable discrete-output dynamics.

We validate robustness to manufacturing and environmental variation through transistor-level Monte Carlo and corner simulations. Monte Carlo mismatch analysis (200 samples, $3\sigma$ on all transistors) confirms that misclassifications under mismatch occur only when the network is already uncertain under nominal conditions (Appendix~\ref{sec:mismatch_analysis}). The bistable dynamics provide inherent mismatch immunity: discrete thresholding ensures that small parameter perturbations do not propagate to output errors unless the input is already near the decision boundary. PVT (Process, Voltage, Temperature) corner analysis validates robustness across all five process corners (TT, FF, SS, FS, SF), the temperature range [$-27^\circ$C, $81^\circ$C], and $\pm 10\%$ supply voltage variation (Appendix~\ref{sec:pvt_corners}). Classification correctness is maintained across all conditions for representative inferences.

Beyond matching final predictions, the hardware implementation matches the software model at every intermediate stage. Appendix~\ref{sec:signal_agreement} provides a stage-by-stage comparison between software predictions and Cadence-simulated signals at the input projection, both FQ BMRU layer candidates and states, the skip connection, and the output logits. The close agreement throughout validates the one-to-one correspondence between learned parameters and circuit elements, and demonstrates that the FQ reformulation produces a faithful mapping between the algorithmic and physical domains.

The signal-level comparison also reveals the noise suppression mechanism in action. At the second recurrent layer, where accumulated analog errors are largest, the candidate signal $\hat{h}_t$ shows a mean absolute error of approximately \SI{60}{pA} relative to software predictions. After discrete thresholding at the cell boundary, this error drops to approximately \SI{3}{pA} on the output state $h_t$, a 20-fold suppression (Appendix~\ref{sec:signal_agreement}, Figure~\ref{fig:agreement_schematic}). The residual error is dominated by transistor subthreshold leakage when cells should output zero current, rather than by systematic signal corruption. This confirms that the candidate-to-state thresholding acts as a structural error-suppressing nonlinearity, preventing noise from accumulating through the recurrent feedback loop. A detailed stage-by-stage comparison is provided in Appendix~\ref{sec:signal_agreement}.

Finally, to confirm that the approach generalizes beyond binary classification, we evaluate the FQ BMRU on 11-class digit recognition (``zero'' through ``nine'' plus background noise). A $2\times 16$ architecture achieves competitive accuracy and substantially increases output margin separation between classes (Appendix~\ref{sec:multiclass_kws}), improving robustness to mismatch. Such networks remain within sub-microwatt consumption estimates (Appendix~\ref{sec:power_breakdown}, Table~\ref{tab:power_scaling}). Given the strong hardware-software agreement demonstrated at $d=4$, we expect similar behavior in hardware at larger dimensions.


\section{Noise immunity and power scaling}\label{sec:robustness_power}
The end-to-end agreement established in Section~\ref{sec:kws_proof_of_concept} demonstrates that the software model is a high-fidelity simulator of the physical circuit. We exploit this fidelity to conduct noise immunity and power scaling analyses on three benchmark tasks from Table~\ref{tab:benchmarks} (sMNIST, pMNIST, dKWS), sweeping over noise magnitudes and state dimensions that would be prohibitively expensive in transistor-level simulation. All three architectures are retrained using the hardware backbone (Appendix~\ref{sec:backbone_hardware}).

We first test network robustness to increasing levels of analog noise and mismatch. We extract worst-case noise levels from transistor-level simulations under $3\sigma$ mismatch and PVT variations (200 Monte Carlo samples across all building blocks), expressed as a worst-case fraction of nominal signal amplitude across multiple operating points, and inject this noise into the software simulator. Multiple noise levels are tested: 0.5$\times$, 1$\times$, 2$\times$, and beyond the measured analog noise magnitude. For each benchmark and test sample, we generate 10 noisy instantiations of the network for each inference.

Figure~\ref{fig:noise_analysis} presents the results for FQ BMRU, LRU, and minGRU, with noise injected at the same relative magnitude for fairness. At the measured analog noise level, both FQ BMRU and minGRU operate with no significant accuracy decrease. LRU exhibits catastrophic failure even at low noise levels, consistent with its purely linear recurrence offering no mechanism to attenuate accumulated perturbations. As noise increases beyond the measured level, FQ BMRU maintains robustness up to approximately 2$\times$ the analog noise before a sharp transition to failure. minGRU shows a similar pattern with slightly higher tolerance, which we attribute to multiplicative gating acting as an implicit noise regularizer. These results confirm that discrete-output dynamics provide sufficient noise immunity for scalable analog deployment, and that this immunity is a structural property of the architecture rather than an artifact of a specific circuit or task.

\begin{figure}[t!]
  \centering
  \includegraphics[width=\linewidth]{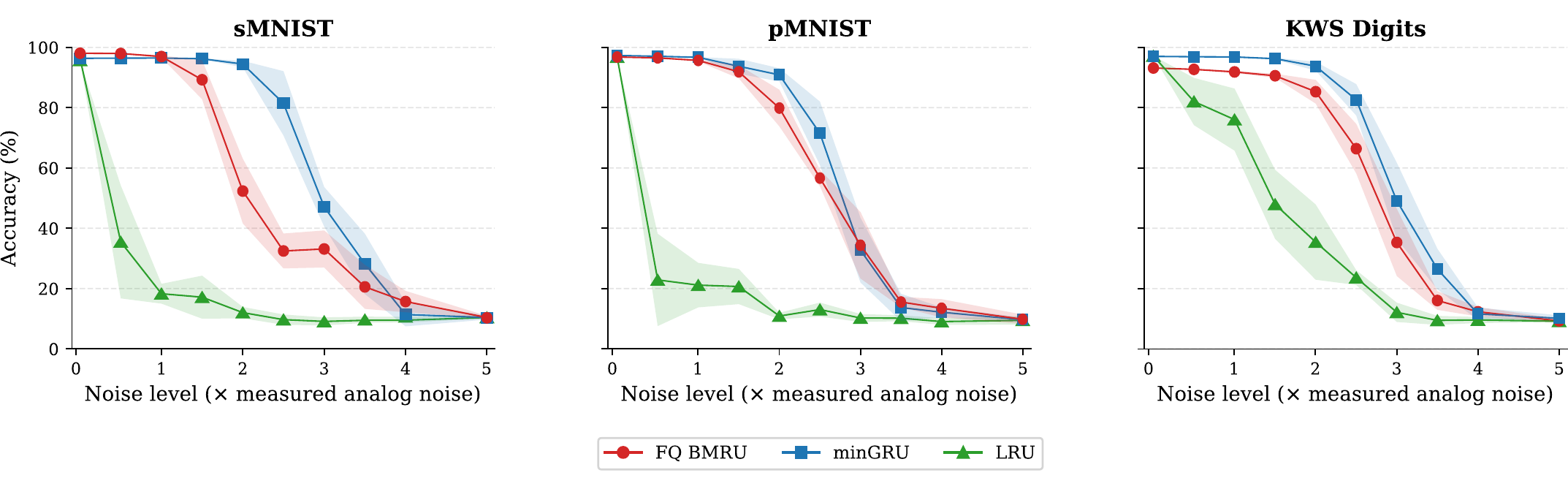}
  \caption{%
    \textbf{Large-scale noise robustness analysis across three benchmarks (sMNIST, pMNIST and dKWS).}
    Accuracy as a function of injected noise level (relative to measured analog noise from transistor-level simulations) for FQ BMRU, LRU, and minGRU. At analog noise level, FQ BMRU and minGRU maintain full accuracy, while LRU fails catastrophically. FQ BMRU exhibits robust performance up to approximately 2$\times$ the analog noise level across all benchmarks. Results computed over the full test set with 10 noisy instantiations per sample.
  }
  \label{fig:noise_analysis}
\end{figure}

We then estimated power scaling based on per-component measurements at $d=4$ from Cadence Spectre simulation, and extrapolated to higher state dimensions (Appendix~\ref{sec:power_breakdown}). At $d=32$, FC layers consume approximately 6$\times$ the power of BMRU cells, meaning recurrence adds less than 15\% overhead relative to a hypothetical feedforward-only analog network of the same dimension. This power structure has an important implication: analog feedforward inference using current mirrors and KCL is well-established, and the FC layers implementing it dominate total power at practical network sizes. The BMRU cells add recurrence, and with it, the ability to process temporal sequences rather than static inputs, at a marginal power cost that is linear and modest. 


\section{Perspectives and limitations}\label{sec:perspandlimit}
Several architectural extensions could enhance capability while preserving the co-design principles.

\textbf{Increased state dimension.} Scaling to larger state dimensions ($d = 8, 16, 32$) improves accuracy (Appendix~\ref{sec:softwareKWS}, Table~\ref{tab:accuracy}), with sub-microwatt operation remaining achievable up to $d=16$ where accuracy saturates (Appendix~\ref{sec:power_breakdown}, Table~\ref{tab:power_scaling}). Increasing state dimension also substantially increases output margin separation in multi-class settings, improving robustness to mismatch.

\textbf{Quantized parameter mapping.} Programmable parameters can be achieved through binary-weighted current mirror banks~\cite{allen2011cmos, razavi2001design} enabled through digital registers. This approach, analogous to programmable crossbar arrays~\cite{xia2019memristive, strukov2008missing, alibart2013pattern}, enables $B$-bit resolution with $B$ transistors per parameter. Preliminary experiments indicate that 4-bit quantization introduces negligible accuracy degradation compared to full-precision training, even without quantization-aware pretraining~\cite{gholami2022survey, courbariaux2015binaryconnect, hubara2016binarized, yin2019understanding}. Quantized accuracy on ``yes'' KWS is reported in Appendix~\ref{sec:quantization}. Estimated power and area overhead for the programmable version, including binary-weighted current mirrors, and shift registers is provided in Appendix~\ref{sec:programmable_overhead}.

\textbf{Limitations.}  The \SI{180}{\nano\meter} CMOS simulations, while validated in Cadence with Monte Carlo mismatch and PVT corner analysis (Appendices~\ref{sec:mismatch_analysis} and~\ref{sec:pvt_corners}), await physical fabrication and measurement, and focus on the binary KWS classification task. Multi-class KWS and more complex benchmarks are evaluated in software with realistic noise levels, but remain to be validated on hardware.

\section{Conclusion}
We presented the first fully analog RNN implementation that combines structural noise immunity with parallelizable training, through hardware-software co-design of the BMRU. By reformulating the BMRU to operate in the first quadrant with fixed thresholds, and by designing from first principles an ultra-low power current-mode bistable circuit whose behavior matches the reformulated cell, we established a one-to-one correspondence between each learned parameter and a programmable circuit element. Transistor-level simulations in \SI{180}{\nano\meter} CMOS show near-perfect agreement between software predictions and circuit-level dynamics at every stage of the network, turning the software model into a high-fidelity simulator of the physical hardware. We leveraged this fidelity to establish, across multiple benchmarks, that discrete-valued outputs provide structural noise immunity, breaking the noise accumulation barrier that has prevented analog recurrence, and that recurrence can be added to analog inference at linear marginal power cost relative to the quadratically-scaling feedforward backbone. End-to-end validation on keyword spotting achieves sub-microwatt power consumption for the RNN core. As always-on edge AI becomes increasingly prevalent, this co-design approach opens a path toward adding temporal intelligence to analog inference platforms at minimal additional cost. Broader impacts and code/data availability are discussed in Appendices~\ref{sec:broaderimpact} and~\ref{sec:softanddata}.

\section*{Acknowledgements}
A.F., J.B. and L.M. are, respectively, postdoctoral researcher (grant ASP-REN-40024838) and FRIA grantees (grants FRIA-40038025 and FRIA-B2-40029909) of the Fonds de la Recherche Scientifique -- FNRS. This work is supported by the Belgian Government, Federal Public Service Policy and Support, through grant NEMODEI2. Computational resources have been provided by the Consortium des Équipements de Calcul Intensif (CÉCI), funded by the FNRS under Grant No.\ 2.5020.11 and by the Walloon Region. The present research also benefited from computational resources made available on Lucia, the Tier-1 supercomputer of the Walloon Region, infrastructure funded by the Walloon Region under grant agreement No. 1910247. 

\section*{Disclosure}
This work has been the subject of patent applications under numbers EP26175243.0 and EP26175248.9.

\bibliographystyle{unsrtnat}
\bibliography{bib/bib}

\clearpage
\appendix

\section{Related work}\label{sec:related_work}

This appendix provides a comprehensive overview of related work across the domains bridged by our contribution: energy-efficient sequence models, neuromorphic and analog computing, and hardware-software co-design for neural networks.

\subsection{Sequence models for edge deployment}

\paragraph{Transformers and attention mechanisms.} Transformer architectures have achieved state-of-the-art performance across sequence modeling domains~\cite{vaswani2017attention}, but their quadratic complexity in sequence length and large memory footprint make them unsuitable for resource-constrained edge devices~\cite{fichtl2025end, tay2022efficient}. Various efficiency improvements have been proposed, including sparse attention patterns and linear attention variants, but these typically remain energy inefficient.

\paragraph{State-space models.} SSMs offer an efficient alternative through linear recurrence, enabling parallelizable training~\cite{gu2021efficiently, gu2024mamba, dao2024transformers, orvieto2023resurrecting, yue2024linear, wang2024mamba, gu2022parameterization}. However, their transient dynamics---states decaying exponentially toward equilibrium---prevent multistability~\cite{boyd1985fading}, limiting performance on tasks requiring persistent memory~\cite{vecoven2021bio, lambrechts2023warming}.

\paragraph{Recurrent neural networks.} Traditional RNNs provide multistability essential for long-term memory~\cite{marder1996memory} through nonlinear recurrent dynamics. LSTMs~\cite{hochreiter1997long} and GRUs~\cite{cho2014learning} have dominated sequence modeling for decades, though their complex gating mechanisms and inherently sequential nature create both training bottlenecks~\cite{bengio1994learning, pascanu2013difficulty} and hardware implementation challenges. Recent work on parallelizable RNN training~\cite{martin2017parallelizing, danieli2025pararnn} addresses the training bottleneck, and the Memory Recurrent Unit (MRU) family~\cite{BMRUref} combines instantaneous convergence with multistability, enabling parallel training via associative scans while maintaining persistent memory.

\subsection{TinyML and digital edge AI}

\paragraph{Application-specific integrated circuits.} ASICs achieve efficiency gains through parallelized operations and quantization~\cite{machupalli2022review, chen2019eyeriss}, but remain fundamentally digital with separated memory and computation~\cite{maass2002real}.

\paragraph{TinyML on microcontrollers.} The TinyML movement has extended neural network deployment to microcontrollers, typically enabling milliwatt-level inference~\cite{heydari2025tiny, banbury2021mlperf, warden2019tinyml, lin2020mcunet, fedorov2019sparse}. KWS has served as a primary benchmark for these systems~\cite{warden2018speech, lopez2022deep, hinton2012deep, chen2014small, arik2017convolutional, zhang2017hello, mazumder2022fast}. While significant progress has been made in model compression and efficient architectures~\cite{lin2020mcunet}, power consumption remains orders of magnitude above what analog approaches can achieve.

\paragraph{Quantization techniques.} Quantization-aware training and post-training quantization enable deployment on resource-constrained hardware~\cite{gholami2022survey, courbariaux2015binaryconnect, hubara2016binarized, yin2019understanding}. Our work demonstrates that the discrete-output nature of BMRUs provides inherent robustness to quantization, achieving minimal accuracy degradation with 4-bit parameter precision.

\subsection{Neuromorphic computing}

\paragraph{Spiking neural networks.} Digital neuromorphic chips like Intel Loihi 2 implement spiking neural networks (SNNs) with event-driven efficiency~\cite{davies2018loihi, orchard2021efficient, mead1990neuromorphic}. SNNs offer biological plausibility and potential energy benefits through sparse, event-driven computation. However, SNN performance still lags behind conventional RNNs on many sequence tasks~\cite{wu2024review, roy2019towards, lebow2021real, schuman2022opportunities}, and training remains challenging due to non-differentiable spike functions. Surrogate gradient methods~\cite{neftci2019surrogate, tian2023recent} have enabled gradient-based SNN training, techniques we also employ for training BMRUs with Heaviside nonlinearities.

\paragraph{Analog feedforward inference.} Analog crossbar arrays achieve in-memory computation through KCL~\cite{xia2019memristive, prezioso2015training, azghadi2020hardware, li2019long}. The Aspinity AML100~\cite{aspinity2022aml100} performs always-on event detection entirely in the analog domain at $\sim$30--100~$\mu$W using configurable analog blocks, representing the state of the art in programmable analog ML. However, all fully existing analog implementations---including crossbar arrays, the AML100, and flash-based approaches---are restricted to feedforward architectures. Extending fully analog computation to recurrent dynamics has been considered impractical due to noise accumulation through temporal feedback~\cite{semenova2022noise, kendall2020training}. Our work demonstrates that discrete-output BMRU dynamics overcome this barrier, enabling the first analog RNN.

\paragraph{FPGA implementations.} Digital FPGA implementations of recurrent architectures exist~\cite{chang2015recurrent}, offering flexibility and moderate efficiency gains over general-purpose processors, but lack the power efficiency of subthreshold analog circuits.

\subsection{Analog circuit design for neural networks}

\paragraph{Subthreshold analog computing.} Subthreshold operation exploits the exponential current-voltage relationship of MOSFETs below threshold for ultra-low power signal processing~\cite{soeleman2002robust, vittoz1977cmos, mead1989analog}. This regime enables computation at picoampere current levels---six orders of magnitude below typical digital switching currents. Our implementation builds on these principles, operating entirely in subthreshold with nanowatt power consumption.

\paragraph{Schmitt triggers.} The Schmitt trigger implements bistable thresholding with hysteresis, providing noise immunity through its characteristic hysteresis window~\cite{filanovsky1994cmos}. Traditional voltage-mode implementations operate with milliwatt power consumption~\cite{siripruchyanun2021fully}, unsuitable for ultra-low power applications. Our work leverages a current-mode Schmitt trigger design that achieves fully tunable thresholds and output amplitudes while operating in the nanowatt regime.

\paragraph{Current-mode circuits.} Current-mode signal processing offers advantages for analog neural networks, including natural summation via KCL and compatibility with subthreshold operation. Current mirrors with controlled width ratios implement weighted connections~\cite{allen2011cmos, razavi2001design}, with precision determined by layout matching~\cite{pelgrom1989matching, kinget2005device}, typically achieving 6--8 bit equivalent resolution in standard CMOS processes.

\subsection{Hardware-software co-design}

\paragraph{Co-design principles.} Hardware-software co-design seeks to jointly optimize algorithms and their physical implementations~\cite{wolf1994hardware}. Rather than adapting algorithms to existing hardware or designing hardware to implement existing algorithms, co-design identifies synergies where algorithmic structure aligns with hardware capabilities. Our work exemplifies this approach: the BMRU discrete outputs and hysteretic dynamics map directly to Schmitt trigger behavior, enabling efficiency gains unattainable by optimizing either domain independently.

\paragraph{Neural network accelerators.} Prior co-design efforts have focused primarily on digital accelerators~\cite{chen2019eyeriss, machupalli2022review}, optimizing dataflow and memory hierarchy for specific network architectures. Analog co-design remains less explored, particularly for recurrent networks where temporal feedback complicates noise management.

\subsection{Energy consumption in AI systems}

The computational demands of modern AI have grown exponentially, with AI systems now consuming significant portions of global energy resources~\cite{faiz2024llmcarbon, barbierato2024toward}. Energy efficiency in computing has been extensively studied~\cite{horowitz2014computing, sze2017efficient}, with particular attention to the memory wall and data movement costs inherent to Von Neumann architectures~\cite{sarpeshkar1998analog, sebastian2020memory, ielmini2018memory}. Always-on applications---environmental sensor networks, biomedical monitoring devices, and edge devices processing sequential data~\cite{lee2021neural, seo2013neural, sharifshazileh2021electronic}---face particularly acute power constraints that motivate our ultra-low power approach.

\section{Expressivity proofs}\label{sec:proofs}

This appendix establishes that the simplifications introduced in the FQ BMRU do not reduce expressivity compared to the original BMRU in an approximate sense. We prove two results: (1) restricting outputs to $\{0, \alpha\}$ instead of $\{-\alpha, +\alpha\}$ preserves expressivity exactly (Proposition~\ref{prop:output}), and (2) replacing input-dependent thresholds with fixed thresholds, when combined with a preceding ReLU MLP, can approximate the original BMRU arbitrarily well (Proposition~\ref{prop:threshold}). Together they show that the FQ BMRU is a universal approximator of the original BMRU function class.

\begin{definition}[Original BMRU]
The original BMRU cell computes:
\begin{align*}
    \hat{h}_t &= W_x x_t + b_x, \\
    \beta_t &= \left|W_\beta x_t + b_\beta\right|, \\
    z_t &= \mathcal{H}\bigl(|\hat{h}_t| - \beta_t\bigr), \\
    h_t &= z_t \odot S(\hat{h}_t) \odot \alpha + (1 - z_t) \odot h_{t-1},
\end{align*}
where $x_t \in \mathbb{R}^n$, $W_x, W_\beta \in \mathbb{R}^{d \times n}$, $b_x, b_\beta, \alpha \in \mathbb{R}^d$, $\mathcal{H}(\cdot)$ is the element-wise Heaviside step function, and $\odot$ denotes Hadamard product.
\end{definition}

\begin{definition}[FQ BMRU]
The FQ BMRU cell computes:
\begin{align*}
    \hat{h}_t &= \text{ReLU}\left(\tilde{W}_x x_t + \tilde{b}_x\right), \\
    z_{\mathrm{lo},t} &= \mathcal{H}\bigl(\beta_{\mathrm{lo}} - \hat{h}_t\bigr), \\
    z_{\mathrm{hi},t} &= \mathcal{H}\bigl(\hat{h}_t - \beta_{\mathrm{hi}}\bigr), \\
    h_t &= z_{\mathrm{hi},t} \odot \alpha + (1 - z_{\mathrm{lo},t}) \odot (1 - z_{\mathrm{hi},t}) \odot h_{t-1},
\end{align*}
where $x_t \in \mathbb{R}^n$, $\tilde{W}_x \in \mathbb{R}^{d \times n}$, $\tilde{b}_x, \beta_{\mathrm{lo}}, \beta_{\mathrm{hi}}, \alpha \in \mathbb{R}^d$, and $\beta_{\mathrm{hi}} > \beta_{\mathrm{lo}}$ element-wise.
\end{definition}

\subsection{Output range equivalence}

\begin{proposition}\label{prop:output}
Let $f$ be a function computed by a BMRU cell producing bipolar outputs $h^{(\pm)} \in \{-\alpha, +\alpha\}^d$, followed by a linear layer $y = W h^{(\pm)} + b$ with $W \in \mathbb{R}^{m \times d}$ and $b \in \mathbb{R}^m$. Then the same function is computed by a BMRU cell with unipolar outputs $h^{(+)} \in \{0, \alpha\}^d$ followed by $y = \tilde{W} h^{(+)} + \tilde{b}$, where $\tilde{W} = 2W$ and $\tilde{b} = b - W \alpha$.
\end{proposition}

\begin{proof}
Define the mapping $h^{(+)}_i = \frac{1}{2}(h^{(\pm)}_i + \alpha_i)$ for each dimension $i \in \{1, \ldots, d\}$. This is a bijection between $\{-\alpha_i, +\alpha_i\}$ and $\{0, \alpha_i\}$:
\begin{itemize}
    \item $h^{(\pm)}_i = +\alpha_i \implies h^{(+)}_i = \frac{1}{2}(\alpha_i + \alpha_i) = \alpha_i$,
    \item $h^{(\pm)}_i = -\alpha_i \implies h^{(+)}_i = \frac{1}{2}(-\alpha_i + \alpha_i) = 0$.
\end{itemize}

The inverse mapping is $h^{(\pm)} = 2 h^{(+)} - \alpha$. Substituting into the linear layer:
\begin{align*}
    y = W h^{(\pm)} + b &= W (2 h^{(+)} - \alpha) + b \\
    &= 2W h^{(+)} + (b - W\alpha).
\end{align*}

Thus, setting $\tilde{W} = 2W$ and $\tilde{b} = b - W\alpha$ yields $y = \tilde{W} h^{(+)} + \tilde{b}$. Since this parameter transformation is bijective, the two representations span identical function spaces.
\end{proof}

\subsection{Fixed threshold approximation}

\begin{proposition}[Fixed threshold approximation]\label{prop:threshold}
Let $\mathcal{X}\subset\mathbb{R}^n$ be a compact set. For any original BMRU cell and any $\epsilon>0$, there exists a ReLU MLP $\Phi:\mathbb{R}^n\to\mathbb{R}^d$ and fixed thresholds $\beta_{\mathrm{lo}},\beta_{\mathrm{hi}}\in\mathbb{R}^d$ with $\beta_{\mathrm{hi}}>\beta_{\mathrm{lo}}$ such that the FQ BMRU cell, with $\hat{h}_t = \Phi(x_t)$ (absorbing the cell's own ReLU layer into the MLP), reproduces the gating behaviour of the original cell on $\mathcal{X}$ except on a set of measure at most $\epsilon$. The resulting state update can therefore be approximated arbitrarily closely.
\end{proposition}

\begin{proof}
The original BMRU gate for dimension $i$ compares $\hat{h}_i(x)=[W_x]_i x+[b_x]_i$ against an input-dependent threshold $\beta_i(x)=|[W_\beta]_i x+[b_\beta]_i|$. Because $\beta_i(x)$ is piecewise-linear, an MLP before the cell can directly absorb this computation: a ReLU MLP can represent any continuous piecewise-linear function on a compact set, so it can learn to produce features that implicitly encode the needed decision functions.

A second, geometric difference is that the original hold region $\{x:|\hat{h}_i(x)|<\beta_i(x)\}$ is symmetric around zero after thresholding, while the FQ BMRU hold region $\{x:\beta_{\mathrm{lo}}<\hat{h}_i<\beta_{\mathrm{hi}}\}$ is not necessarily centered. However, this asymmetry can be compensated by a simple affine transformation inside the MLP: let $\mu_i=(\beta_{\mathrm{hi}}+\beta_{\mathrm{lo}})/2$ and $\sigma_i=(\beta_{\mathrm{hi}}-\beta_{\mathrm{lo}})/2$. Then shifting the MLP output by $\mu_i$ and scaling by $1/\sigma_i$ maps the FQ hold region to a symmetric interval $(-1,1)$. In other words, the preceding MLP can learn to offset the input so that the effective hysteretic window becomes centered, with half-width $\sigma_i$ playing exactly the role of the original $\beta_i$.

Consequently, an MLP can first project the input $x$ into a space where the required piecewise-linear boundaries are encoded, and then shift and scale the projections so that fixed thresholds $\beta_{\mathrm{lo}},\beta_{\mathrm{hi}}$ separate the ``set high'', ``hold'' and ``set low'' regions in the same way as the original boundaries, up to an arbitrarily narrow gap where the original touching regions would meet. Because ReLU MLPs can approximate the signed distance to these boundaries, the measure of the disagreement region can be made smaller than any given $\epsilon$, and the resulting gate outputs (and thus the hidden state update) approximate the original BMRU uniformly on $\mathcal{X}$ except on that small set.
\end{proof}

\subsection{Combined result}

\begin{theorem}[FQ BMRU Universal Approximation]\label{thm:main}
The FQ BMRU with fixed thresholds and unipolar outputs $\{0, \alpha\}$, when preceded by a ReLU MLP and followed by linear projection layers, can approximate any network based on the original BMRU (with input-dependent thresholds and bipolar outputs) arbitrarily closely on compact subsets of the input space.
\end{theorem}

\begin{proof}
Given an original BMRU network, first apply Proposition~\ref{prop:threshold} to replace each original cell with a fixed-threshold FQ cell plus a preceding ReLU MLP; the approximation error in the gating can be made arbitrarily small. Next, apply Proposition~\ref{prop:output} to convert the bipolar outputs of the original network to unipolar outputs via an affine reparameterization of the subsequent linear layers; this conversion is exact and bijective. Composing the two steps yields an FQ BMRU network that approximates the original network's function uniformly on compacta with arbitrarily small error.
\end{proof}

\section{Training and software inference details}\label{sec:softwareKWS}
This appendix details the implementation for training and evaluating the BMRU and FQ BMRU in software, including benchmark task descriptions, dataset characteristics, and evaluation protocols.

\subsection{Benchmark task descriptions}\label{sec:taskdetails}
\subsubsection{Sequential MNIST (sMNIST)}
Sequential MNIST presents handwritten digit images ($28 \times 28$ pixels) as sequences of 784 grayscale pixel values in raster scan order (left-to-right, top-to-bottom). The task requires classifying the digit (0--9) based solely on the sequential pixel stream.

\paragraph{Task structure.} Each input sequence $\{x_t\}_{t=1}^{784}$ consists of normalized pixel intensities $x_t \in [0,1]$. The model processes the sequence and produces a classification decision at the final timestep. Ground truth labels $y \in \{0,1,\ldots,9\}$ provide supervision at all steps.

\paragraph{Dataset.} We use the standard MNIST dataset~\cite{lecun1998gradient}, comprising 70,000 images partitioned into approximately \SI{80}{\percent}/\SI{10}{\percent}/\SI{10}{\percent} training/validation/test splits.

\paragraph{Evaluation metric.} Classification accuracy on the test set. Random-guess baseline is \SI{10}{\percent}.

\subsubsection{Permuted MNIST (pMNIST)}
Permuted MNIST applies a fixed random permutation to the pixel ordering of MNIST images, destroying spatial correlations while preserving information content.

\paragraph{Task structure.} Given a fixed permutation $\pi: \{1,\ldots,784\} \to \{1,\ldots,784\}$, each image is presented as the sequence $\{x_{\pi(t)}\}_{t=1}^{784}$. The permutation is identical across all samples and remains fixed during training and evaluation.

\paragraph{Dataset and evaluation.} Identical to sMNIST.





\subsubsection{ListOps}

ListOps evaluates compositional reasoning over hierarchically structured sequences of operations \cite{tay2021long}.

\paragraph{Task structure.} Sequences represent nested operations (MAX, MIN, MEDIAN, SUM\_MOD) applied to integer arguments. Sequences are linearized in prefix notation with special tokens marking brackets and operations. Each token is mapped to a learned embedding, and the model processes sequences with mean pooling to produce the final numerical result.

\paragraph{Dataset.} We use the standard ListOps dataset \cite{nangia2018listops} with sequences up to length 2,000. The training set consists of 96,000 samples, while both the validation and test sets contain 2,000 samples each.

\paragraph{Evaluation metric.} Classification accuracy on the discrete output values (10 possible output classes).

\subsubsection{Keyword spotting: digits (dKWS)}
This task evaluates spoken digit recognition using the Google Speech Commands dataset (v0.02)~\cite{warden2018speech}, containing 105,829 one-second audio clips of 35 spoken words recorded at \SI{16}{kHz}.

\paragraph{Task structure.} 11-way classification distinguishing spoken digits (``zero'' through ``nine'') and background noise extracted from ambient recordings in the dataset.

\paragraph{Feature extraction.} Raw audio is processed using \texttt{librosa}~\cite{mcfee2015librosa} to extract 13-dimensional Mel-frequency cepstral coefficients (MFCCs) with the following parameters: sample rate \SI{16}{kHz}, FFT window size 512 samples, hop length 160 samples (\SI{10}{ms}), yielding 101 frames per one-second clip. Features are normalized per coefficient to zero mean and unit variance. MFCC extraction is assumed to occur off-chip; only the extracted features are processed by the analog network.

\paragraph{Dataset.} We use the official dataset partitioning, yielding approximately \SI{80}{\percent}/\SI{10}{\percent}/\SI{10}{\percent} training/validation/test splits with no speaker overlap between partitions. Negative examples (background noise) are sampled to match positive examples, ensuring balanced training.

\paragraph{Evaluation metric.} Classification accuracy on the test set. Random-guess baseline is \SI{9.09}{\percent}.

\subsubsection{Shakespeare Language Modeling}
Character-level language modeling on the complete works of William Shakespeare evaluates a model ability to capture long-range linguistic dependencies in English text.

\paragraph{Task structure.}
Given a sequence of characters $\{c_1, c_2, \dots, c_T\}$ drawn from the text, the model receives the embedding of $c_t$ at timestep $t$ and must predict a probability distribution over the next character $c_{t+1}$. The vocabulary consists of $V = 65$ distinct characters (lowercase and uppercase letters, punctuation, and whitespace). The loss is the mean cross-entropy across all prediction steps.

\paragraph{Dataset.}
We use the complete works of Shakespeare, a standard benchmark for character-level language modeling~\cite{karpathy2015unreasonable}. The corpus contains approximately 5 million characters, which we split by plays into training (\SI{80}{\percent}), validation (\SI{10}{\percent}), and test (\SI{10}{\percent}) sets, preserving document boundaries.

\paragraph{Evaluation metric.}
Mean cross-entropy loss on the test set. For comparability with prior work, bits-per-character (BPC) can be obtained as $\text{BPC} = \mathcal{L} / \ln 2$, where $\mathcal{L}$ is the average cross-entropy. A uniform random predictor over the 65-character vocabulary yields $\text{BPC} = \log_2(65) \approx 6.02$ (or equivalently a cross-entropy loss of $\ln 65 \approx 4.17$).

\subsubsection{Keyword spotting: yes versus others (``yes'' KWS)}
This binary classification task represents a wake-word detection scenario, distinguishing a target keyword from confusable alternatives and background noise.

\paragraph{Task structure.} Binary classification detecting the word ``yes'' against negative words: ``no'', ``up'', ``down'', ``left'', ``right'', and background noise.

\paragraph{Feature extraction and dataset.} Identical to dKWS. Negative examples are sampled equally from each negative category to match the number of positive examples.

\paragraph{Evaluation metric.} Classification accuracy on the test set, computed over 100 random pairings of all 397 positive samples with 397 randomly selected negative samples across 5 random seeds, reporting mean and min-max range. Random-guess baseline is \SI{50}{\percent}. Table~\ref{tab:accuracy} presents classification accuracy as a function of state dimension $d$ for the FQ BMRU trained in software. 

\begin{table}[t!]
\caption{%
  Binary ``yes'' KWS accuracy (Google Speech Commands dataset).
  Results averaged over 5 training seeds; brackets show [min; max] across seeds.
  Each seed is evaluated on 100 balanced test sets (397 ``yes'' samples paired with 397 random negatives), with accuracy min-max ranges capturing variability across these pairings.
}
\label{tab:accuracy}
\centering
\small
\setlength{\tabcolsep}{3pt}
\begin{tabular}{lccccccc}
\toprule
State dimension & 4 & 8 & 16 & 32 & 64 & 128 & 256 \\
\midrule
Accuracy (\%) &
\shortstack{93.86 \\ {\scriptsize [91.47; 94.98]}} &
\shortstack{96.98 \\ {\scriptsize [96.00; 98.19]}} &
\shortstack{98.08 \\ {\scriptsize [97.85; 98.44]}} &
\shortstack{98.16 \\ {\scriptsize [97.71; 98.69]}} &
\shortstack{97.51 \\ {\scriptsize [96.97; 97.93]}} &
\shortstack{97.65 \\ {\scriptsize [97.48; 97.81]}} &
\shortstack{97.78 \\ {\scriptsize [97.48; 98.10]}}  \\
\bottomrule
\end{tabular}
\end{table}

\subsection{Training procedure}\label{sec:training}
We use standard training practices throughout; no architectural modifications or hyperparameter tuning specific to the first-quadrant formulation are required.

\subsubsection{Additional recurrent cell implementations}
This section provides detailed specifications of the recurrent cells evaluated in our experiments. For the BMRU and FQ BMRU architectures, we refer to Equations~\eqref{eq:bmru_candidate}--\eqref{eq:fq_update} in the main text.

\paragraph{Linear Recurrent Unit (LRU).} 
The LRU \cite{orvieto2023resurrecting} implements a diagonal complex-valued state-space model with strictly decaying exponentially parameterized eigenvalues for numerical stability. The state update follows:
\begin{align}
    x_t &= \Lambda \odot x_{t-1} + B u_t\,, \label{eq:LRU1}\\
    y_t &= \mathrm{Re}(C x_t) + D u_t\,,\label{eq:LRU2}
\end{align}
where $x_t \in \mathbb{C}^d$ is the complex-valued hidden state, $u_t \in \mathbb{R}^m$ is the input, $y_t \in \mathbb{R}^d$ is the output, and $\Lambda \in \mathbb{C}^d$ is a diagonal matrix of eigenvalues. The eigenvalues are parameterized as:
\begin{equation}
    \Lambda = \exp(-\exp(\nu) + i\exp(\theta))\,,
\end{equation}
where $\nu, \theta \in \mathbb{R}^d$ are learned parameters initialized such that eigenvalue magnitudes lie in $[r_{\min}, r_{\max}]$ with uniformly distributed phases in $[0, 2\pi]$. The input matrix $B \in \mathbb{C}^{d \times m}$ is normalized by $\gamma = \sqrt{1 - |\Lambda|^2}$ to maintain consistent input scaling across different eigenvalue magnitudes. The output matrices $C \in \mathbb{C}^{d \times d}$ and $D \in \mathbb{R}^{d \times m}$ are learned parameters.

The LRU is fully parallelizable via associative scan over the linear recurrence.

\paragraph{Minimal Gated Recurrent Unit (minGRU).}

The minGRU \cite{feng2024minGRU} simplifies the traditional GRU by removing hidden state dependencies from the gating mechanism, enabling full parallelization while maintaining competitive performance. The update rule is:
\begin{align}
    z_t &= \sigma(W_z x_t + b_z)\,, \label{eq:minGRU1}\\
    \tilde{h}_t &= W_h x_t + b_h\,, \label{eq:minGRU2}\\
    h_t &= (1 - z_t) \odot h_{t-1} + z_t \odot \tilde{h}_t\,,\label{eq:minGRU3}
\end{align}
where $h_t \in \mathbb{R}^d$ is the hidden state, $x_t \in \mathbb{R}^m$ is the input, $z_t \in [0,1]^d$ is the update gate, $\tilde{h}_t \in \mathbb{R}^d$ is the candidate hidden state, and $\sigma(\cdot)$ denotes the sigmoid function.

A critical architectural property is that the sigmoid gate ensures $z_t \in (0, 1)$ (open interval), meaning the values $z_t = 0$ and $z_t = 1$ are only approached asymptotically but never attained.

The update equation is a linear recurrence in $h_t$ and can be parallelized via associative scan.

\subsubsection{Backbone architecture for software inference (Table~\ref{tab:benchmarks})}\label{sec:backbone_software}

The benchmark experiments in Table~\ref{tab:benchmarks} use the following parallelizable RNN backbone. Only the recurrent cell varies across rows of Table~\ref{tab:benchmarks}; the encoder, blocks, decoder, and training procedure are shared. The model dimension is fixed at $m = 256$. Classification rows use $r = 2$ stacked recurrent blocks with state dimension $d = 64$; the Shakespeare row uses depth $r = 6$ and $d = m = 256$.

\paragraph{Input encoder.} Raw task-specific inputs $\hat{x}_t \in \mathbb{R}^{d_{\text{task}}}$ are first projected to the model dimension and then refined by a residual MLP:
\begin{equation}
    x_t = \tilde{x}_t + \MLP{\tilde{x}_t}, \qquad \tilde{x}_t = \Linear{\hat{x}_t} \in \mathbb{R}^m.
\end{equation}

\paragraph{Positional encoding.} For all tasks of Table~\ref{tab:benchmarks}, a 32-dimensional sinusoidal Positional Encoding (PE)~\cite{zheng2021rethinking} is concatenated to $x_t$ at each timestep and projected back to $\mathbb{R}^m$ through an additional $\Linear{\cdot}$.

\paragraph{Block structure.} The backbone consists of $r$ stacked blocks. Each block alternates a recurrent sub-layer and a point-wise MLP sub-layer, both wrapped in a pre-norm residual scheme:
\begin{equation}
    y = \upsilon \odot x + \SubLayer{\Norm{x}},
\end{equation}
where $\upsilon \in \mathbb{R}^m$ is a learnable scaling vector initialized to one, $\Norm{\cdot}$ denotes LayerNorm~\cite{ba2016layernormalization}, and $\SubLayer{\cdot}$ is either the recurrent sub-layer or the MLP sub-layer.

\paragraph{Recurrent sub-layer.} The recurrent sub-layer maps the input sequence to a hidden state sequence $h_{1:T} = \Cell{x_{1:T}; h_0}$ with $h_t \in \mathbb{R}^{d}$, then projects it back to $\mathbb{R}^m$ through an input-gated normalized projection:
\begin{equation}
    \CellB{x_{1:T}} = \Norm{\Linear{\Cell{x_{1:T}; h_0}}} \odot \sigma\!\left(\Linear{x_{1:T}}\right),
\end{equation}
where $\sigma(\cdot)$ is the element-wise sigmoid. The inner normalization keeps the magnitude of the recurrent output bounded even for cells with unitary eigenvalues, and the input-dependent gate decouples information storage from information consumption. The recurrent cell varies across rows of Table~\ref{tab:benchmarks}: BMRU and FQ BMRU follow Equations~\eqref{eq:bmru_candidate}--\eqref{eq:fq_update}; LRU follows Equations~\eqref{eq:LRU1}--\eqref{eq:LRU2} with diagonal complex states, exponentially parameterized eigenvalues initialized in $[r_{\min}, r_{\max}] = [0.9, 0.999]$ with phases in $[0, 2\pi]$, and input matrix normalized by $\gamma = \sqrt{1 - |\Lambda|^2}$; minGRU follows Equations~\eqref{eq:minGRU1}--\eqref{eq:minGRU3} with a single input-dependent sigmoid gate and no reset gate. All four cells admit parallel evaluation through associative scans~\cite{martin2017parallelizing, blelloch2002scans}.

\paragraph{MLP sub-layer.} The MLP sub-layer is point-wise, with input and output dimensions $m$, hidden dimension $4m$, and GLU activation~\cite{dauphin2017languagemodelinggatedconvolutional}:
\begin{equation}
    \MLP{x} = \Linear{\Dropout{\GLU{\Linear{x}}}}.
\end{equation}

\paragraph{Sequence pooling.} For classification tasks, the cross-entropy loss is averaged over all timesteps during training, and predictions at inference use majority voting across timesteps (except for ListOps). Shakespeare uses per-token cross-entropy with no pooling.

\paragraph{Output decoder.} The pooled representation $y_{\text{pool}} \in \mathbb{R}^m$ is transformed back to the task output dimension by a residual MLP symmetric to the input encoder:
\begin{equation}
    \hat{y} = \tilde{y} + \MLP{\tilde{y}}, \qquad \tilde{y} = \Linear{y_{\text{pool}}}.
\end{equation}

\subsubsection{Backbone architecture for hardware inference (Figures~\ref{fig:fig2} and \ref{fig:noise_analysis})}\label{sec:backbone_hardware}

The proof-of-concept network targeted by transistor-level simulation uses a deliberately simplified backbone, illustrated in Figure~\ref{fig:fig2}A. Components incompatible with current-mode subthreshold CMOS (LayerNorm, GLU, sigmoid gates, residual encoders) are removed, and every remaining operation maps directly to a circuit primitive available in the analog library: current mirrors, diode-connected transistors, and the FQ BMRU cell of Section~\ref{sec:circuit}. The configuration used in Sections~\ref{sec:kws_proof_of_concept} and ~\ref{sec:robustness_power} (as well as all subsequent appendices) is $N = 2$ stacked FQ BMRU layers with state dimension $d = 4$.

\paragraph{Input encoder.} Raw 13-dimensional MFCC features $\hat{x}_t$ are projected to the state dimension by a single linear layer:
\begin{equation}
    x_t = \Linear{\hat{x}_t} \in \mathbb{R}^d.
\end{equation}
No residual MLP and no normalization are used. The projection is realized in hardware by the current-mirror FC layer of Appendix~\ref{sec:FClayers}.

\paragraph{Recurrent layers.} The backbone contains $N = 2$ stacked FQ BMRU layers separated by an inter-layer FC transformation followed by ReLU. Each FQ BMRU cell follows Equations~\eqref{eq:fq_candidate}--\eqref{eq:fq_update}, with parameters $\alpha$, $\beta_\text{lo}$, and $\beta_\text{hi}$ realized as bias currents (Appendix~\ref{sec:CMOSBMRU}). Skip connections~\cite{he2016deep} bypass each FQ BMRU layer and are realized in the current domain by routing positive and negative branches separately (Appendix~\ref{sec:skip}).

\paragraph{Activation domain.} All signals throughout the backbone are non-negative currents. The FQ reformulation of Section~\ref{sec:FQ_BMRU} ensures unipolar BMRU outputs; FC layers use ReLU activation realized as a diode-connected PMOS (Appendix~\ref{sec:FClayers}); skip connections preserve the unipolar constraint by separating positive and negative branches.

\paragraph{Positional encoding.} No positional encoding is used. Keyword detection on a fixed 1-second window is position-invariant, and adding sinusoidal PE would require time-dependent reference currents incompatible with the static-bias design philosophy and is not necessary.

\paragraph{Output decoder.} The final FQ BMRU layer feeds an output FC layer producing class currents (two for binary ``yes'' classification, eleven for the multi-class extension of Appendix~\ref{sec:multiclass_kws}). Classification is read out directly from the output currents.

\paragraph{Sequence pooling.} Classification on a 1-second clip uses majority voting across the 101 timesteps. At each timestep the highest output current votes for its class, and the class with the most votes over the sequence is the prediction. This avoids any temporal aggregation circuitry that would itself need to be implemented in analog.

\subsubsection{Recurrent cell configuration}
Each (FQ) BMRU cell uses the following configuration:
\begin{itemize}
    \item Dropout rate of 0.1 during training~\cite{srivastava2014dropout};
    \item All parameters initialized using a default uniform distribution, constrained to positive values for $\alpha$, $\beta_\text{lo}$, and $\delta$ ($= \beta_\text{hi} - \beta_\text{lo}$);
    \item Initial state sampled uniformly in $[0, 1]$ and binarized via surrogate Heaviside during training; zero initialization at inference for the FQ BMRU in hardware.
\end{itemize}

\subsubsection{Training generalities}
All experiments follow a unified training procedure with hyperparameters that are not optimized per task. The objective is to observe the robustness of the cells to standard hyperparameter values, while acknowledging that task-specific optimization could improve absolute performance metrics.

\paragraph{Optimizer and schedule.} Optimization is performed with AdamW~\cite{loshchilov2019decoupled} using the following configuration:
\begin{itemize}
    \item Initial learning rate: $10^{-3}$;
    \item Weight decay: $10^{-4}$;
    \item Cosine decay schedule~\cite{loshchilov2017sgdr} with linear warm-up over \SI{1}{\percent} of training steps;
    \item Gradient clipping: global norm capped at $1.0$.
\end{itemize}

\paragraph{Training budget.} A batch size of 64 is used for all experiments. The maximum number of training iterations varies per task to ensure convergence within reasonable computational limits: 70,000 for MNIST tasks, 100,000 for the language modeling and ListOps tasks, and 35,000 for KWS tasks. Models are evaluated every 64 iterations on 20 batches sampled from the validation set, with the best-performing checkpoint saved for final evaluation. All experiments are repeated across 5 random seeds (1, 2, 3, 4, 5).

\paragraph{Loss function.} Cross-entropy loss is computed over all timestep predictions for all tasks:
\begin{equation}
    \mathcal{L} = -\frac{1}{T}\sum_{t=1}^{T} \sum_{c} y_{c} \log\!\left(\hat{y}_{t,c}\right),
\end{equation}
where $y_{c}$ is the one-hot target (target depends on $t$ for Shakespeare), $\hat{y}_{t,c}$ is the softmax output at timestep $t$ for class $c$, and $T$ denotes the sequence length.

\paragraph{Evaluation protocol.} For sMNIST, pMNIST, ListOps and digit KWS tasks, test accuracy is reported as mean across 5 training seeds. For the Shakespeare task, mean cross-entropy loss is reported. For the ``yes'' KWS task, accuracy is computed over 100 random pairings of all 397 positive samples with 397 randomly selected negative samples across 5 seeds, reporting mean and min-max ranges to account for variability in negative sample selection.

\subsubsection{Gradient estimation}
\paragraph{Surrogate gradients.} The Heaviside nonlinearity in the (FQ) BMRU cell is non-differentiable. We employ surrogate gradients~\cite{neftci2019surrogate} with the approximation:
\begin{equation}
    \frac{\mathrm{d}\mathcal{H}}{\mathrm{d}x} \stackrel{\text{backward}}{\approx} \frac{1}{1+(\pi x)^2}.
\end{equation}

\paragraph{Decaying cumulative update for improved gradient flow.}

For the (FQ) BMRU, during training the state update equation is augmented as:
\begin{equation}
    h_t = f_\theta(x_t, h_{t-1}) + \varepsilon \odot h_{t-1},
\end{equation}
where $f_\theta(x_t, h_{t-1})$ corresponds to Equation~\eqref{eq:bmru_update} for BMRU and Equation~\eqref{eq:fq_update} for FQ BMRU. Training with $\varepsilon = 1$ provides superior gradient flow and convergence stability~\cite{cBMRUref}, while $\varepsilon = 0$ recovers the original cell dynamics that map directly onto Schmitt trigger circuits for ultra-low power deployment. We adopt an intermediate approach: training initially with $\varepsilon = 1$, then annealing $\varepsilon$ toward zero to produce models compatible with the target hardware. Specifically, $\varepsilon$ remains at 1 for the first 5\% of training, decays linearly to 0 over the next 70\%, then stays at 0 for the final 25\%. Best-loss checkpointing is performed only after $\varepsilon$ reaches zero, ensuring that inference uses the original cell dynamics. This annealing procedure is specific to BMRU and FQ BMRU; the LRU and minGRU baselines are trained with their standard formulations and do not require it.

\subsubsection{Implementation}
Training is implemented in JAX~\cite{jax2018github}/Flax~\cite{flax2020github} using parallel associative scans for efficient (FQ) BMRU state computation across the sequence dimension.

\subsection{Post-training quantization}\label{sec:quantization}
To assess hardware deployment feasibility, we evaluate post-training quantization without retraining. All network parameters (weights, biases, thresholds $\beta_\text{lo}$, widths $\delta$, and amplitudes $\alpha$) are uniformly quantized to $n$ bits by mapping to $2^n$ discrete levels within each parameter dynamic range~\cite{gray1998quantization}:
\begin{equation}
\begin{aligned}
    w_q
    &= \text{round}\!\left(
        \frac{w - w_{\min}}{w_{\max} - w_{\min}}
        \cdot (2^n - 1)
    \right) \\
    &\quad \cdot \frac{w_{\max} - w_{\min}}{2^n - 1}
    + w_{\min}
\end{aligned}
\end{equation}

Table~\ref{tab:quantization} summarizes accuracy under different quantization levels across five training seeds. Four-bit quantization incurs minimal degradation (typically $<$\SI{4}{\percent}), supporting the programmable hardware approach described in Section~\ref{sec:perspandlimit} of the main paper. Two-bit quantization causes significant accuracy loss, indicating that at least 4 bits are required for this application. Interestingly, 6-bit and 8-bit quantization occasionally yield slight accuracy improvements over full precision, likely due to implicit regularization effects.

\begin{table*}[ht!]
\centering
\caption{Post-training quantization results. Accuracy (\%) on the ``yes'' KWS task for various state dimensions and bit widths, averaged over 5 training seeds. FP32 stands for 32-bit Floating Point.}
\label{tab:quantization}
\begin{tabular}{cccccc}
\toprule
State dim $d$ & 2-bit & 4-bit & 6-bit & 8-bit & FP32 \\
\midrule
4   & 51.20 & 89.31  & 92.75  & 93.16 & 93.86 \\
8   & 67.56  & 93.29  & 96.65  & 96.67  & 96.98 \\
16  & 61.04  & 97.48  & 97.75  & 97.60  & 98.08 \\
32  & 63.77  & 98.25  & 98.23  & 98.25  & 98.16 \\
64  & 59.69  & 97.17  & 97.27  & 97.38  & 97.51 \\
\bottomrule
\end{tabular}
\end{table*}

The consistent robustness to 4-bit quantization across all state dimensions validates the co-design approach: networks can be trained in full precision and subsequently quantized for hardware implementation using binary-weighted current mirrors with 4-bit programmability. Further accuracy gains could be achieved through fine-tuning the quantized network or by adopting quantization-aware training techniques.

\subsection{Compute resources} 
All training experiments were performed on an academic HPC cluster, using a single NVIDIA Tesla A100 GPU (40 GB VRAM) per run on a node with AMD EPYC Milan CPUs and up to 512 GB of system RAM. Training a single seed takes approximately a few hours for all tasks; software inference completes in minutes. The total compute footprint of the project, including preliminary experiments not reported in this paper, remains modest by modern deep learning standards.

\section{Circuit schematics}\label{sec:circuits}
This supplementary section provides detailed circuit schematics for the CMOS implementation of the FQ BMRU described in the main paper as well as for the FC layers.

\paragraph{Technology and operating point.} All circuits are implemented in a \SI{180}{\nano\meter} CMOS X-FAB technology in Cadence Virtuoso with a supply voltage of \SI{1.8}{V}, operating in the subthreshold regime. All simulations are conducted using Spectre with PRIMLIB \SI{180}{\nano\meter} models (ne/pe devices). A value of ``1'' in the software model maps to a current of \SI{1}{nA} in hardware, establishing the correspondence between digital and analog representations. Transistor dimensions are sized to a few square micrometers to mitigate Early voltage effects and minimize process-induced mismatch. The topology used in this work uses transistor stacking, motivating the implementation of this design in the \SI{180}{\nano\meter} CMOS technology from X-FAB as it allows for a larger voltage headroom.

All transistor-level simulations were run on a workstation equipped with an AMD Ryzen Threadripper 3970X (32 cores, 3.7 GHz) and 256 GB of RAM, inside a virtual machine. A single transient inference simulation completes in under one minute; Monte Carlo (200 samples) and PVT corner analyses scale linearly with the number of samples and corners.

\subsection{Current mirror weight implementation}\label{sec:currentmirror}
Figure~\ref{fig:current_mirror} illustrates the fundamental current mirror topology used to implement weighted connections. In subthreshold operation, a diode-connected input transistor converts an input current $I_x$ into a gate voltage $V_x = V_y$, which is shared with the output transistor. Since both transistors operate at identical gate-source voltages, their drain currents are primarily determined by their width ratio:
\begin{equation}
    I_y \approx \frac{W_\text{out}}{W_\text{in}} I_x.
\end{equation}

In practice, the relationship deviates from ideal scaling due to drain-source voltage dependence and other second-order effects. To account for this, we perform exhaustive characterization sweeping $W_\text{out}$ and $I_x$ for fixed $W_\text{in}$, measuring the resulting $I_y$. This calibration produces a lookup table mapping desired output currents to optimal width selections, ensuring accurate weight implementation despite non-ideal transistor behavior. PMOS mirrors source current from the supply voltage, while NMOS mirrors sink current to ground---both polarities are required for implementing signed weights in FC layers.

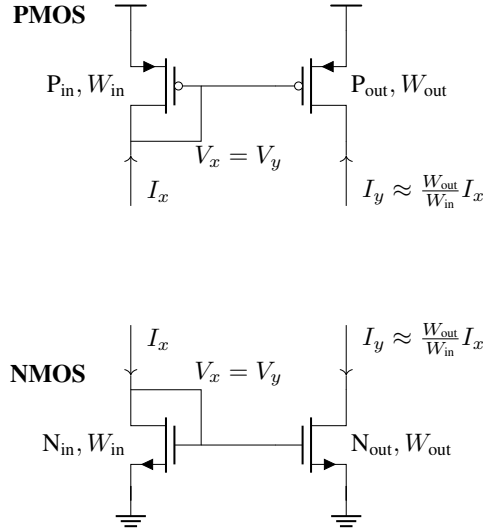
\begin{figure*}[ht!]
    \centering
    \scalebox{0.95}{\begin{circuitikz}[american, transform shape]
    \draw (0,6) node[pmos, xscale=-1] (pm1) {\ctikzflipx{$\text{P}_{\text{in}},  W_{\text{in}}$}};
    \draw (pm1.source) node[rground, yscale=-1] {};
    \draw (pm1.gate) |- (pm1.drain);
    
    \draw (3,6) node[pmos] (pm2) {$\text{P}_{\text{out}},  W_{\text{out}}$};
    \draw (pm2.source) node[rground, yscale=-1] {};
    \draw (pm2.gate) -- ++(-0.5,0) |- (pm1.gate);
    
    \node[anchor=south] at (1.5,4.7) {$V_x = V_y$};
    
    \draw[<-] ($(pm1.drain) + (0,-0.2)$) -- ++(0,-0.7);
    \draw (pm1.drain) -- ++(0,-0.2);
    \node[anchor=west] at ($(pm1.drain) + (0.1,-0.7)$) {$I_{x}$};
    \draw[<-] ($(pm2.drain) + (0,-0.2)$) -- ++(0,-0.7);
    \draw (pm2.drain) -- ++(0,-0.2);
    \node[anchor=west] at ($(pm2.drain) + (0.1,-0.7)$) {$I_{y} \approx \frac{W_{\text{out}}}{W_{\text{in}}} I_{x}$};
    
    \node[anchor=east, font=\bfseries] at (-0.5,7) {PMOS};
    
    \draw (0,1) node[nmos, xscale=-1] (nm1) {\ctikzflipx{$\text{N}_{\text{in}},  W_{\text{in}}$}};
    \draw (nm1.source) -- ++(0,0.2) node[ground] {};
    \draw (nm1.gate) |- (nm1.drain);
    
    \draw (3,1) node[nmos] (nm2) {$\text{N}_{\text{out}}, W_{\text{out}}$};
    \draw (nm2.source) -- ++(0,0.2) node[ground] {};
    \draw (nm2.gate) -- ++(-0.5,0) |- (nm1.gate);
    
    \node[anchor=south] at (1.5,1.7) {$V_x = V_y$};
    
    \draw[<-] ($(nm1.drain) + (0,0.2)$) -- ++(0,0.7);
    \node[anchor=west] at ($(nm1.drain) + (0.1,0.7)$) {$I_{x}$};
    \draw (nm1.drain) -- ++(0,0.2);
    \draw[<-] ($(nm2.drain) + (0,0.2)$) -- ++(0,0.7);
    \draw (nm2.drain) -- ++(0,0.2);
    \node[anchor=west] at ($(nm2.drain) + (0.1,0.7)$) {$I_{y} \approx \frac{W_{\text{out}}}{W_{\text{in}}} I_{x}$};
    
    \node[anchor=east, font=\bfseries] at (-0.5,2) {NMOS};
\end{circuitikz}}
    \caption{\textbf{Current mirror topologies for weight implementation.}
    Top: PMOS current mirror (sources current from supply).
    Bottom: NMOS current mirror (sinks current to ground).}
    \label{fig:current_mirror}
\end{figure*}

Programmable weight implementation is achieved using binary-quantized current mirrors controlled by a shift register, as illustrated for the PMOS case in Figure~\ref{fig:binary_current_mirror}.

\begin{figure*}[ht!]
    \centering
    \scalebox{0.95}{\begin{tikzpicture}[transform shape]
	\node[pmos, xscale=-1] at (4.375, 2.855){};
	\node[rground, yscale=-1] at (4.375, 3.625){};
	\node[pmos] at (7.125, 2.855){};
	\draw (5.355, 2.855) -- (6.145, 2.855);
	\draw (4.375, 2.085) -| (4.375, 1.375);
	\draw (5.625, 2.875) |- (4.375, 1.875);
	\node[rground, yscale=-1] at (7.125, 3.625){};
	\node[currarrow] at (4.625, 3.125){};
	\node[currarrow, xscale=-1] at (6.875, 3.125){};
	\node[currarrow, rotate=-90] at (4.375, 1.625){};
	\node[shape=rectangle, minimum width=1.715cm, minimum height=0.715cm] at (3.5, 2.875){} node[anchor=north, align=center, text width=1.327cm, inner sep=6pt] at (3.5, 3.25){$\text{P}_\mathrm{in}, \text{W}_\mathrm{in}$};
	\node[shape=rectangle, minimum width=1.34cm, minimum height=0.715cm] at (7.812, 2.875){} node[anchor=north, align=center, text width=0.952cm, inner sep=6pt] at (7.812, 3.25){$\frac{\text{W}_\mathrm{out}}{2}$};
	\node[shape=rectangle, minimum width=0.715cm, minimum height=0.715cm] at (4, 1.625){} node[anchor=north, align=center, text width=0.327cm, inner sep=6pt] at (4, 2){$I_\mathrm{in}$};
	\node[shape=rectangle, minimum width=12.84cm, minimum height=0.715cm] at (9.062, 4.375){} node[anchor=north, align=center, text width=12.452cm, inner sep=6pt] at (9.062, 4.75){\textbf{Binary-weighted PMOS current mirror}};
	\node[shape=rectangle, minimum width=0.715cm, minimum height=0.715cm] at (5.75, 3.25){} node[anchor=north, align=center, text width=0.327cm, inner sep=6pt] at (5.75, 3.625){$V_\mathrm{g}$};
	\node[nmos] at (7.125, 1.315){};
	\draw (7.125, 0.545) -| (7.125, -0.165);
	\node[currarrow, rotate=-90] at (7.125, 0.085){};
	\node[shape=rectangle, minimum width=1.735cm, minimum height=0.715cm] at (5.26, 1.315){} node[anchor=north, align=center, text width=2.6cm, inner sep=6pt] at (5.26, 1.69){$b_{1} = 0\, \text{or}\,1$};
	\node[pmos] at (10.23, 2.875){};
	\node[rground, yscale=-1] at (10.23, 3.645){};
	\node[currarrow, xscale=-1] at (9.98, 3.145){};
	\node[nmos] at (10.23, 1.335){};
	\draw (10.23, 0.565) -| (10.23, -0.145);
	\node[currarrow, rotate=-90] at (10.23, 0.105){};
	\node[shape=rectangle, minimum width=1.235cm, minimum height=0.715cm] at (10.865, 2.875){} node[anchor=north, align=center, text width=0.847cm, inner sep=6pt] at (10.865, 3.25){$\frac{\text{W}_\mathrm{out}}{4}$};
	\node[shape=rectangle, minimum width=0.715cm, minimum height=0.715cm] at (8.875, 3.25){} node[anchor=north, align=center, text width=0.327cm, inner sep=6pt] at (8.875, 3.625){$V_\mathrm{g}$};
	\node[shape=rectangle, minimum width=1.715cm, minimum height=0.715cm] at (8.375, 1.335){} node[anchor=north, align=center, text width=2.6cm, inner sep=6pt] at (8.375, 1.71){$b_{2} = 0\, \text{or}\,1$};
	\node[currarrow] at (6.834, 1.045){};
	\node[currarrow] at (9.961, 1.068){};
	\node[shape=rectangle, minimum width=2.985cm, minimum height=1.465cm] at (11.74, 2.065){} node[anchor=center, align=center, text width=2.597cm, inner sep=6pt] at (11.74, 2.065){\Huge ...};
	\node[pmos] at (13.974, 2.857){};
	\node[rground, yscale=-1] at (13.974, 3.627){};
	\node[currarrow, xscale=-1] at (13.724, 3.127){};
	\node[nmos] at (13.974, 1.317){};
	\draw (13.974, 0.547) -| (13.974, -0.163);
	\node[currarrow, rotate=-90] at (13.974, 0.087){};
	\node[shape=rectangle, minimum width=1.241cm, minimum height=0.715cm] at (14.612, 2.875){} node[anchor=north, align=center, text width=0.853cm, inner sep=6pt] at (14.612, 3.25){$\frac{\text{W}_\mathrm{out}}{2^B}$};
	\node[shape=rectangle, minimum width=0.715cm, minimum height=0.715cm] at (12.619, 3.232){} node[anchor=north, align=center, text width=0.327cm, inner sep=6pt] at (12.619, 3.607){$V_\mathrm{g}$};
	\node[shape=rectangle, minimum width=1.959cm, minimum height=0.715cm] at (11.997, 1.317){} node[anchor=north, align=center, text width=3cm, inner sep=6pt] at (11.997, 1.692){$b_{B} = 0\, \text{or}\,1$};
	\node[currarrow] at (13.705, 1.049){};
	\draw (7.125, -0.165) -| (10.23, -0.145);
	\draw (13.974, -0.163) |- (10.23, -0.165);
	\draw (11, -0.668) -- (11, -0.168);
	\node[shape=rectangle, minimum width=3.715cm, minimum height=0.715cm] at (11, -1.043){} node[anchor=north, align=center, text width=6.6cm, inner sep=6pt] at (11, -0.668){$I_\mathrm{out} \approx \sum_{i=1}^B b_i\frac{\text{W}_\mathrm{out}}{2^i\text{W}_\mathrm{in}}I_\mathrm{in}$};
	\node[currarrow, rotate=-90] at (11, -0.381){};
\end{tikzpicture}}
    \caption{
    \textbf{Binary-weighted PMOS current mirror for programmable weight implementation.}
    The effective output current is set by enabling combinations of binary-scaled mirror branches, allowing discrete (quantized) weight tuning via a shift register.
    }
    \label{fig:binary_current_mirror}
\end{figure*}
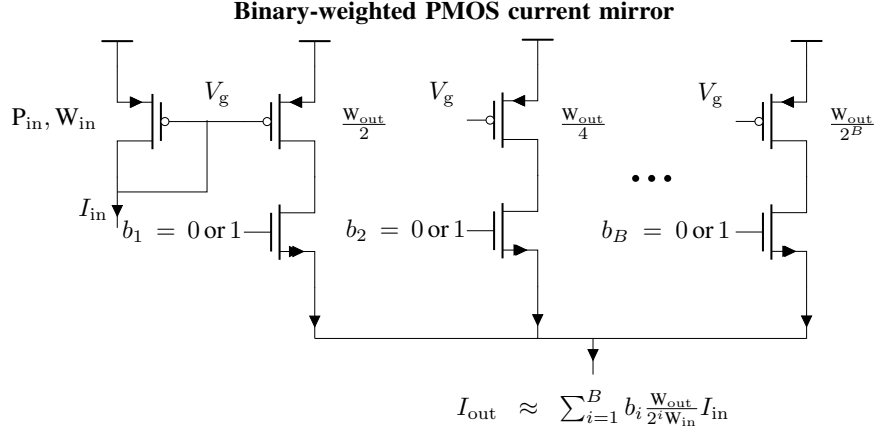

\subsection{FC layers}\label{sec:FClayers}
The FC layer computes 
\begin{equation}
    y_j = \text{ReLU}\left(\sum_{i=1}^d w_{ij} x_i + b_j\right) \quad \forall j\in [1, d].
\end{equation} 

This is realized entirely in the current domain using parallel current mirror banks (Figure~\ref{fig:FC_ReLU}). Each input voltage $V_i$ drives both PMOS and NMOS mirrors: PMOS transistors with width ratios $W_{i}^-/W_j$ implement negative weights by sourcing current to the upper summation node $\Sigma^-$, while NMOS transistors with ratios $W_{i}^+/W_j$ implement positive weights by sinking current from the lower node $\Sigma^+$. A vertical wire connects both nodes, where KCL computes the weighted sum. Bias currents $I_{b,j}^+$ and $I_{b,j}^-$ add the affine term. A diode-connected PMOS at the output conducts only when net current flows from the supply, implementing ReLU activation and providing output voltage $V_{\text{out},j}$ for subsequent stages.

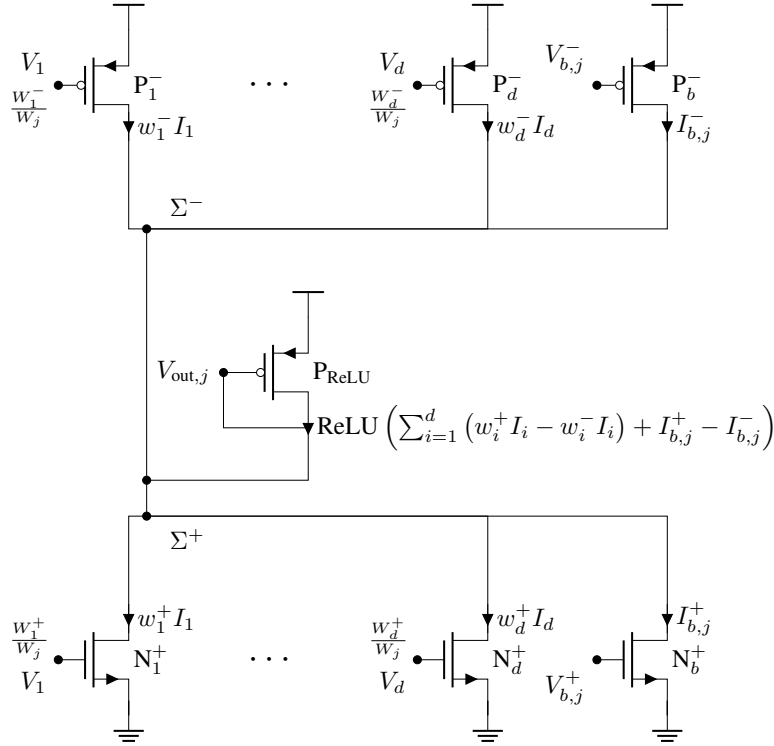
\begin{figure*}[ht!]
    \centering
    \scalebox{0.95}{\begin{circuitikz}[american, transform shape]

    \draw (0,6) node[pmos] (pn1) {$\text{P}_{1}^-$};
    \draw (pn1.source) node[rground, yscale=-1] {};
    \draw (pn1.gate) node[circ, label={above left:$V_{1}$}] {};
    
    \node at (2,6) {\Large $\cdots$};
    
    \draw (5,6) node[pmos] (pnd) {$\text{P}_{d}^-$};
    \draw (pnd.source) node[rground, yscale=-1] {};
    \draw (pnd.gate) node[circ, label={above left:$V_{d}$}] {};

    \draw (7.5,6) node[pmos] (pnb) {$\text{P}_{b}^-$};
    \draw (pnb.source) node[rground, yscale=-1] {};
    \draw (pnb.gate) node[circ, label={above left:$V_{b,j}^-$}] {};

    \draw (0,-2) node[nmos] (nn1) {$\text{N}_{1}^+$};
    \draw (nn1.source) -- ++(0,0.2) node[ground] {};
    \draw (nn1.gate) node[circ, label={below left:$V_{1}$}] {};
    
    \node at (2,-2) {\Large $\cdots$};
    
    \draw (5,-2) node[nmos] (nnd) {$\text{N}_{d}^+$};
    \draw (nnd.source) -- ++(0,0.2) node[ground] {};
    \draw (nnd.gate) node[circ, label={below left:$V_{d}$}] {};

    \draw (7.5,-2) node[nmos] (nnb) {$\text{N}_{b}^+$};
    \draw (nnb.source) -- ++(0,0.2) node[ground] {};
    \draw (nnb.gate) node[circ, label={below left:$V_{b,j}^+$}] {};

    \coordinate (sumtop) at (0.25,4);
    
    \draw (pn1.drain) |- (sumtop);
    \draw (pnd.drain) |- (sumtop);
    \draw (pnb.drain) |- (sumtop);

    \coordinate (sumbot) at (0.25,0);
    
    \draw (nn1.drain) |- (sumbot);
    \draw (nnd.drain) |- (sumbot);
    \draw (nnb.drain) |- (sumbot);

    \draw (sumtop) -- (sumbot);
    \draw (2.5,2) node[pmos] (prelu) {$\text{P}_\text{ReLU}$};
    \draw (prelu.source) node[rground, yscale=-1] {};
    \draw (prelu.gate) -- ++(-0.2,0) |- (prelu.drain);
    \draw (prelu.drain) |- (0.25,0.5);
    \node[circ] at (0.25,0.5) {};

    \node[circ] at (sumtop) {};
    \node[circ] at (sumbot) {};

    \draw (pn1.drain) to[short, i_<=$w_{1}^- I_{1}$] ++(0,0.4);
    \draw (pnd.drain) to[short, i_<=$w_{d}^- I_{d}$] ++(0,0.4);
    \draw (pnb.drain) to[short, i_<=$I_{b,j}^-$] ++(0,0.4);
    \draw (nn1.drain) to[short, i=$w_{1}^+ I_{1}$] ++(0,-0.4);
    \draw (nnd.drain) to[short, i=$w_{d}^+ I_{d}$] ++(0,-0.4);
    \draw (nnb.drain) to[short, i=$I_{b,j}^+$] ++(0,-0.4);
    \draw (prelu.drain) to[short, i_<=$\text{ReLU}\left(\sum_{i=1}^d\left(w_{i}^+ I_{i}-w_{i}^- I_{i}\right)+I_{b,j}^+ -I_{b,j}^-\right)$] ++(0,0);

    \draw (prelu.gate) ++(-0.2,0) node[circ, label={left:$V_{\text{out},j}$}] {};
    
    \node[above right=0.05 and 0.2 of sumtop] {$\Sigma^-$};
    \node[below right=0.05 and 0.2 of sumbot] {$\Sigma^+$};

    \node[anchor=east, font=\footnotesize] at ($(pn1.gate) + (0,-0.3)$) {$\frac{W_{1}^-}{W_j}$};
    \node[anchor=east, font=\footnotesize] at ($(pnd.gate) + (0,-0.3)$) {$\frac{W_{d}^-}{W_j}$};
    \node[anchor=east, font=\footnotesize] at ($(nn1.gate) + (0,0.3)$) {$\frac{W_{1}^+}{W_j}$};
    \node[anchor=east, font=\footnotesize] at ($(nnd.gate) + (0,0.3)$) {$\frac{W_{d}^+}{W_j}$};

\end{circuitikz}}
    \caption{\textbf{FC layer with ReLU activation.}
    PMOS mirrors (top) implement negative weights; NMOS mirrors (bottom) implement positive weights. The diode-connected PMOS harvests net positive current.}
    \label{fig:FC_ReLU}
\end{figure*}

For layers requiring anti-ReLU activation, the output transistor is replaced with a diode-connected NMOS that sinks current to ground when the net sum is negative (Figure~\ref{fig:FC_anti_ReLU}), implementing $y = \min(0, x)$.

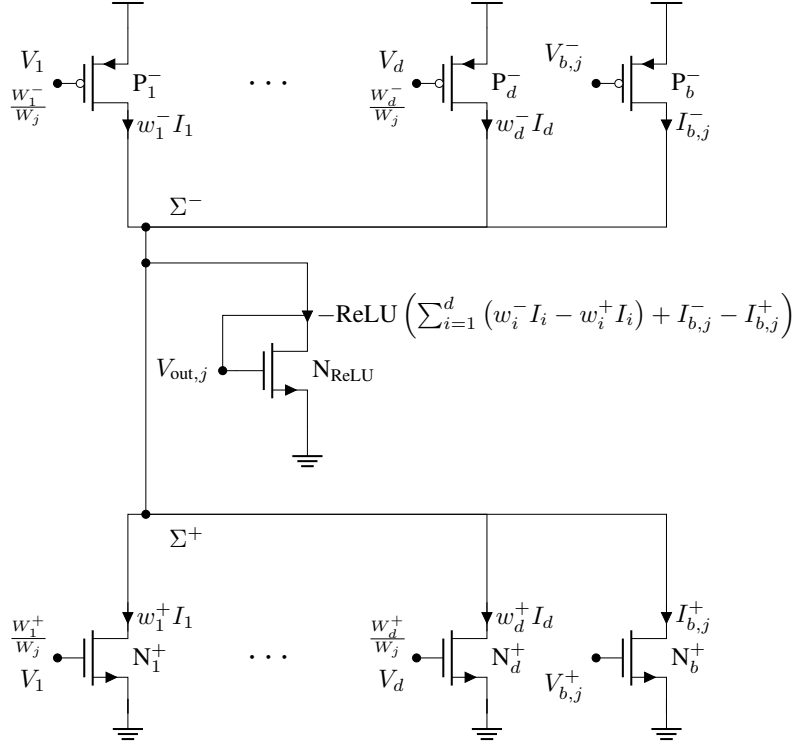
\begin{figure*}[ht!]
    \centering
    \scalebox{0.95}{\begin{circuitikz}[american, transform shape]

    \draw (0,6) node[pmos] (pn1) {$\text{P}_{1}^-$};
    \draw (pn1.source) node[rground, yscale=-1] {};
    \draw (pn1.gate) node[circ, label={above left:$V_{1}$}] {};
    
    \node at (2,6) {\Large $\cdots$};
    
    \draw (5,6) node[pmos] (pnd) {$\text{P}_{d}^-$};
    \draw (pnd.source) node[rground, yscale=-1] {};
    \draw (pnd.gate) node[circ, label={above left:$V_{d}$}] {};

    \draw (7.5,6) node[pmos] (pnb) {$\text{P}_{b}^-$};
    \draw (pnb.source) node[rground, yscale=-1] {};
    \draw (pnb.gate) node[circ, label={above left:$V_{b,j}^-$}] {};

    \draw (0,-2) node[nmos] (nn1) {$\text{N}_{1}^+$};
    \draw (nn1.source) -- ++(0,0.2) node[ground] {};
    \draw (nn1.gate) node[circ, label={below left:$V_{1}$}] {};
    
    \node at (2,-2) {\Large $\cdots$};
    
    \draw (5,-2) node[nmos] (nnd) {$\text{N}_{d}^+$};
    \draw (nnd.source) -- ++(0,0.2) node[ground] {};
    \draw (nnd.gate) node[circ, label={below left:$V_{d}$}] {};

    \draw (7.5,-2) node[nmos] (nnb) {$\text{N}_{b}^+$};
    \draw (nnb.source) -- ++(0,0.2) node[ground] {};
    \draw (nnb.gate) node[circ, label={below left:$V_{b,j}^+$}] {};

    \coordinate (sumtop) at (0.25,4);
    
    \draw (pn1.drain) |- (sumtop);
    \draw (pnd.drain) |- (sumtop);
    \draw (pnb.drain) |- (sumtop);

    \coordinate (sumbot) at (0.25,0);
    
    \draw (nn1.drain) |- (sumbot);
    \draw (nnd.drain) |- (sumbot);
    \draw (nnb.drain) |- (sumbot);

    \draw (sumtop) -- (sumbot);
    \draw (2.5,2) node[nmos] (nrelu) {$\text{N}_\text{ReLU}$};
    \draw (nrelu.source) node[ground, yscale=1] {};
    \draw (nrelu.gate) -- ++(-0.2,0) |- (nrelu.drain);
    \draw (nrelu.drain) |- (0.25,3.5);
    \node[circ] at (0.25,3.5) {};

    \node[circ] at (sumtop) {};
    \node[circ] at (sumbot) {};

    \draw (pn1.drain) to[short, i_<=$w_{1}^- I_{1}$] ++(0,0.4);
    \draw (pnd.drain) to[short, i_<=$w_{d}^- I_{d}$] ++(0,0.4);
    \draw (pnb.drain) to[short, i_<=$I_{b,j}^-$] ++(0,0.4);
    \draw (nn1.drain) to[short, i=$w_{1}^+ I_{1}$] ++(0,-0.4);
    \draw (nnd.drain) to[short, i=$w_{d}^+ I_{d}$] ++(0,-0.4);
    \draw (nnb.drain) to[short, i=$I_{b,j}^+$] ++(0,-0.4);
    \draw (nrelu.drain) to[short, i_<=$-\text{ReLU}\left(\sum_{i=1}^d\left(w_{i}^- I_{i}-w_{i}^+ I_{i}\right)+I_{b,j}^- -I_{b,j}^+\right)$] ++(0,0);

    \draw (nrelu.gate) ++(-0.2,0) node[circ, label={left:$V_{\text{out},j}$}] {};
    
    \node[above right=0.05 and 0.2 of sumtop] {$\Sigma^-$};
    \node[below right=0.05 and 0.2 of sumbot] {$\Sigma^+$};

    \node[anchor=east, font=\footnotesize] at ($(pn1.gate) + (0,-0.3)$) {$\frac{W_{1}^-}{W_j}$};
    \node[anchor=east, font=\footnotesize] at ($(pnd.gate) + (0,-0.3)$) {$\frac{W_{d}^-}{W_j}$};
    \node[anchor=east, font=\footnotesize] at ($(nn1.gate) + (0,0.3)$) {$\frac{W_{1}^+}{W_j}$};
    \node[anchor=east, font=\footnotesize] at ($(nnd.gate) + (0,0.3)$) {$\frac{W_{d}^+}{W_j}$};

\end{circuitikz}}
    \caption{\textbf{FC layer with anti-ReLU activation.}
    Same structure as Figure~\ref{fig:FC_ReLU}, with diode-connected NMOS at output.}
    \label{fig:FC_anti_ReLU}
\end{figure*}

All NMOS transistors are sized with a width of \SI{5}{\micro\meter} and a length of \SI{5}{\micro\meter}, while PMOS transistors are sized with a width of \SI{5.5}{\micro\meter} and a length of \SI{5}{\micro\meter}. This sizing ensures robustness against mismatch and mitigates undesired short-channel effects. An exception is made for current mirrors implementing weight multiplication, where output transistor widths are selected from a precomputed lookup table. This table was generated by exhaustively simulating all combinations of input currents and transistor widths, mapping each desired weight-current product to the optimal width selection.

\subsection{Skip connections}\label{sec:skip}
Skip connections, which bypass intermediate layers to add earlier activations directly to later ones, are implemented naturally in the current domain using KCL. Since the circuit operates exclusively with positive currents, skip connections are separated into positive and negative branches. This separation preserves the unipolar constraint throughout the network while enabling residual-style architectures. No additional active components are required---skip connections reduce to wire routing that joins current paths at the appropriate summation nodes.

\subsection{CMOS FQ BMRU}\label{sec:CMOSBMRU}
The proposed circuit is based on a dual-Heaviside feedback structure, conceptually illustrated in Figure~\ref{fig:CMOSBMRU} (top left panel). It consists of two interdependent nonlinear thresholding elements:
\begin{itemize}
    \item A primary Heaviside element that activates when the input current $I_{\text{in}}$ exceeds a threshold $I_{\text{thresh}}$ (the high threshold of the Schmitt trigger), producing an output current $I_{\text{gain}}$ (the high output value).
    \item A secondary Heaviside element in feedback, triggered by the output of the first stage. It activates to inject a feedback current $I_{\text{width}}$ at the input of the circuit.
\end{itemize}

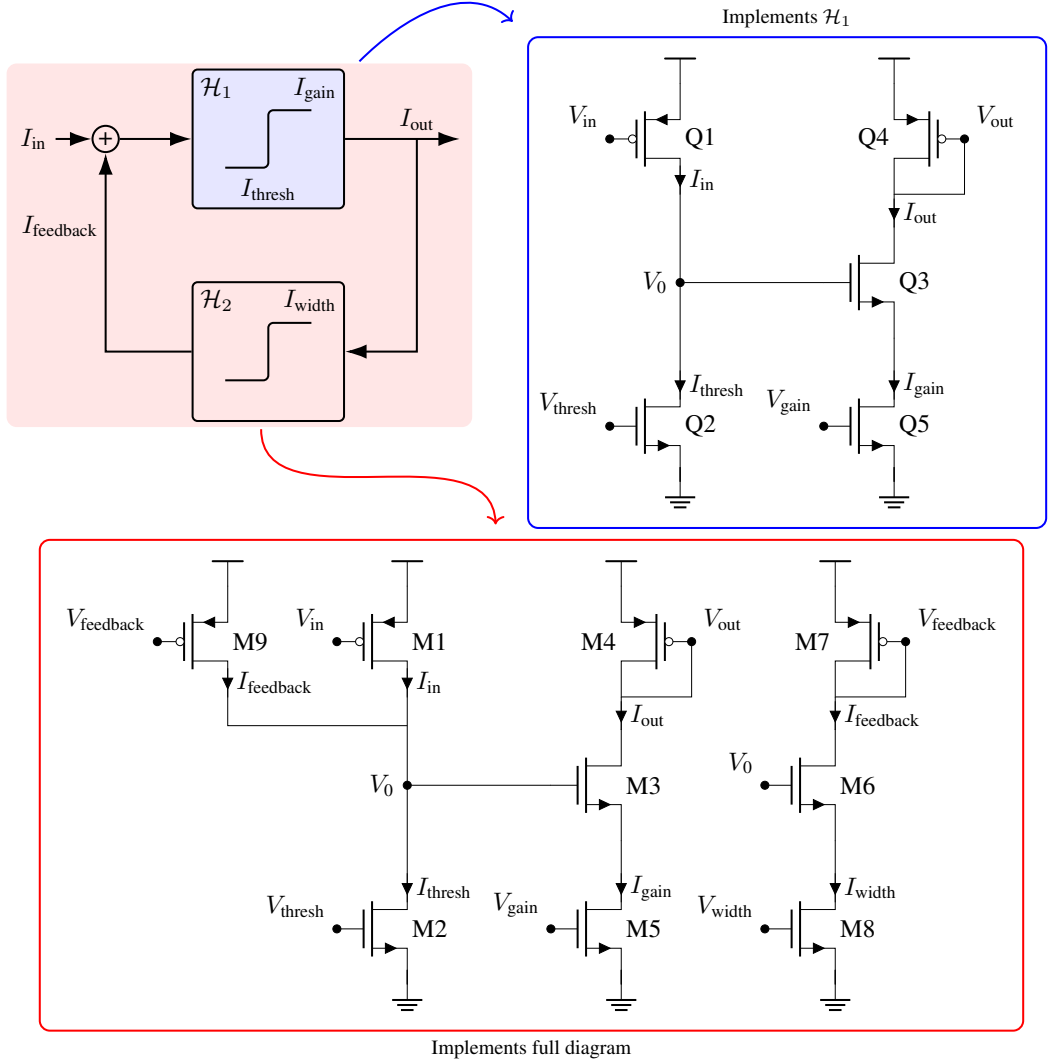
\begin{figure*}[ht!]
    \centering
    \scalebox{0.95}{\pgfdeclarelayer{background}
\pgfsetlayers{background,main}
\begin{circuitikz}[american, transform shape]
    \node[sum] (sum1) {+};
    \node[sblock, right=1. of sum1, label={[anchor=north west]north west:$\mathcal{H}_1$}, label={[anchor=north east]north east:$I_\text{gain}$}, label={[anchor=south]south:$I_\text{thresh}$}] (h1) {};
    \node[sblock, below=of h1, label={[anchor=north west]north west:$\mathcal{H}_2$}, label={[anchor=north east]north east:$I_\text{width}$}] (h2) {};
    
    \node[left=0.5 of sum1] (Iapp) {$I_\text{in}$};
    \coordinate[right=1. of h1, label=above:{$I_\text{out}$}] (Iout) {};
    \coordinate[below left=0.8 and 0.5 of sum1, label=below:{$I_\text{feedback}$}] (Ifb) {};
    
    \draw[signal] (Iapp) -- (sum1) node[pos=1, above] (sig1) {};
    \draw[signal] (sum1) -- (h1) node[midway, above] (sig2) {};
    \draw[signal] (h1) -- (Iout) -- ++(0.6,0) node[midway, right] (sig3) {};
    \draw[signal] (Iout) |- (h2.east) node[pos=0.25, right] (sig4) {};
    \draw[signal] (h2.west) -| (sum1.south) node[pos=0.5, left] (sig5) {};
    \coordinate (leftmargin) at ($(Iout) + (0.7,0)$);

    \draw ($(sum1) +(8,0)$) node[pmos] (q1) {Q1};
    \draw ($(sum1) +(8,-4)$) node[nmos] (q2) {Q2};
    
    \draw (q1.drain) -- (q2.drain) coordinate[midway,circ,label={left:$V_\text{0}$}] (v1);
    \draw (v1) -- ++(2.,0) node[nmos, anchor=G] (q3) {Q3};
    \draw (q2.source) -- ++(0,0.2) node[ground] {};
    \draw (q1.source) node[rground, yscale=-1] {};
    
    \draw ($(q3) +(0,2)$) node[pmos, xscale=-1] (q4) {\ctikzflipx{Q4}};
    \draw (q3.drain) -- (q4.drain);
    \draw (q4.gate) |- (q4.drain);
    \draw (q4.source) node[rground, yscale=-1] {};
    
    \draw ($(q3) -(0,2)$) node[nmos] (q5) {Q5};
    \draw (q3.source) -- (q5.drain);
    \draw (q5.source) -- ++(0,0.2) node[ground] {};
    
    \draw (q2.gate) node[circ,label={above left:$V_\text{thresh}$}] {};
    \draw (q2.drain) to[short, i=$I_\text{thresh}$] ++(0,-0.4);
    
    \draw (q1.gate) node[circ,label={above left:$V_\text{in}$}] {};
    \draw (q1.drain) to[short, i_<=$I_\text{in}$] ++(0,0.4);
    
    \draw (q5.gate) node[circ,label={above left:$V_\text{gain}$}] {};
    \draw (q5.drain) to[short, i=$I_\text{gain}$] ++(0,-0.4);
    
    \draw (q3.drain) to[short, i_<=$I_\text{out}$] ++(0,0.4);
    \draw (q4.gate) node[circ,label={above right:$V_\text{out}$}] {};

    \draw ($(sum1) +(4.2,-7)$) node[pmos] (m1) {M1};
    \draw ($(sum1) +(4.2,-11)$) node[nmos] (m2) {M2};

    \draw (m1.drain) -- (m2.drain) coordinate[midway,circ,label={left:$V_\text{0}$}] (v1_2);
    \draw (v1_2) -- ++(2.,0) node[nmos, anchor=G] (m3) {M3};
    \draw (m2.source) -- ++(0,0.2) node[ground] {};
    \draw (m1.source) node[rground, yscale=-1] {};

    \draw ($(m3) +(0,2)$) node[pmos, xscale=-1] (m4) {\ctikzflipx{M4}};
    \draw (m3.drain) -- (m4.drain);
    \draw (m4.gate) |- (m4.drain);
    \draw (m4.source) node[rground, yscale=-1] {};

    \draw ($(m3) -(0,2)$) node[nmos] (m5) {M5};
    \draw (m3.source) -- (m5.drain);
    \draw (m5.source) -- ++(0,0.2) node[ground] {};

    \draw (m3) ++(2.,0) node[nmos, anchor=G] (m6) {M6};

    \draw ($(m6) +(0,2)$) node[pmos, xscale=-1] (m7) {\ctikzflipx{M7}};
    \draw (m6.drain) -- (m7.drain);
    \draw (m7.gate) |- (m7.drain);
    \draw (m7.source) node[rground, yscale=-1] {};

    \draw ($(m6) -(0,2)$) node[nmos] (m8) {M8};
    \draw (m6.source) -- (m8.drain);
    \draw (m8.source) -- ++(0,0.2) node[ground] {};

    \draw ($(m1) -(2.5,0)$) node[pmos] (m9) {M9};
    \draw (m9.source) node[rground, yscale=-1] {};
    \draw (m9.drain) |- ($(m1.drain) -(0,0.4)$);

    \draw (m2.gate) node[circ,label={above left:$V_\text{thresh}$}] {};
    \draw (m2.drain) to[short, i=$I_\text{thresh}$] ++(0,-0.4);
    
    \draw (m1.gate) node[circ,label={above left:$V_\text{in}$}] {};
    \draw (m1.drain) to[short, i_<=$I_\text{in}$] ++(0,0.4);

    \draw (m5.gate) node[circ,label={above left:$V_\text{gain}$}] {};
    \draw (m5.drain) to[short, i=$I_\text{gain}$] ++(0,-0.4);
    
    \draw (m3.drain) to[short, i_<=$I_\text{out}$] ++(0,0.4);
    \draw (m4.gate) node[circ,label={above right:$V_\text{out}$}] {};

    \draw (m6.gate) node[circ,label={above left:$V_\text{0}$}] {};

    \draw (m6.drain) to[short, i_<=$I_\text{feedback}$] ++(0,0.4);
    \draw (m7.gate) node[circ,label={above right:$V_\text{feedback}$}] {};

    \draw (m8.gate) node[circ,label={above left:$V_\text{width}$}] {};
    \draw (m8.drain) to[short, i=$I_\text{width}$] ++(0,-0.4);

    \draw (m9.gate) node[circ,label={above left:$V_\text{feedback}$}] {};
    \draw (m9.drain) to[short, i_<=$I_\text{feedback}$] ++(0,0.4);

    \coordinate (hbox-nw) at ($(q1) + (-2.,1.3)$);
    \coordinate (hbox-se) at ($(q5) + (2.,-1.3)$);
    
    \node[draw=blue, thick, rounded corners, fit=(hbox-nw)(hbox-se), label=above:{\footnotesize Implements $\mathcal{H}_1$}] (hbox) {};
    
    \coordinate (fbox-nw) at ($(m9) + (-2.5,1.3)$);
    \coordinate (fbox-se) at ($(m8) + (2.5,-1.3)$);
    
    \node[draw=red, thick, rounded corners, fit=(fbox-nw)(fbox-se), label=below:{\footnotesize Implements full diagram}] (fbox) {};

    \draw[->, blue, thick] ($(h1.north east) + (0.2,0.1)$) to[out=45, in=135] ($(hbox.north west) + (-0.2,0.2)$);
    
    \draw[->, red, thick] ($(h2.south) + (-0.1,-0.1)$) to[out=-90, in=90] ($(fbox.north) + (-0.5,0.2)$);

    \begin{pgfonlayer}{background} 
      \node[
        fill=red!10,
        rounded corners,
        inner sep=2pt,
        fit=(sum1)(h1)(h2)(Iapp)(Iout)(Ifb)(sig1)(sig2)(sig3)(sig4)(sig5)(leftmargin)
      ] (redbox) {};

      \node[
        fill=blue!10,
        rounded corners,
        inner sep=0pt,
        fit=(h1)
      ] {}; 
    \end{pgfonlayer}
\end{circuitikz}}
    \caption{\textbf{CMOS implementation of the FQ BMRU cell.}
    Top left: conceptual dual-Heaviside feedback architecture. Top right (blue): single Heaviside element $\mathcal{H}_1$ using 5 transistors. Bottom (red): complete Schmitt trigger with feedback, using 9 transistors total.}
    \label{fig:CMOSBMRU}
\end{figure*}

The hysteresis behavior arises naturally from the feedback current $I_{\text{width}}$, which lowers the effective input threshold to $I_{\text{thresh}} - I_{\text{width}}$, thus enabling bistability. The three control parameters $I_{\text{thresh}}$, $I_{\text{gain}}$, and $I_{\text{width}}$ are all independently adjustable via bias voltages, allowing straightforward tuning. Because the circuit operates only with positive currents, the constraint $I_{\text{thresh}} > I_{\text{width}}$ must be met to preserve bistability, ensuring that the positive feedback does not override the comparator threshold.

The top right panel (blue) of Figure~\ref{fig:CMOSBMRU} shows a minimal CMOS circuit that implements a fully tunable unipolar Heaviside function ($\mathcal{H}_1$) in current mode. The input current is provided through a voltage $V_\text{in}$ at the output of a PMOS current mirror (Q1), while the threshold current $I_{\text{thresh}}$ is set through a voltage $V_\text{thresh}$ by an NMOS transistor (Q2). Together, Q1 and Q2 form a switching current comparator. If $I_{\text{in}} < I_{\text{thresh}}$, then $V_0 \approx 0$ and Q3 is off, resulting in a zero output current $I_{\text{out}}$. As soon as $I_{\text{in}} > I_{\text{thresh}}$, $V_0$ increases to the supply voltage and Q3 turns on, allowing current to flow in the output branch. At that point, the current flowing through this branch corresponds to the bias current $I_{\text{gain}}$, also set through a voltage $V_\text{gain}$ by an NMOS transistor (Q5). The output current flows through the diode-connected transistor Q4 and can be mirrored to other stages via a PMOS transistor driven by $V_\text{out}$.

\begin{figure*}[t!]
  \vskip 0.1in
  \begin{center}
    \centerline{\includegraphics[width=\linewidth]{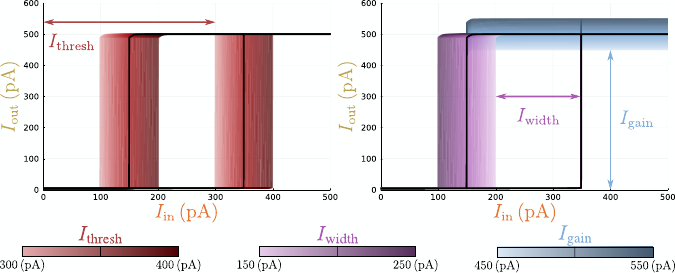}}
    \caption{
      \textbf{Tunability of the CMOS implementation of the FQ BMRU cell.}
      Input-output current relationship of the CMOS implementation of the FQ BMRU cell. All thresholds and output gain are independently tunable via bias currents.
    }
    \label{fig:CMOSBMRU_colorbar}
  \end{center}
\end{figure*}

The bottom panel (red) of Figure~\ref{fig:CMOSBMRU} represents the feedback interconnection of the two Heaviside functions, forming a current-mode unipolar Schmitt trigger capable of operating with \si{\nano\watt}-level supply power. $\mathcal{H}_1$ (Q1–Q5) is implemented through transistors M1–M5. For $\mathcal{H}_2$ to switch concurrently with $\mathcal{H}_1$, they share the same switching voltage $V_0$, removing the need for a separate comparator branch (like M1–M2). This means that if $I_{\text{in}} < I_{\text{thresh}}$, M6 is off and $I_\text{feedback} = I_\text{out} = 0$. Once $I_{\text{in}} > I_{\text{thresh}}$, M6 turns on and $I_\text{feedback}$ reaches $I_\text{width}$, the bias current of the feedback branch set through $V_\text{width}$ by an NMOS transistor (M8). This feedback current is mirrored through a PMOS current mirror (M7 and M9), and injected back into the comparator branch of $\mathcal{H}_1$ (M1, M2, and M9), thereby increasing the input current by $I_\text{feedback}$.

These three parameters correspond directly to the FQ BMRU formulation: $I_\text{thresh}$ implements $\beta_\text{hi}$, $I_\text{thresh} - I_\text{width}$ implements $\beta_\text{lo}$, and $I_\text{gain}$ implements $\alpha$. Figure~\ref{fig:CMOSBMRU_colorbar} represents independent tunability of these three parameters.

All NMOS transistors of the cell are sized with a width of \SI{5}{\micro\meter} and a length of \SI{5}{\micro\meter}, while PMOS transistors of the cell are sized with a width of \SI{5.5}{\micro\meter} and a length of \SI{5}{\micro\meter}. This sizing ensures robustness against mismatch and mitigates undesired short-channel effects. An exception was made for M3 and M6, which operate primarily in a digital regime and were sized with a width of \SI{2}{\micro\meter} and minimal length.

\subsection{Simulation results of the CMOS FQ BMRU}\label{sec:simulation_results}

All schematic-level simulations were conducted using Spectre with X-FAB \SI{180}{\nano\meter} models at a supply voltage of \SI{1.8}{V}. Baseline parameters were selected using the lookup table approach described in Section~\ref{sec:currentmirror}: $I_\text{gain} = \SI{486}{pA}$, $I_\text{thresh} = \SI{368}{pA}$, and $I_\text{width} = \SI{216}{pA}$, yielding switching thresholds of \SI{150}{pA} and \SI{350}{pA} with a high output state of \SI{500}{pA}.

\begin{figure*}[ht!]
    \centering
    \includegraphics[width=\linewidth]{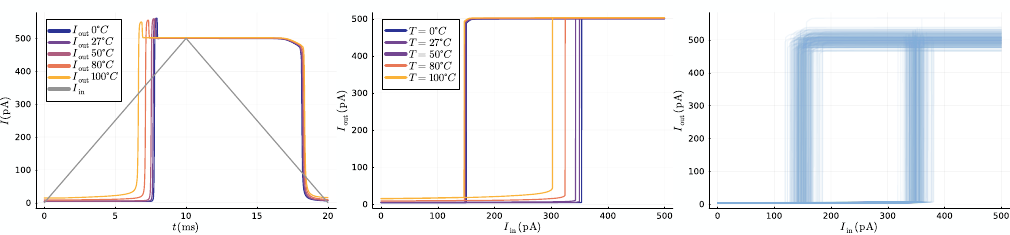}
    \caption{\textbf{CMOS FQ BMRU cell simulation results.}
    Transient simulation under triangular input current for different operating temperatures (left). DC sweep demonstrating hysteretic behavior for different operating temperatures (middle). Monte Carlo analysis with $3\sigma$ process variation at room temperature (right). Baseline parameters: $I_\text{gain} = \SI{486}{pA}$, $I_\text{thresh} = \SI{368}{pA}$, $I_\text{width} = \SI{216}{pA}$.}
    \label{fig:transient}
\end{figure*}

\begin{figure*}[ht!]
    \centering
    \includegraphics[width=\linewidth]{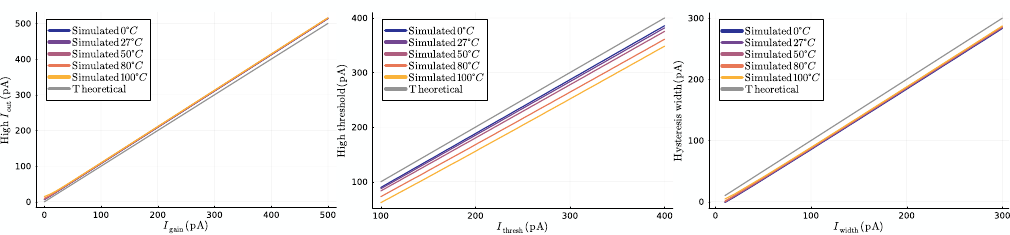}
    \caption{\textbf{Tunability of CMOS FQ BMRU cell parameters.}
    $I_\text{gain}$ sweep from 0 to \SI{500}{pA} for different operating temperatures  (left). $I_\text{thresh}$ sweep from \SI{100}{pA} to \SI{400}{pA} with $I_\text{width} = \SI{50}{pA}$ (middle). $I_\text{width}$ sweep from \SI{10}{pA} to \SI{300}{pA} (right).
    Baseline parameters as in Figure~\ref{fig:transient}.}
    \label{fig:params}
\end{figure*}

\paragraph{Transient response and DC characteristics.}
To evaluate hysteresis, the input current $I_{\text{in}}$ was swept linearly from 0 to \SI{500}{pA} and back using a triangular waveform. As shown in Figure~\ref{fig:transient} left for different operating temperatures, the circuit exhibits clear bistable behavior with sharp transitions. The switching points align with the expected levels of $I_{\text{thresh}} - I_{\text{width}}$ and $I_{\text{thresh}}$, and the high-state output current closely follows $I_\text{gain}$. A small ($\sim$\SI{10}{\percent}) overshoot is observed at transitions, attributed to the positive feedback inherent in the dual-Heaviside topology. DC simulations confirm the existence of two stable states and well-defined hysteresis (Figure~\ref{fig:transient} middle). Temperature mainly affects the upper switching point, determined by the M1--M2 comparator. However, since the circuit operates in the ultra-low power regime, typical operating conditions remain well below \SI{80}{\celsius}, ensuring robust operation.

\paragraph{Mismatch analysis.}
Monte Carlo simulations were performed under $3\sigma$ process variation to evaluate robustness to device mismatch. The circuit consistently maintained correct bistable switching behavior, with only minor deviations---typically a few tens of \si{pA}---in the threshold currents, hysteresis width, and high output level (Figure~\ref{fig:transient} right). These variations remain well within functional margins. The observed mismatches arise primarily from variations in current mirrors and biasing transistors, and could be further reduced by increasing device dimensions at the cost of silicon area.

\paragraph{Tuning properties.}
Systematic sweeps of each bias current confirm that $I_\text{gain}$, $I_\text{thresh}$, and $I_\text{width}$ independently control their corresponding functions (output level, threshold, and hysteresis width) in a linear manner across operating temperatures (Figure~\ref{fig:params}). Small offsets due to subthreshold leakage currents result in relative errors of \SI{2.8}{\percent}, \SI{4.5}{\percent}, and \SI{5.25}{\percent} for $I_\text{gain}$, $I_\text{thresh}$, and $I_\text{width}$ respectively at room temperature and highest tested currents. These predictable deviations do not impair functionality and can be compensated through calibration via lookup tables, as described in Section~\ref{sec:currentmirror}.

\section{Power breakdown analysis}\label{sec:power_breakdown}
 
\begin{figure*}[ht!]
  \begin{center}
    \centerline{\includegraphics[width=0.6\linewidth]{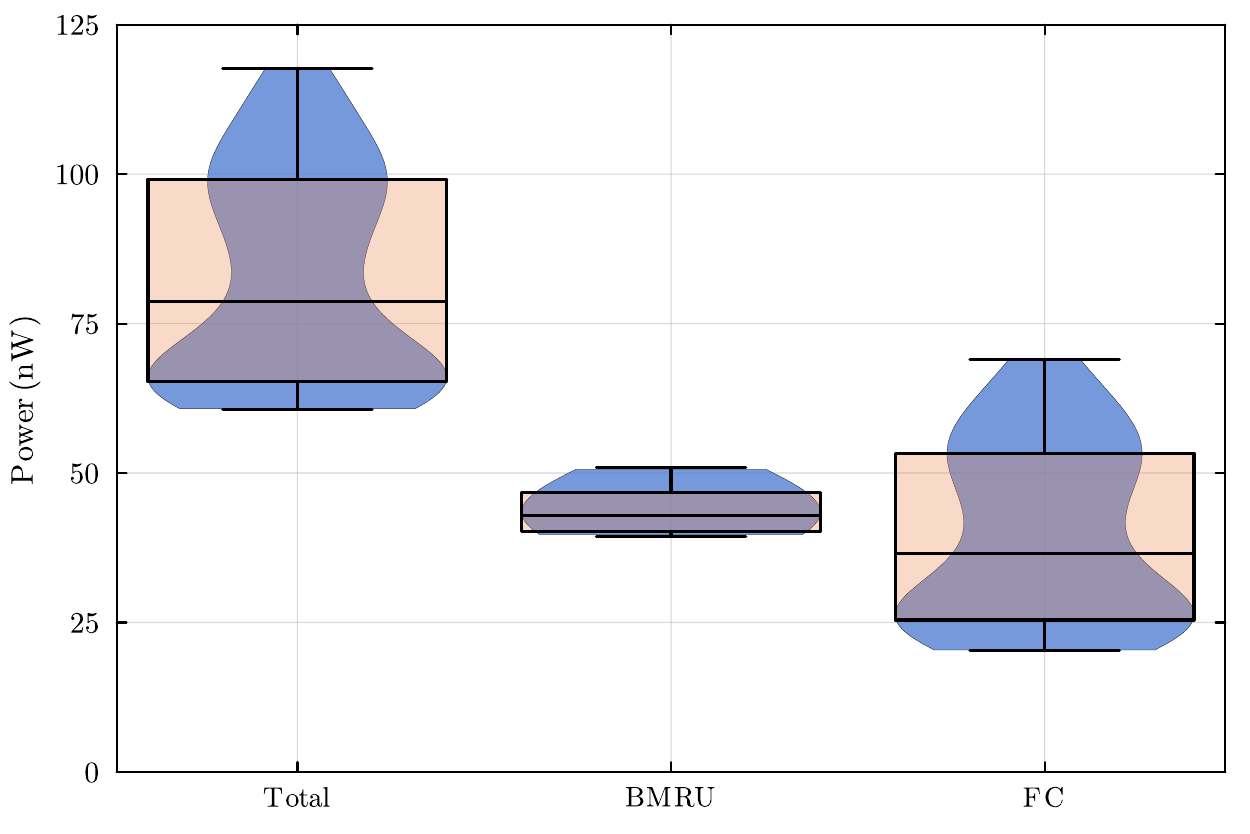}}
    \caption{
      \textbf{Component-level power breakdown across 50 inferences.}
      Power consumption of FQ BMRU cells versus FC layers for 50 inference samples. The approximately even split at $d=4$ indicates that both components contribute comparably to efficiency. FQ BMRU cells exhibit substantially lower power variance, consistent with stable discrete-output dynamics.
    }
    \label{fig:power_breakdown}
  \end{center}
\end{figure*}
 
To isolate the contribution of each building block, separate supply nets are assigned to BMRU cells and FC layers in Cadence Spectre simulation. Figure~\ref{fig:power_breakdown} presents the resulting power distribution across the 50 inferences reported in the main text. At $d=4$, power splits approximately evenly between BMRU cells (${\sim}40$~nW median) and FC layers (${\sim}30$~nW median). Notably, the power variance across inferences is substantially lower for BMRU cells than for FC layers, consistent with the observation that discrete-output dynamics yield stable, input-independent power consumption.

Component-level power analysis reveals a fundamental asymmetry in the network architecture: FQ BMRU cells scale linearly with state dimension $d$ (each cell is independent), while FC layers scale quadratically ($d \times d$ weight matrices require $d^2$ current mirrors). Estimated power scaling based on per-component measurements at $d=4$ from Cadence Spectre simulation is used to extrapolate to higher state dimensions (Table~\ref{tab:power_scaling}).

\begin{table}[t!]
\caption{%
  \textbf{Estimated power scaling with state dimension.} BMRU power increases linearly with the state dimension ($\propto d$), whereas FC and skip connections power grows quadratically ($\propto d^2$), leading to a decreasing relative cost of recurrence as $d$ increases. Estimates are extrapolated from per-component power measured at $d=4$.
}
\label{tab:power_scaling}
\centering
\small
\begin{tabular}{cccccc}
\toprule
 & $d=4$ & $d=8$  & $d=16$ & $d=32$ & $d=64$\\
\midrule
BMRU ($\sim$nW) & 40 (57\%) & 80 (40\%) & 160 (25\%) & 320 (14\%) & 640 (8\%) \\
FC and skip connections ($\sim$nW) & 30 (43\%) & 120 (60\%) & 480 (75\%)& 1920 (86\%) & 7680 (92\%)\\
\bottomrule
\end{tabular}
\end{table}

\section{Additional inference traces}\label{sec:inference_traces}
This appendix presents additional Cadence Spectre transient simulations that validate the hardware implementation against software predictions. Figures~\ref{fig:figS1}--\ref{fig:figS10} show ten representative inference examples, each displaying the two output logit currents (class~0: ``background''; class~1: ``yes'') and instantaneous power consumption. The samples were drawn from the test set to cover diverse input conditions (positive keyword, negative words, background noise), and include one case of hardware--software discrepancy as well as one misclassification common to both implementations. The main-text figure (Figure~\ref{fig:fig2}) corresponds to seed~51 and is omitted here to avoid redundancy.

Classification is determined by majority voting over all timesteps. Across 50 test samples, hardware predictions match software in 49 cases. The single discrepancy (Figure~\ref{fig:figS8}) occurs on a sample where the software prediction was maximally uncertain, with 51 timesteps predicting ``yes'' and 50 predicting ``background'', effectively a tie. This confirms that hardware--software mismatches arise only at the decision boundary rather than from systematic circuit errors. Power consumption remains consistently around \SI{100}{\nano\watt} across all samples for the RNN core, with minor input-dependent fluctuations reflecting switching activity.

\section{PVT corner validation}\label{sec:pvt_corners}

PVT (Process, Voltage, Temperature) corner analysis assesses robustness to manufacturing and environmental variation. Figures~\ref{fig:pvt_corners0} and~\ref{fig:pvt_corners1} present results across all five process corners (TT, FF, SS, FS, SF), the temperature range [$-27^\circ$C, $27^\circ$C, $81^\circ$C], and $\pm 10\%$ supply voltage variation. Classification correctness is preserved across all conditions for both representative inferences.

\section{Monte Carlo mismatch analysis}\label{sec:mismatch_analysis}

Monte Carlo mismatch analysis (200 samples, $3\sigma$ on all transistors) was performed on all 11 inferences shown in the manuscript. The results (Figures~\ref{fig:monte_carlo0}--\ref{fig:monte_carlo10}) confirm that misclassifications under mismatch occur exclusively on samples for which the network is already uncertain (or wrong) under nominal conditions. The bistable dynamics confer inherent mismatch immunity: discrete thresholding prevents small parameter perturbations from propagating to output errors unless the input already lies near the decision boundary.

\begin{figure*}[ht!]
  \begin{center}
    \centerline{\includegraphics[width=0.7\linewidth]{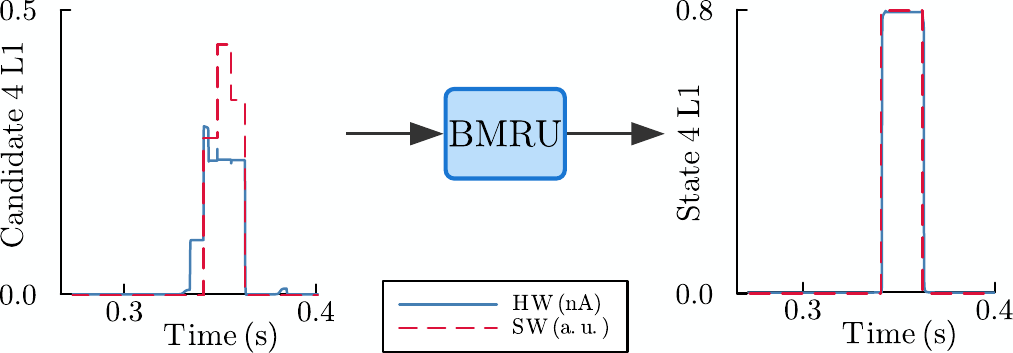}}
    \caption{
    \textbf{Illustration of $20\times$ error suppression at BMRU cell boundaries.}
    During inference for the sample in Figure~\ref{fig:fig2} (layer 1, fourth dimension), candidate signals at the BMRU inputs (left) exhibit discrepancies between software predictions and circuit simulations. After processing by the BMRU, the resulting states (right) show significantly improved agreement, demonstrating both strong noise attenuation and effective error suppression.
    }
    \label{fig:agreement_schematic}
  \end{center}
\end{figure*}

\section{Multi-class KWS evaluation}\label{sec:multiclass_kws}

To verify that the approach generalizes beyond binary classification, the FQ BMRU is evaluated on 11-class digit recognition (``zero'' through ``nine'' plus background noise) from the Google Speech Commands dataset. A $2\times 16$ architecture achieves competitive accuracy and substantially widens the output margin between classes (Figure~\ref{fig:multiclass}), thereby improving robustness to mismatch. Networks of this size remain within the sub-microwatt consumption estimates of Table~\ref{tab:power_scaling}. Given the strong hardware--software agreement demonstrated at $d=4$, comparable behavior is expected in hardware at larger dimensions.

\section{Signal agreement between software and hardware}\label{sec:signal_agreement}
 
Figures~\ref{fig:signal_agreement0}--\ref{fig:signal_agreement11} provide a detailed comparison of intermediate signals between software predictions and Cadence Spectre simulation at every stage of the network. The close agreement validates the end-to-end co-design methodology and confirms that the FQ reformulation produces a faithful mapping between learned parameters and circuit behavior. The $20\times$ error suppression at BMRU cell boundaries is clearly visible: candidate signals accumulate errors on the order of tens of picoamperes, while cell outputs exhibit errors of only a few picoamperes. Figure~\ref{fig:agreement_schematic} illustrates this $20\times$ error suppression at BMRU cell boundaries for the sample shown in Figure~\ref{fig:fig2}, at layer 1 and in the fourth dimension.

\section{Programmable hardware overhead}\label{sec:programmable_overhead}
 
Based on Cadence simulations of the non-programmable circuit (total area ${\sim}20{,}000$~$\mu$m$^2$) together with characterization of programmable building blocks, the following overhead is estimated for a fully programmable version employing binary-weighted current mirrors and digital shift registers.
 
\textbf{Binary-weighted current mirrors (4-bit resolution).} Inactive mirror branches consume negligible leakage, so the power overhead is insignificant. The area overhead is estimated at ${\sim}5{,}000$~$\mu$m$^2$ and scales linearly with the number of parameters.
 
\textbf{Shift registers ($d=4$, 4-bit quantization).} Power overhead is approximately ${\sim}100$~nW; area overhead is ${\sim}60{,}000$~$\mu$m$^2$. Both scale linearly with the number of parameters.
 
\textbf{Bias generation.} Drawing on current bias generators from the literature~\cite{yang2012iscas_bias, richter2024nce_bias}, extended with temperature compensation, an upper bound of ${\sim}50$~nW at 4-bit resolution is estimated. This figure is conservative, as bias currents can be shared across cells.
 
\textbf{Bias routing.} Routing overhead is minimal, as the xh018 process provides nine metal layers.
 
In total, the programmable system would remain within the sub-microwatt envelope for networks up to $d=16$.

\section{Broader impact}\label{sec:broaderimpact}
This work aims to advance machine learning through ultra-low power analog hardware implementations. The central contribution of this work is demonstrating that recurrent temporal processing can be added to analog inference platforms at linear marginal power cost. Since feedforward analog computation using current mirrors and KCL is already well established, this result implies that any existing or future analog feedforward accelerator could, in principle, gain the ability to process temporal sequences by integrating BMRU cells with modest additional overhead. This architectural composability extends further: because the FQ BMRU provides persistent bistable memory while other recurrent architectures (SSMs, gated RNNs) provide transient dynamics, hybrid designs combining both could cover a broader range of temporal processing requirements while remaining within ultra-low power budgets. The practical implications span domains where continuous temporal inference under extreme power constraints is currently infeasible: always-on environmental and structural monitoring over years-long deployments, implantable biomedical devices operating without frequent recharging, and distributed sensor networks where battery replacement is impractical or impossible. More broadly, by narrowing the gap between analog and digital inference capabilities, this work contributes to reducing the energy footprint of edge AI at a time when the computational demands of AI are growing rapidly. As with any advance in always-on sensing, deployment should be accompanied by appropriate consent and governance frameworks, particularly for audio-processing applications.

\section{Software and data}\label{sec:softanddata}
All code required to reproduce the experiments is publicly available on GitHub%
\footnote{\url{https://github.com/arthur-fyon/Fyon_CoDesign_2026}}. This software is not open source. It is protected by patents and is made available under a limited evaluation license only. The repository includes:
\begin{itemize}[topsep=2pt, parsep=0pt, partopsep=0pt]
    \item \textbf{Software training (Table~\ref{tab:benchmarks}):} JAX/Flax implementations of the BMRU, FQ BMRU, LRU, and minGRU, with both the software and hardware backbones, along with training scripts and evaluation code for all benchmarks in Table~\ref{tab:benchmarks}.
    \item \textbf{Noise robustness analysis (Figure~\ref{fig:noise_analysis}):} Complete inference and noise injection pipeline used to produce Figure~\ref{fig:noise_analysis}.
    \item \textbf{``yes'' KWS (Section~\ref{sec:kws_proof_of_concept}):} Complete training pipeline and export scripts for mapping learned parameters to circuit elements for analog inference.
    \item \textbf{Datasets:} The MNIST, Google Speech Commands, and Long Range Arena (ListOps) datasets~\cite{lecun1998gradient, warden2018speech, tay2021long} are publicly available. Preprocessing scripts and train/validation/test splits follow standard practices and are included in the repository.
    \item \textbf{Post-processing:} Julia scripts for parsing Cadence simulation outputs (logit currents, power traces, and all additional figures) and generating transistor width lookup tables from device characterization data.
\end{itemize}
Cadence Virtuoso design files (schematics, symbols, and Maestro configurations) are not included due to software licensing constraints and version compatibility issues, but are available upon reasonable request to the corresponding author for researchers with access to Cadence tools.

\begin{figure*}[ht!]
  \vskip 0.2in
  \begin{center}
    \centerline{\includegraphics[width=\linewidth]{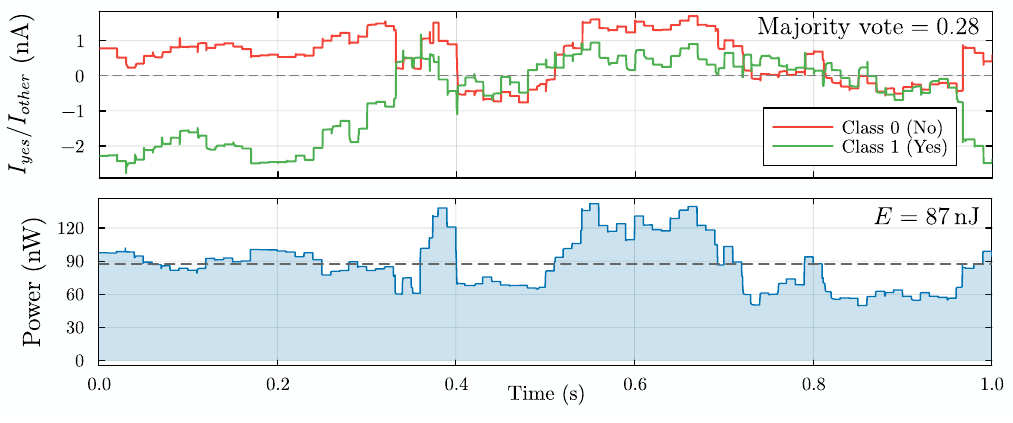}}
    \caption{\textbf{Hardware inference traces from Cadence Spectre simulation (seed 45).} Each panel shows output logit currents (top: $I_{\text{yes}}$ in green, $I_{\text{no}}$ in red) and power consumption (bottom) over the 101-frame input sequence. Spoken word: ``down''. Hardware prediction via majority vote: ``background''. Software prediction: ``background''.}
    \label{fig:figS1}
  \end{center}
\end{figure*}

\begin{figure*}[ht!]
  \vskip 0.2in
  \begin{center}
    \centerline{\includegraphics[width=\linewidth]{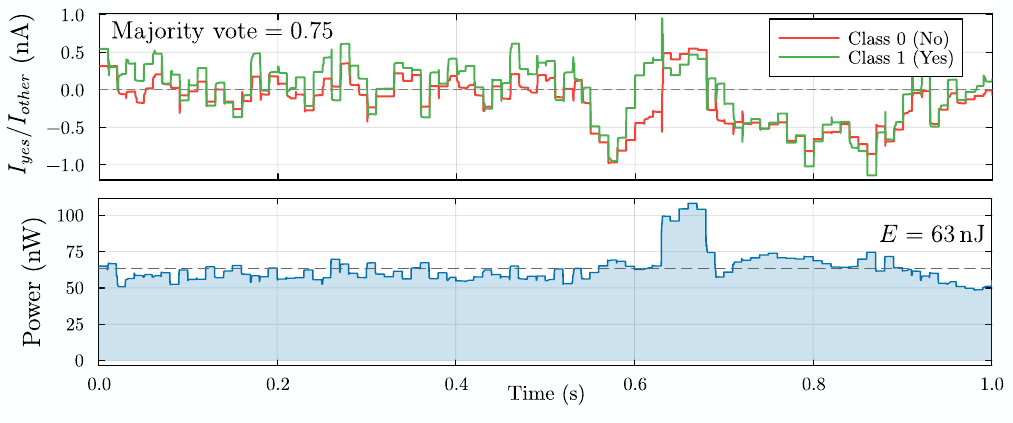}}
    \caption{\textbf{Hardware inference traces from Cadence Spectre simulation (seed 47).} Each panel shows output logit currents (top: $I_{\text{yes}}$ in green, $I_{\text{no}}$ in red) and power consumption (bottom) over the 101-frame input sequence. Spoken word: ``yes''. Hardware prediction via majority vote: ``yes''. Software prediction: ``yes''.}
    \label{fig:figS2}
  \end{center}
\end{figure*}

\begin{figure*}[ht!]
  \vskip 0.2in
  \begin{center}
    \centerline{\includegraphics[width=\linewidth]{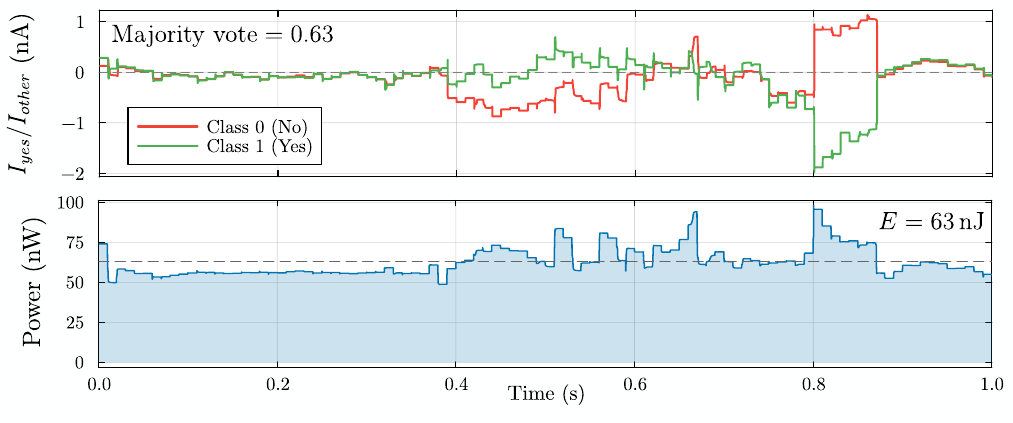}}
    \caption{\textbf{Hardware inference traces from Cadence Spectre simulation (seed 48).} Each panel shows output logit currents (top: $I_{\text{yes}}$ in green, $I_{\text{no}}$ in red) and power consumption (bottom) over the 101-frame input sequence. Spoken word: ``yes''. Hardware prediction via majority vote: ``yes''. Software prediction: ``yes''.}
    \label{fig:figS3}
  \end{center}
\end{figure*}
\clearpage

\begin{figure*}[ht!]
  \vskip 0.2in
  \begin{center}
    \centerline{\includegraphics[width=\linewidth]{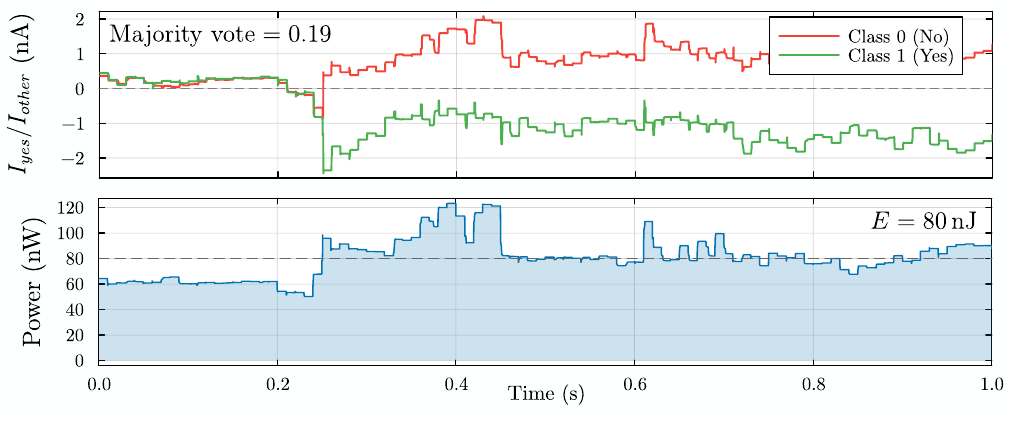}}
    \caption{\textbf{Hardware inference traces from Cadence Spectre simulation (seed 49).} Each panel shows output logit currents (top: $I_{\text{yes}}$ in green, $I_{\text{no}}$ in red) and power consumption (bottom) over the 101-frame input sequence. Spoken word: ``yes''. Hardware prediction via majority vote: ``background''. Software prediction: ``background''. In this case, both implementations misclassify the sample. Note that the spoken ``yes'' has an elongated ``s'' sound, reflected in the output logit dynamics.}
    \label{fig:figS4}
  \end{center}
\end{figure*}

\begin{figure*}[ht!]
  \vskip 0.2in
  \begin{center}
    \centerline{\includegraphics[width=\linewidth]{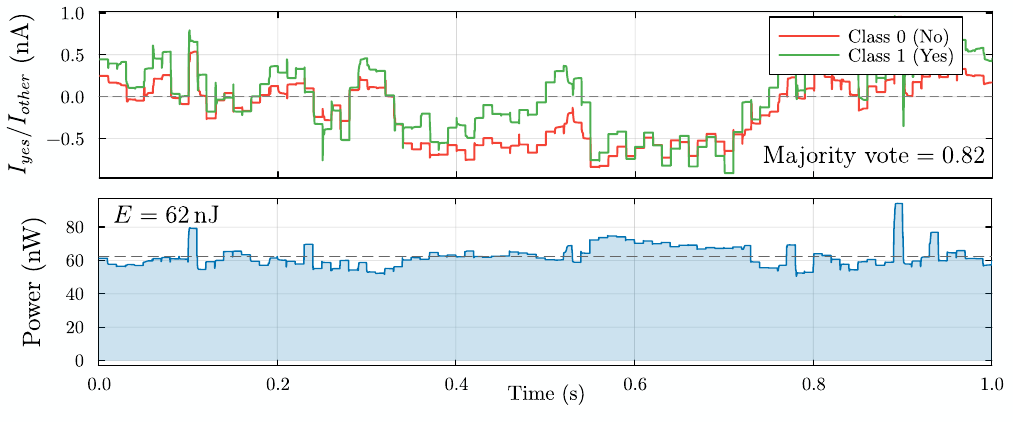}}
    \caption{\textbf{Hardware inference traces from Cadence Spectre simulation (seed 50).} Each panel shows output logit currents (top: $I_{\text{yes}}$ in green, $I_{\text{no}}$ in red) and power consumption (bottom) over the 101-frame input sequence. Spoken word: ``yes''. Hardware prediction via majority vote: ``yes''. Software prediction: ``yes''.}
    \label{fig:figS5}
  \end{center}
\end{figure*}
\clearpage

\begin{figure*}[ht!]
  \vskip 0.2in
  \begin{center}
    \centerline{\includegraphics[width=\linewidth]{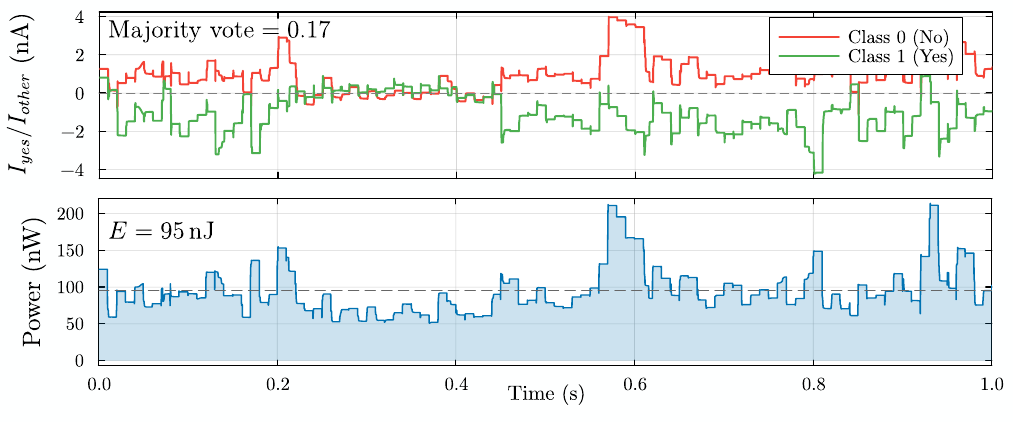}}
    \caption{\textbf{Hardware inference traces from Cadence Spectre simulation (seed 52).} Each panel shows output logit currents (top: $I_{\text{yes}}$ in green, $I_{\text{no}}$ in red) and power consumption (bottom) over the 101-frame input sequence. Spoken word: background noise (no speech). Hardware prediction via majority vote: ``background''. Software prediction: ``background''.}
    \label{fig:figS6}
  \end{center}
\end{figure*}

\begin{figure*}[ht!]
  \vskip 0.2in
  \begin{center}
    \centerline{\includegraphics[width=\linewidth]{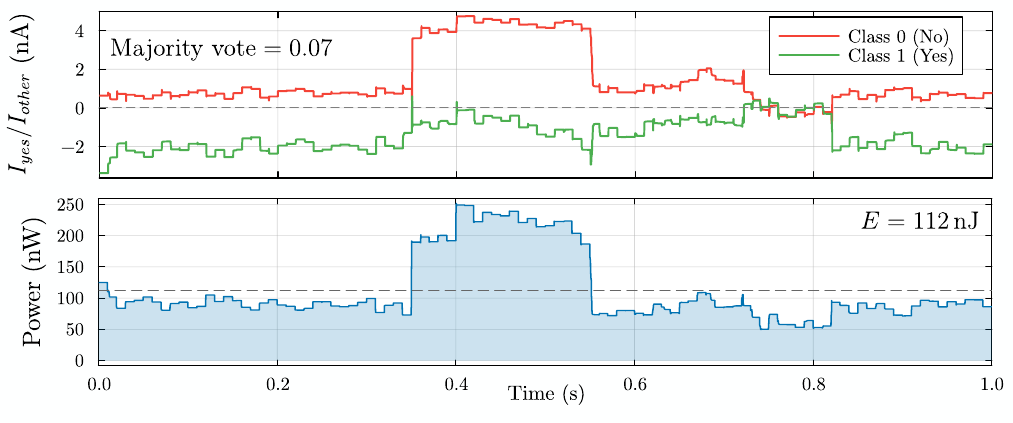}}
    \caption{\textbf{Hardware inference traces from Cadence Spectre simulation (seed 66).} Each panel shows output logit currents (top: $I_{\text{yes}}$ in green, $I_{\text{no}}$ in red) and power consumption (bottom) over the 101-frame input sequence. Spoken word: ``up''. Hardware prediction via majority vote: ``background''. Software prediction: ``background''.}
    \label{fig:figS7}
  \end{center}
\end{figure*}
\clearpage

\begin{figure*}[ht!]
  \vskip 0.2in
  \begin{center}
    \centerline{\includegraphics[width=\linewidth]{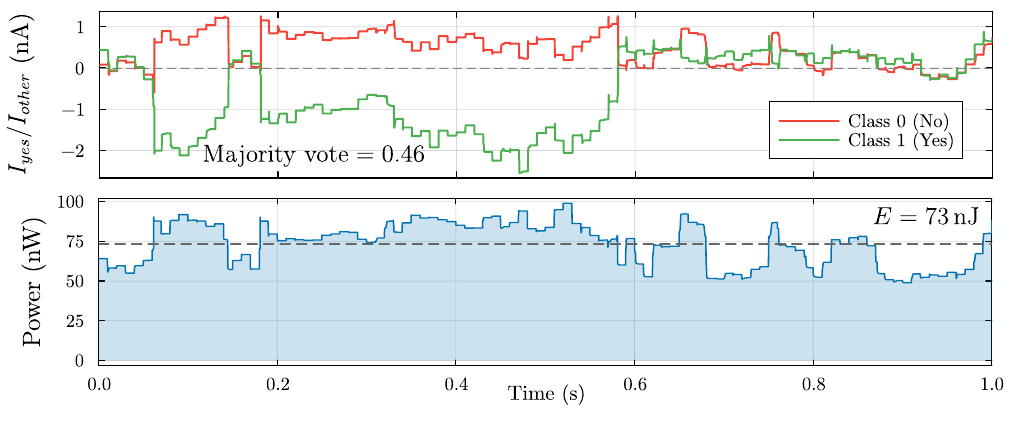}}
    \caption{\textbf{Hardware inference traces from Cadence Spectre simulation (seed 67).} Each panel shows output logit currents (top: $I_{\text{yes}}$ in green, $I_{\text{no}}$ in red) and power consumption (bottom) over the 101-frame input sequence. Spoken word: ``yes''. Hardware prediction via majority vote: ``background''. Software prediction: ``yes''. This is the only case across 50 test samples where hardware and software predictions differ.}
    \label{fig:figS8}
  \end{center}
\end{figure*}

\begin{figure*}[ht!]
  \vskip 0.2in
  \begin{center}
    \centerline{\includegraphics[width=\linewidth]{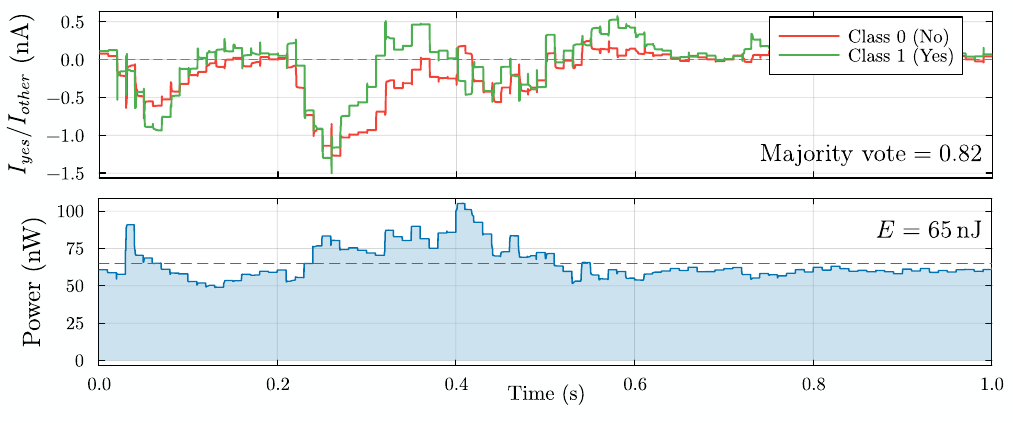}}
    \caption{\textbf{Hardware inference traces from Cadence Spectre simulation (seed 68).} Each panel shows output logit currents (top: $I_{\text{yes}}$ in green, $I_{\text{no}}$ in red) and power consumption (bottom) over the 101-frame input sequence. Spoken word: ``yes''. Hardware prediction via majority vote: ``yes''. Software prediction: ``yes''.}
    \label{fig:figS9}
  \end{center}
\end{figure*}
\clearpage

\begin{figure*}[ht!]
  \vskip 0.2in
  \begin{center}
    \centerline{\includegraphics[width=\linewidth]{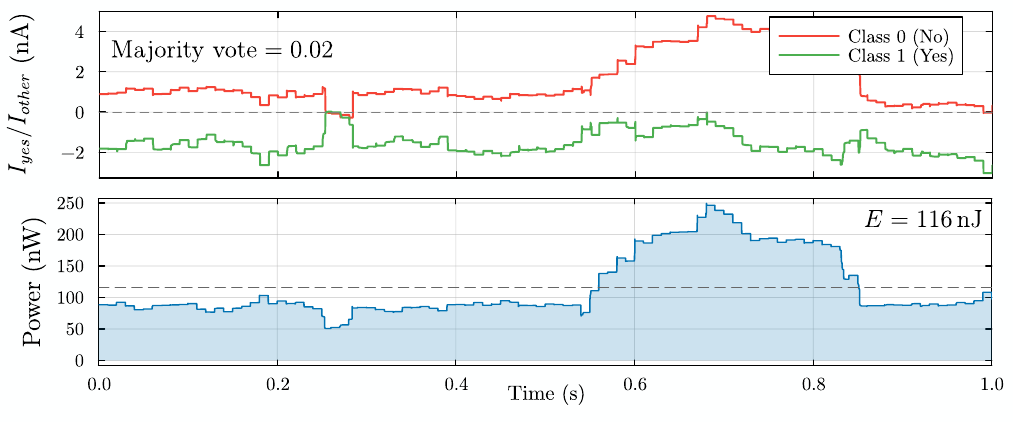}}
    \caption{\textbf{Hardware inference traces from Cadence Spectre simulation (seed 61).} Each panel shows output logit currents (top: $I_{\text{yes}}$ in green, $I_{\text{no}}$ in red) and power consumption (bottom) over the 101-frame input sequence. Spoken word: ``right''. Hardware prediction via majority vote: ``background''. Software prediction: ``background''.}
    \label{fig:figS10}
  \end{center}
\end{figure*}
 
\begin{figure*}[ht!]
  \begin{center}
    \centerline{\includegraphics[width=\linewidth]{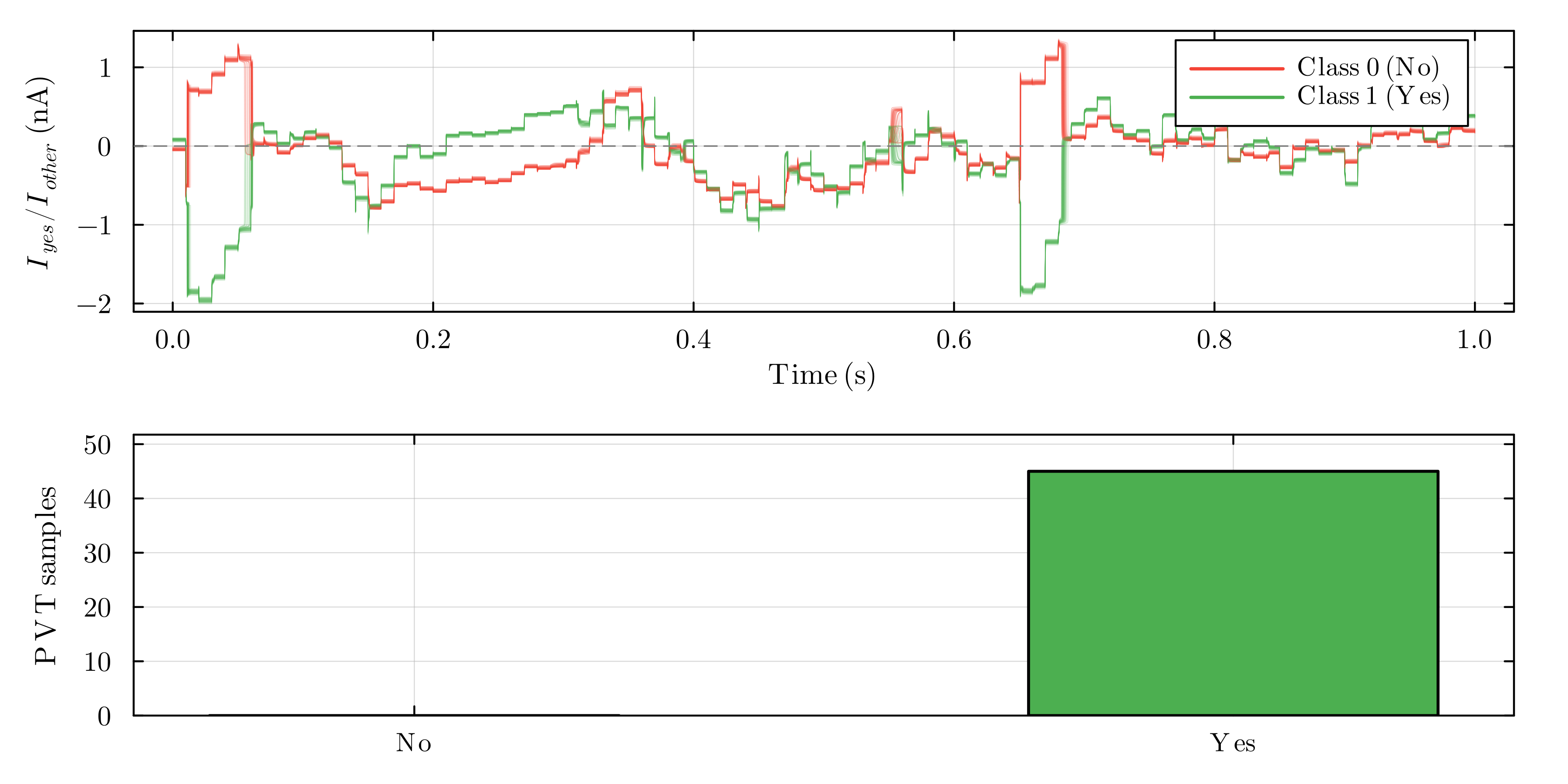}}
    \caption{
      \textbf{PVT corner validation from Cadence Spectre simulation (seed 51).}
      Each panel shows output logit currents (top: $I_{\text{yes}}$ in green, $I_{\text{no}}$ in red) over the 101-frame input sequence and the corresponding prediction for each PVT condition (bottom). All five process corners (TT, FF, SS, FS, SF), three temperatures ($-27^\circ$C, $27^\circ$C, $81^\circ$C), and $\pm 10\%$ supply voltage variation are evaluated. Spoken word: ``yes''. Correct classification is maintained across all conditions.
    }
    \label{fig:pvt_corners0}
  \end{center}
\end{figure*}
\clearpage

\begin{figure*}[ht!]
  \begin{center}
    \centerline{\includegraphics[width=\linewidth]{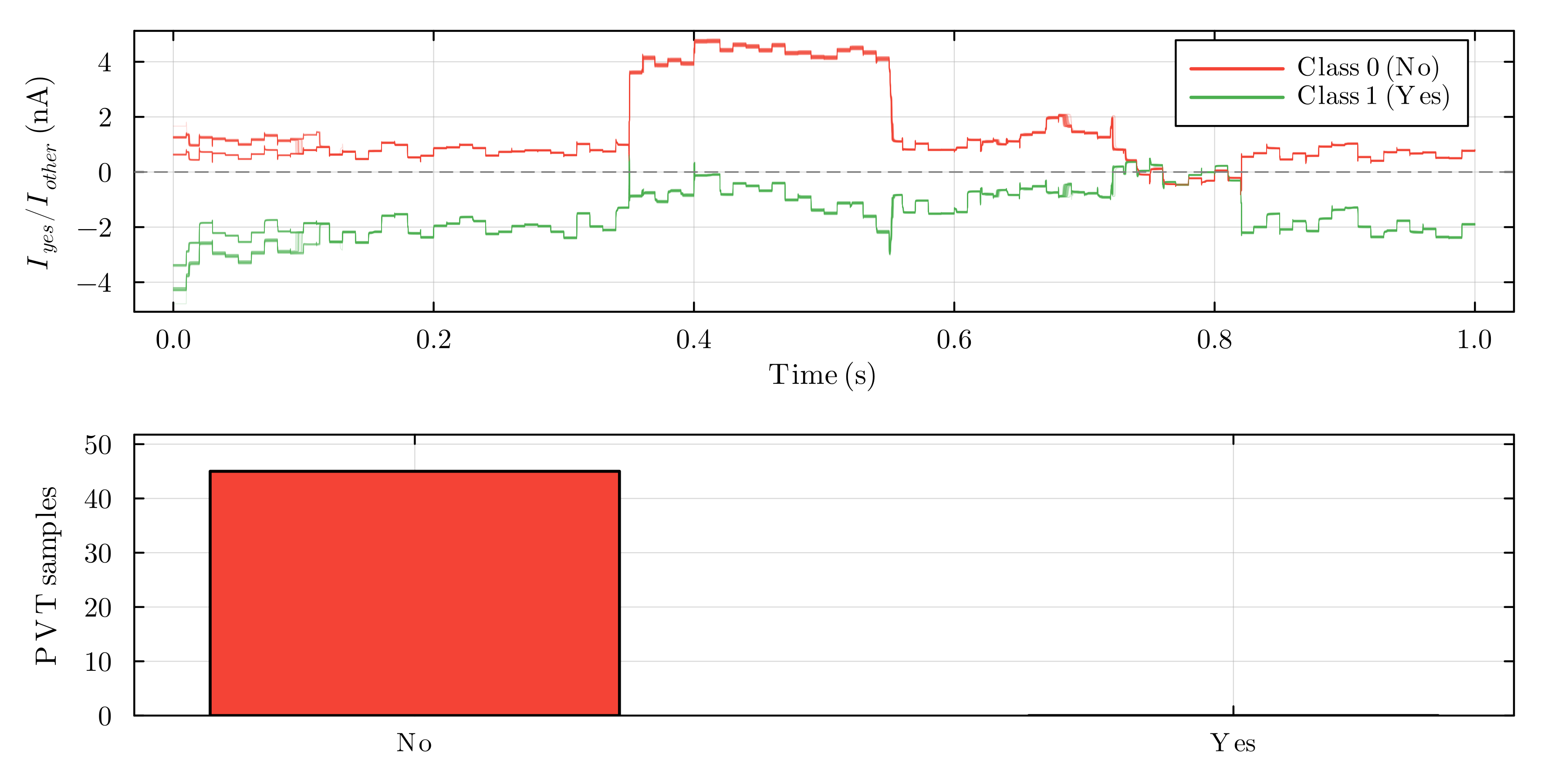}}
    \caption{
      \textbf{PVT corner validation from Cadence Spectre simulation (seed 66).}
      Each panel shows output logit currents (top: $I_{\text{yes}}$ in green, $I_{\text{no}}$ in red) over the 101-frame input sequence and the corresponding prediction for each PVT condition (bottom). All five process corners (TT, FF, SS, FS, SF), three temperatures ($-27^\circ$C, $27^\circ$C, $81^\circ$C), and $\pm 10\%$ supply voltage variation are evaluated. Input: background noise. Correct classification is maintained across all conditions.
    }
    \label{fig:pvt_corners1}
  \end{center}
\end{figure*}

\begin{figure*}[ht!]
  \begin{center}
    \centerline{\includegraphics[width=\linewidth]{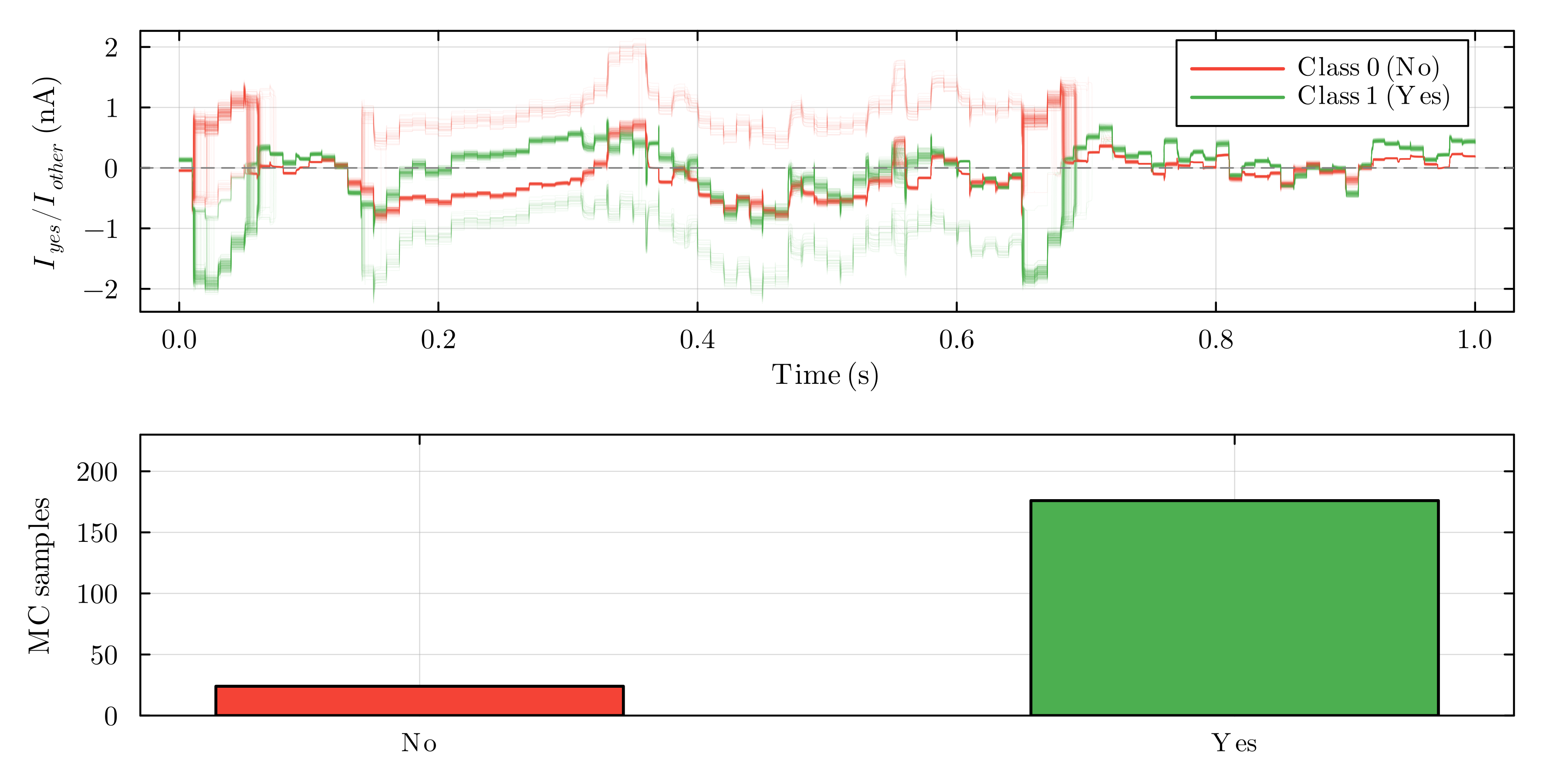}}
    \caption{
      \textbf{Monte Carlo mismatch analysis from Cadence Spectre simulation (seed 51).}
      Each panel shows output logit currents (top: $I_{\text{yes}}$ in green, $I_{\text{no}}$ in red) over the 101-frame input sequence and the corresponding prediction for each Monte Carlo sample (bottom). Analysis performed with $3\sigma$ mismatch variation on all transistors (200 samples). Spoken word: ``yes''. Nominal prediction: ``yes''. Impaired sample rate: 11.5\%.
    }
    \label{fig:monte_carlo0}
  \end{center}
\end{figure*}
\clearpage

\begin{figure*}[ht!]
  \vskip 0.2in
  \begin{center}
    \centerline{\includegraphics[width=\linewidth]{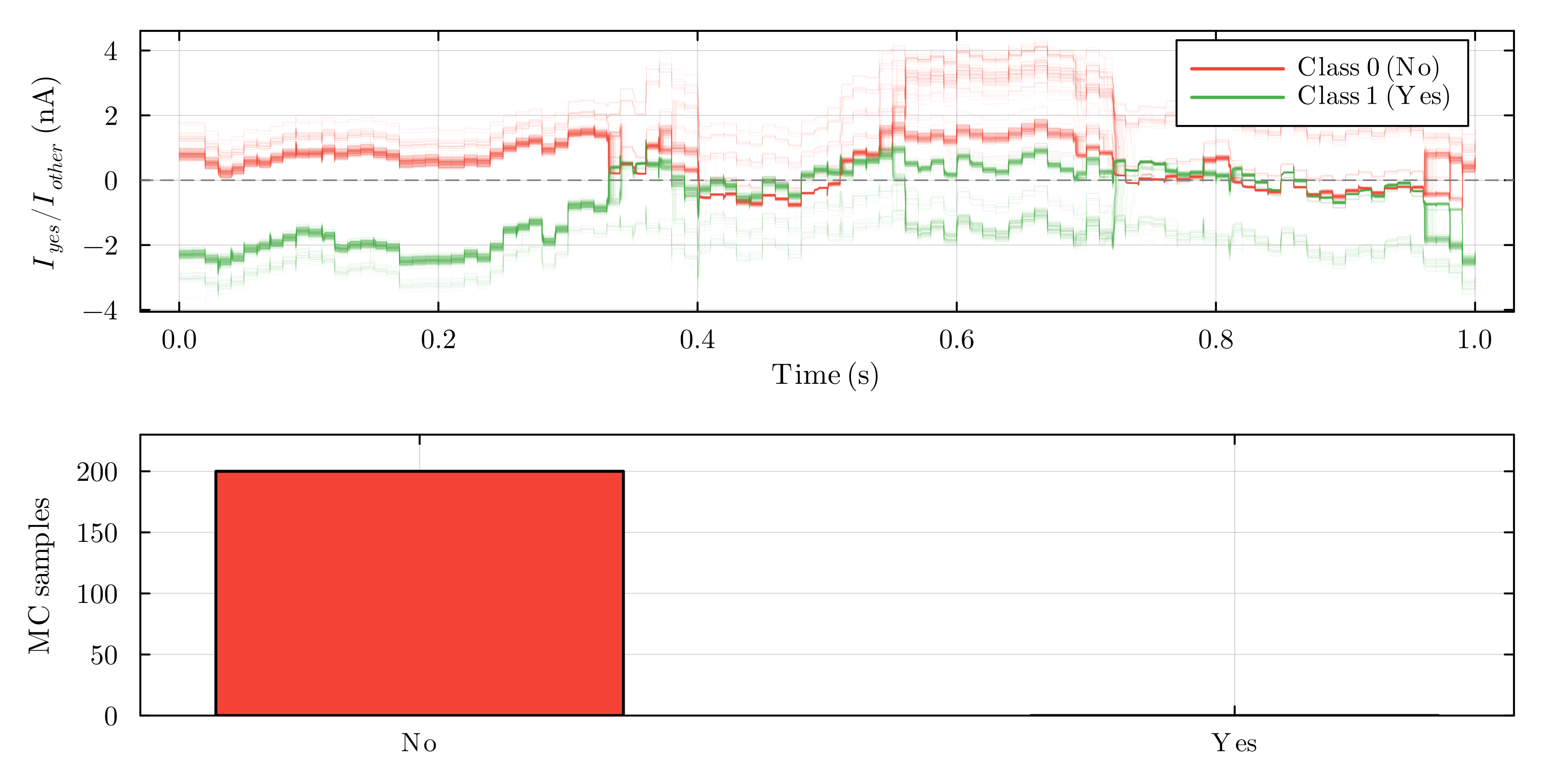}}
    \caption{
      \textbf{Monte Carlo mismatch analysis from Cadence Spectre simulation (seed 45).}
      Each panel shows output logit currents (top: $I_{\text{yes}}$ in green, $I_{\text{no}}$ in red) over the 101-frame input sequence and the corresponding prediction for each Monte Carlo sample (bottom). Analysis performed with $3\sigma$ mismatch variation on all transistors (200 samples). Spoken word: ``down''. Nominal prediction: ``background''. Impaired sample rate: 0\%.
    }
    \label{fig:monte_carlo1}
  \end{center}
\end{figure*}

\begin{figure*}[ht!]
  \vskip 0.2in
  \begin{center}
    \centerline{\includegraphics[width=\linewidth]{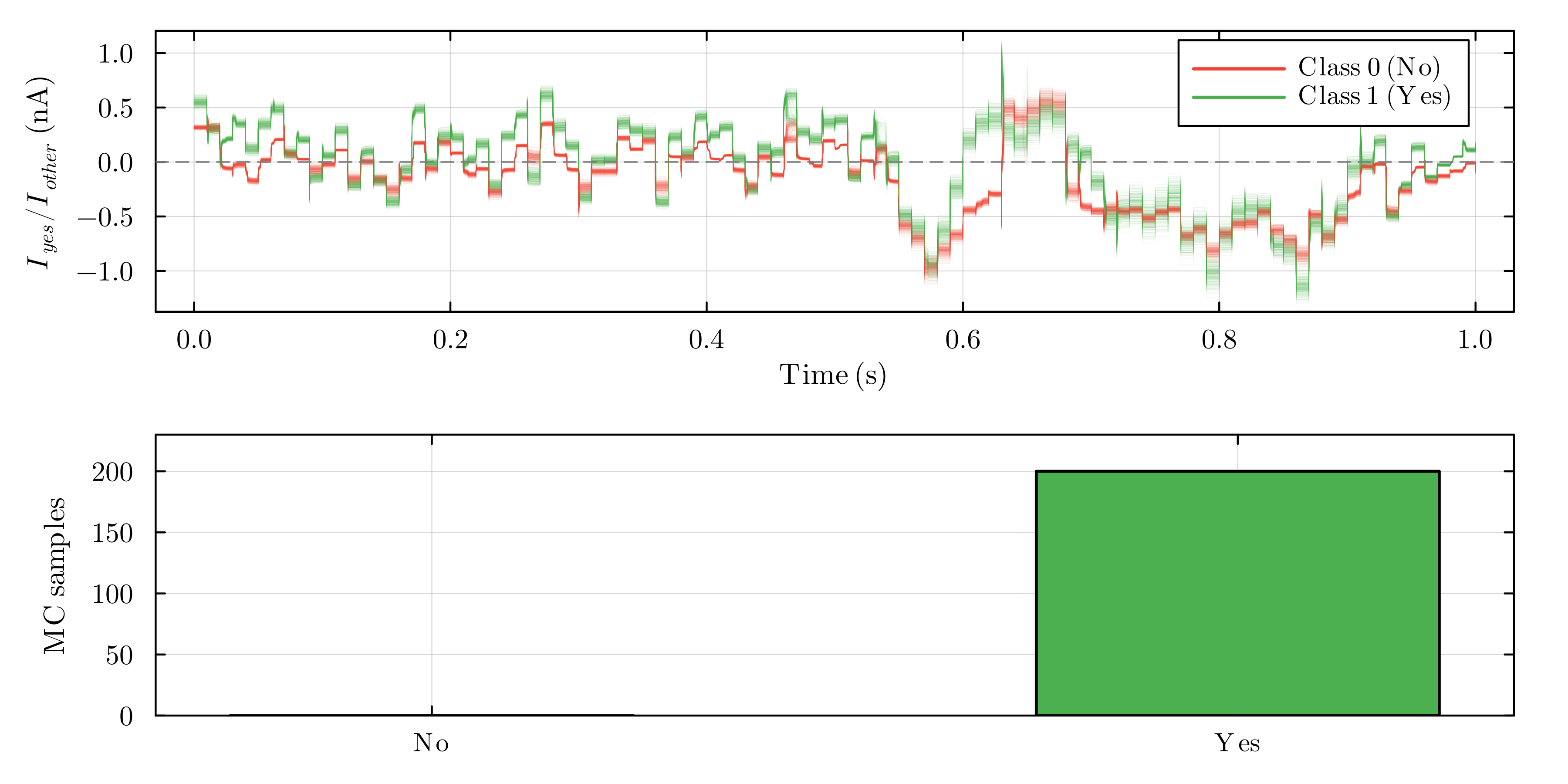}}
    \caption{
      \textbf{Monte Carlo mismatch analysis from Cadence Spectre simulation (seed 47).}
      Each panel shows output logit currents (top: $I_{\text{yes}}$ in green, $I_{\text{no}}$ in red) over the 101-frame input sequence and the corresponding prediction for each Monte Carlo sample (bottom). Analysis performed with $3\sigma$ mismatch variation on all transistors (200 samples). Spoken word: ``yes''. Nominal prediction: ``yes''. Impaired sample rate: 0\%.
    }
    \label{fig:monte_carlo2}
  \end{center}
\end{figure*}
\clearpage
 
\begin{figure*}[ht!]
  \vskip 0.2in
  \begin{center}
    \centerline{\includegraphics[width=\linewidth]{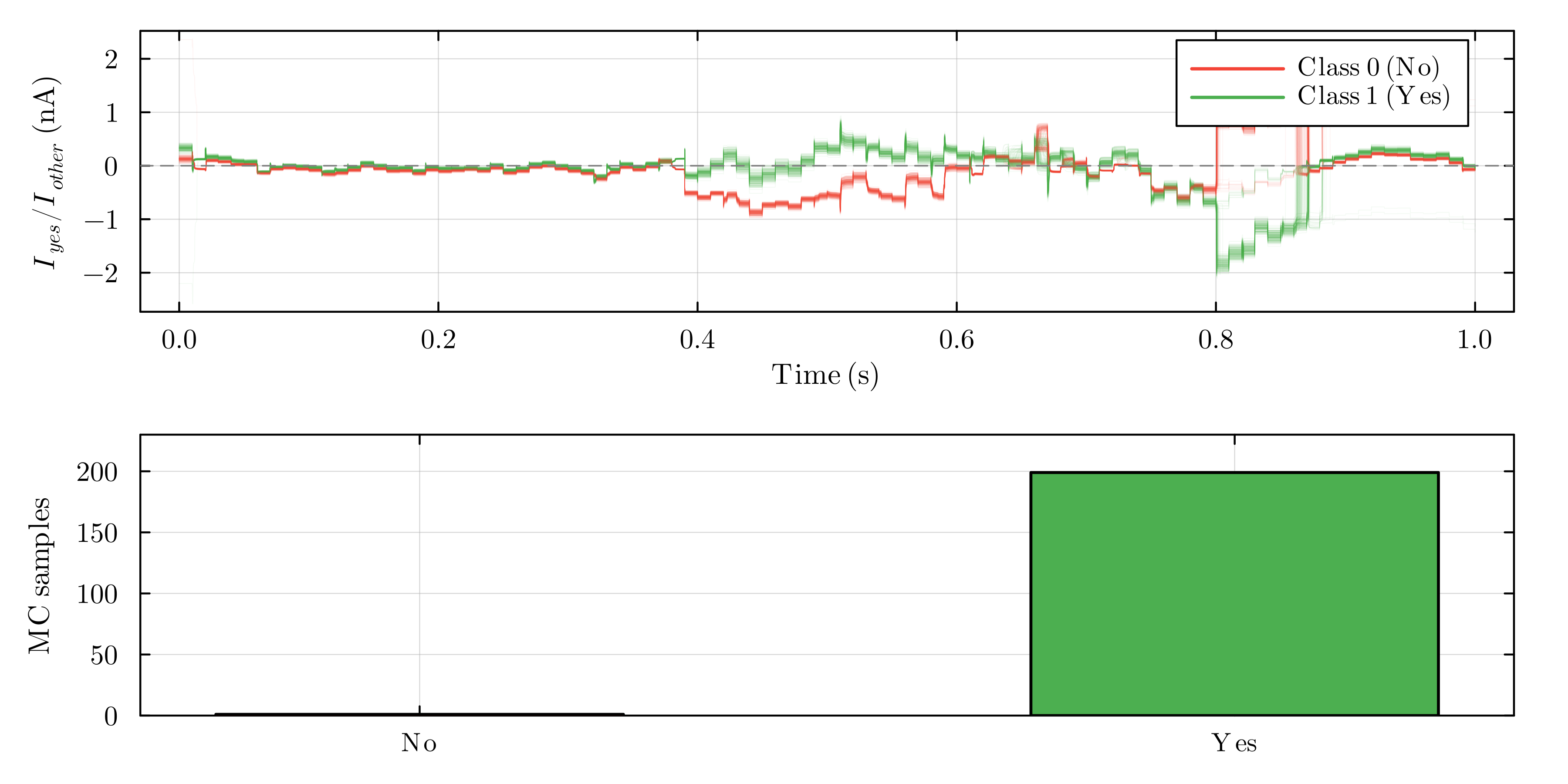}}
    \caption{
      \textbf{Monte Carlo mismatch analysis from Cadence Spectre simulation (seed 48).}
      Each panel shows output logit currents (top: $I_{\text{yes}}$ in green, $I_{\text{no}}$ in red) over the 101-frame input sequence and the corresponding prediction for each Monte Carlo sample (bottom). Analysis performed with $3\sigma$ mismatch variation on all transistors (200 samples). Spoken word: ``yes''. Nominal prediction: ``yes''. Impaired sample rate: 0.5\%.
    }
    \label{fig:monte_carlo3}
  \end{center}
\end{figure*}

\begin{figure*}[ht!]
  \vskip 0.2in
  \begin{center}
    \centerline{\includegraphics[width=\linewidth]{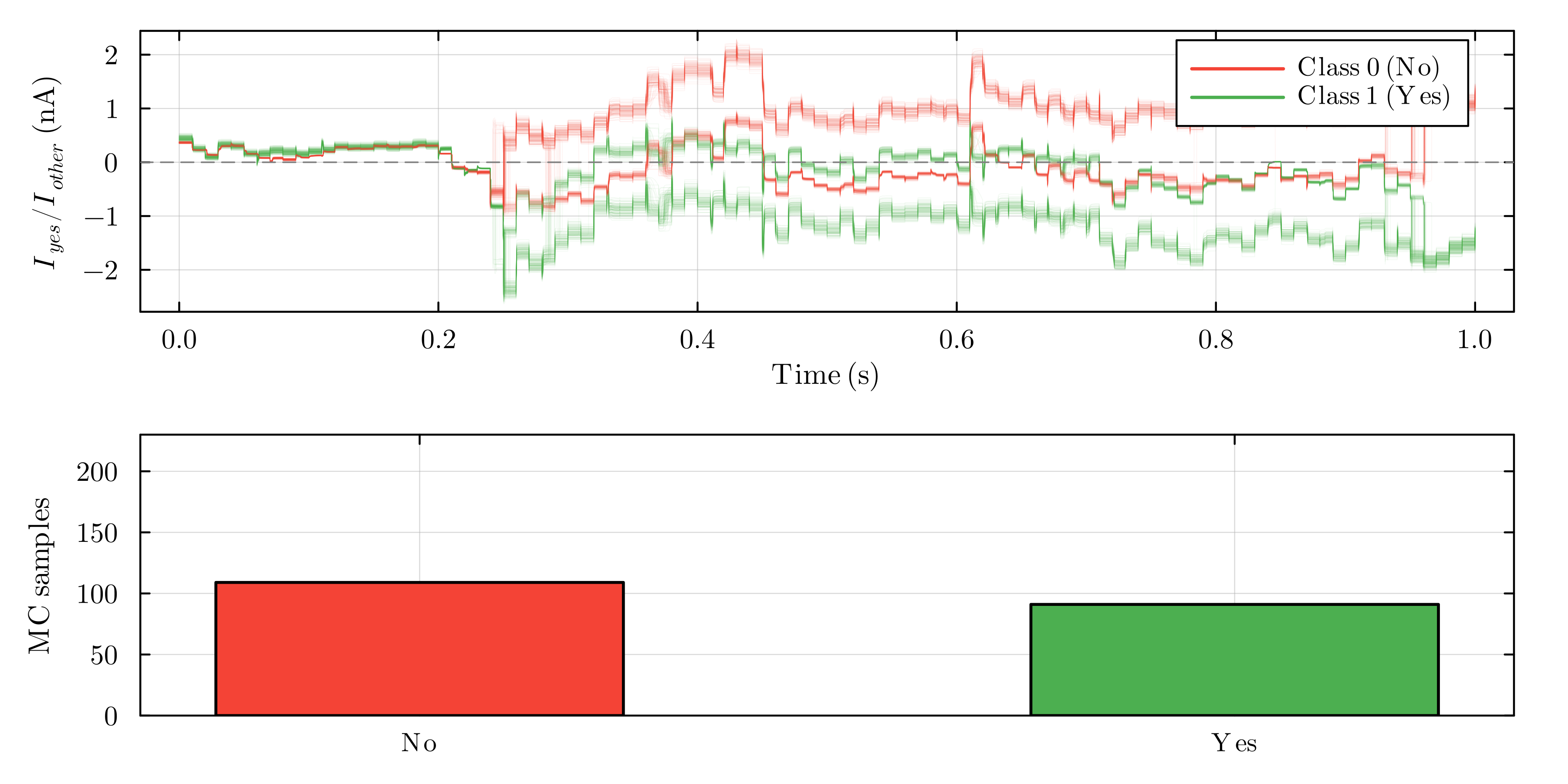}}
    \caption{
      \textbf{Monte Carlo mismatch analysis from Cadence Spectre simulation (seed 49).}
      Each panel shows output logit currents (top: $I_{\text{yes}}$ in green, $I_{\text{no}}$ in red) over the 101-frame input sequence and the corresponding prediction for each Monte Carlo sample (bottom). Analysis performed with $3\sigma$ mismatch variation on all transistors (200 samples). Spoken word: ``yes''. Nominal prediction: ``background'' (misclassified under nominal conditions). Impaired sample rate: 45.5\%.
    }
    \label{fig:monte_carlo4}
  \end{center}
\end{figure*}
\clearpage
 
\begin{figure*}[ht!]
  \vskip 0.2in
  \begin{center}
    \centerline{\includegraphics[width=\linewidth]{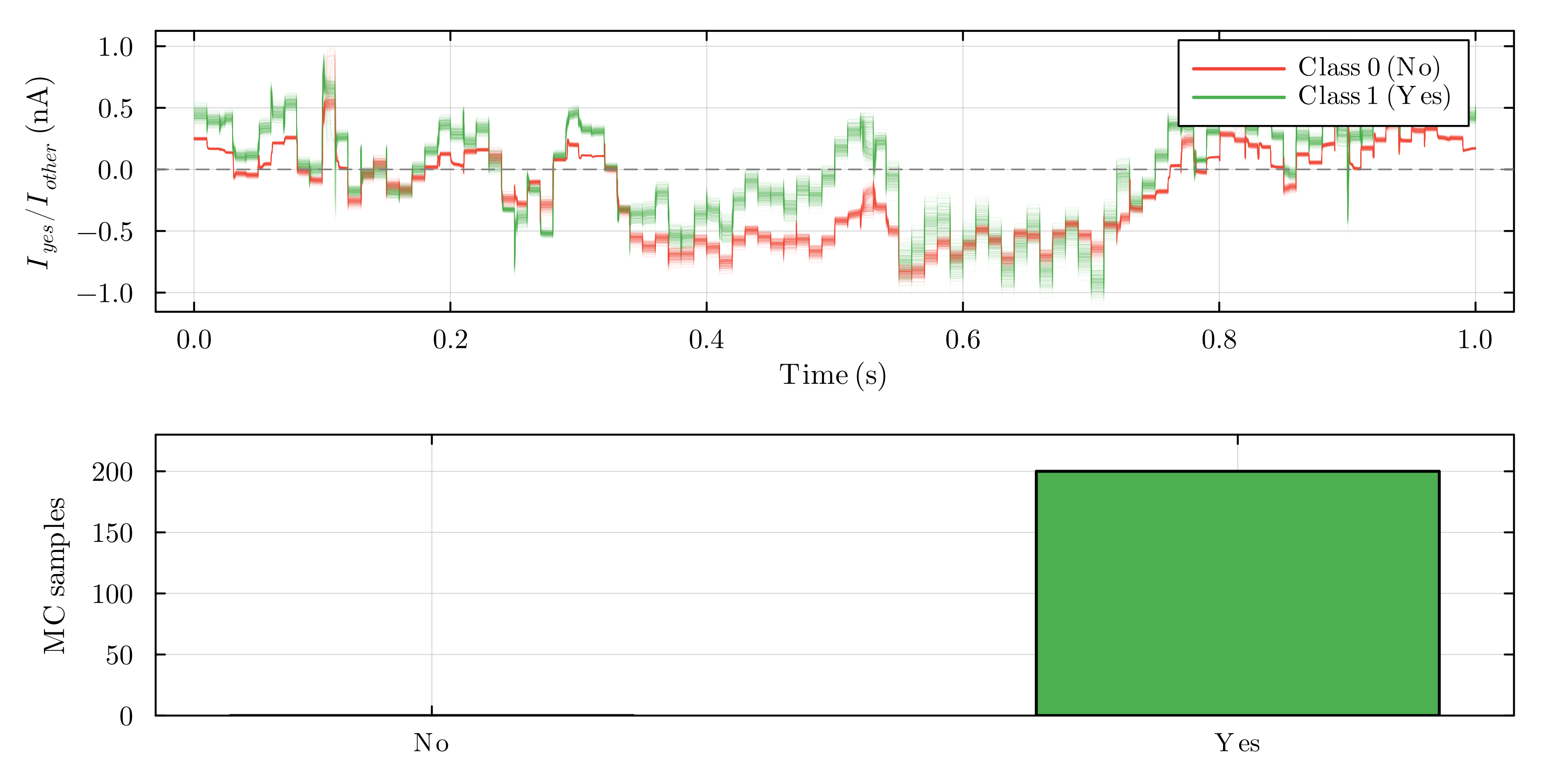}}
    \caption{
      \textbf{Monte Carlo mismatch analysis from Cadence Spectre simulation (seed 50).}
      Each panel shows output logit currents (top: $I_{\text{yes}}$ in green, $I_{\text{no}}$ in red) over the 101-frame input sequence and the corresponding prediction for each Monte Carlo sample (bottom). Analysis performed with $3\sigma$ mismatch variation on all transistors (200 samples). Spoken word: ``yes''. Nominal prediction: ``yes''. Impaired sample rate: 0\%.
    }
    \label{fig:monte_carlo5}
  \end{center}
\end{figure*}

\begin{figure*}[ht!]
  \vskip 0.2in
  \begin{center}
    \centerline{\includegraphics[width=\linewidth]{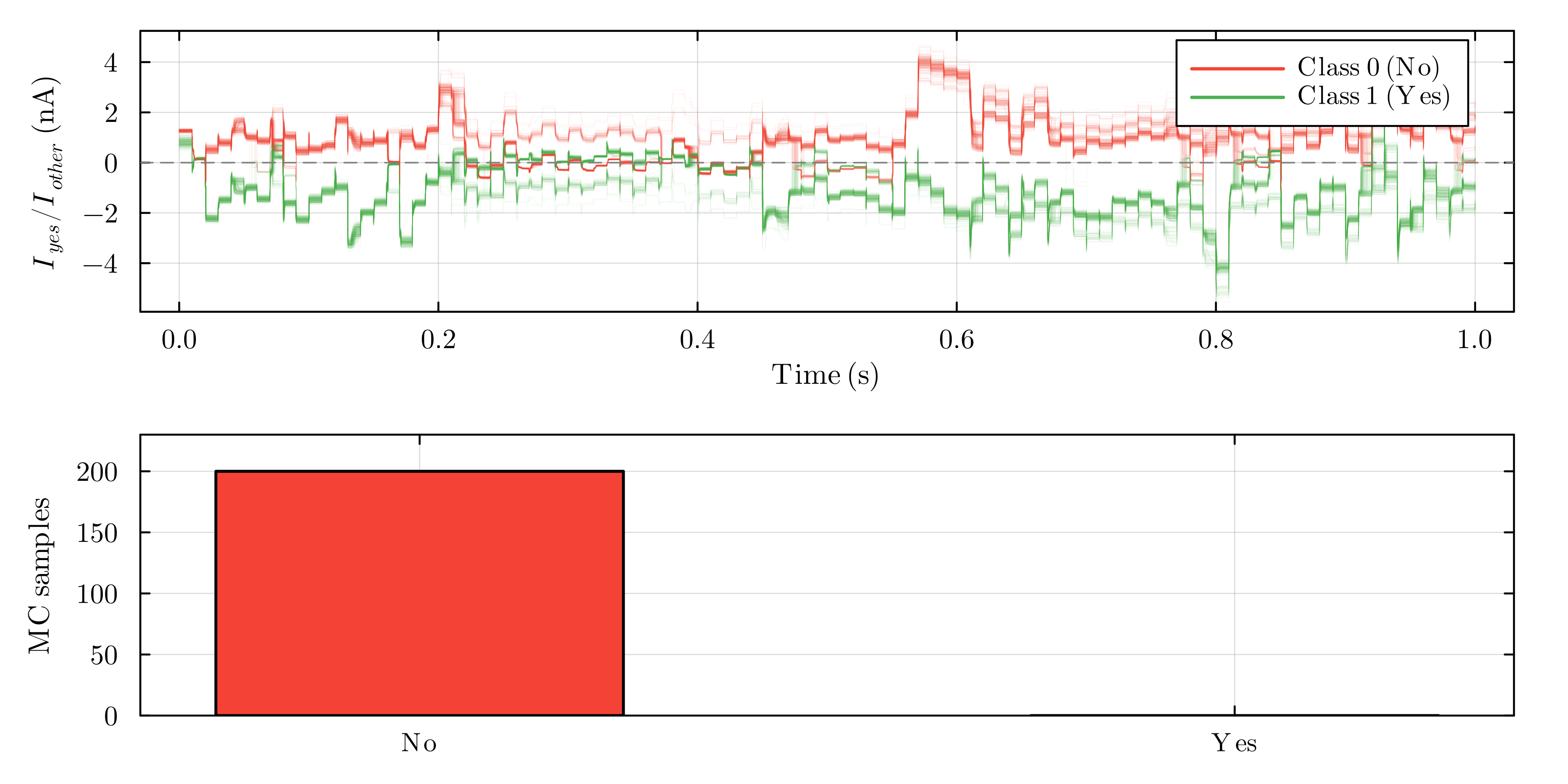}}
    \caption{
      \textbf{Monte Carlo mismatch analysis from Cadence Spectre simulation (seed 52).}
      Each panel shows output logit currents (top: $I_{\text{yes}}$ in green, $I_{\text{no}}$ in red) over the 101-frame input sequence and the corresponding prediction for each Monte Carlo sample (bottom). Analysis performed with $3\sigma$ mismatch variation on all transistors (200 samples). Input: background noise (no speech). Nominal prediction: ``background''. Impaired sample rate: 0\%.
    }
    \label{fig:monte_carlo6}
  \end{center}
\end{figure*}
\clearpage
 
\begin{figure*}[ht!]
  \vskip 0.2in
  \begin{center}
    \centerline{\includegraphics[width=\linewidth]{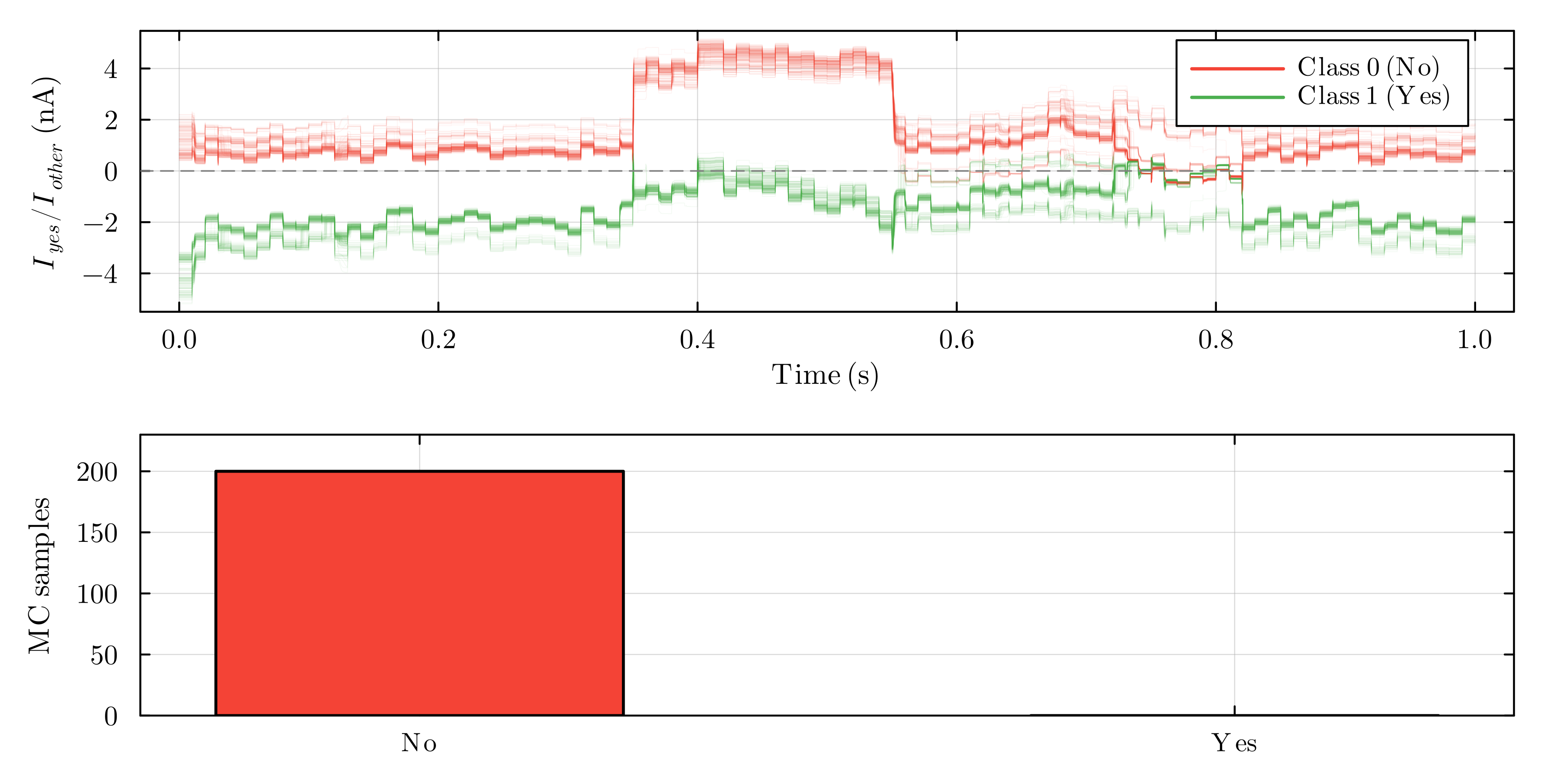}}
    \caption{
      \textbf{Monte Carlo mismatch analysis from Cadence Spectre simulation (seed 66).}
      Each panel shows output logit currents (top: $I_{\text{yes}}$ in green, $I_{\text{no}}$ in red) over the 101-frame input sequence and the corresponding prediction for each Monte Carlo sample (bottom). Analysis performed with $3\sigma$ mismatch variation on all transistors (200 samples). Spoken word: ``up''. Nominal prediction: ``background''. Impaired sample rate: 0\%.
    }
    \label{fig:monte_carlo7}
  \end{center}
\end{figure*}
  
\begin{figure*}[ht!]
  \vskip 0.2in
  \begin{center}
    \centerline{\includegraphics[width=\linewidth]{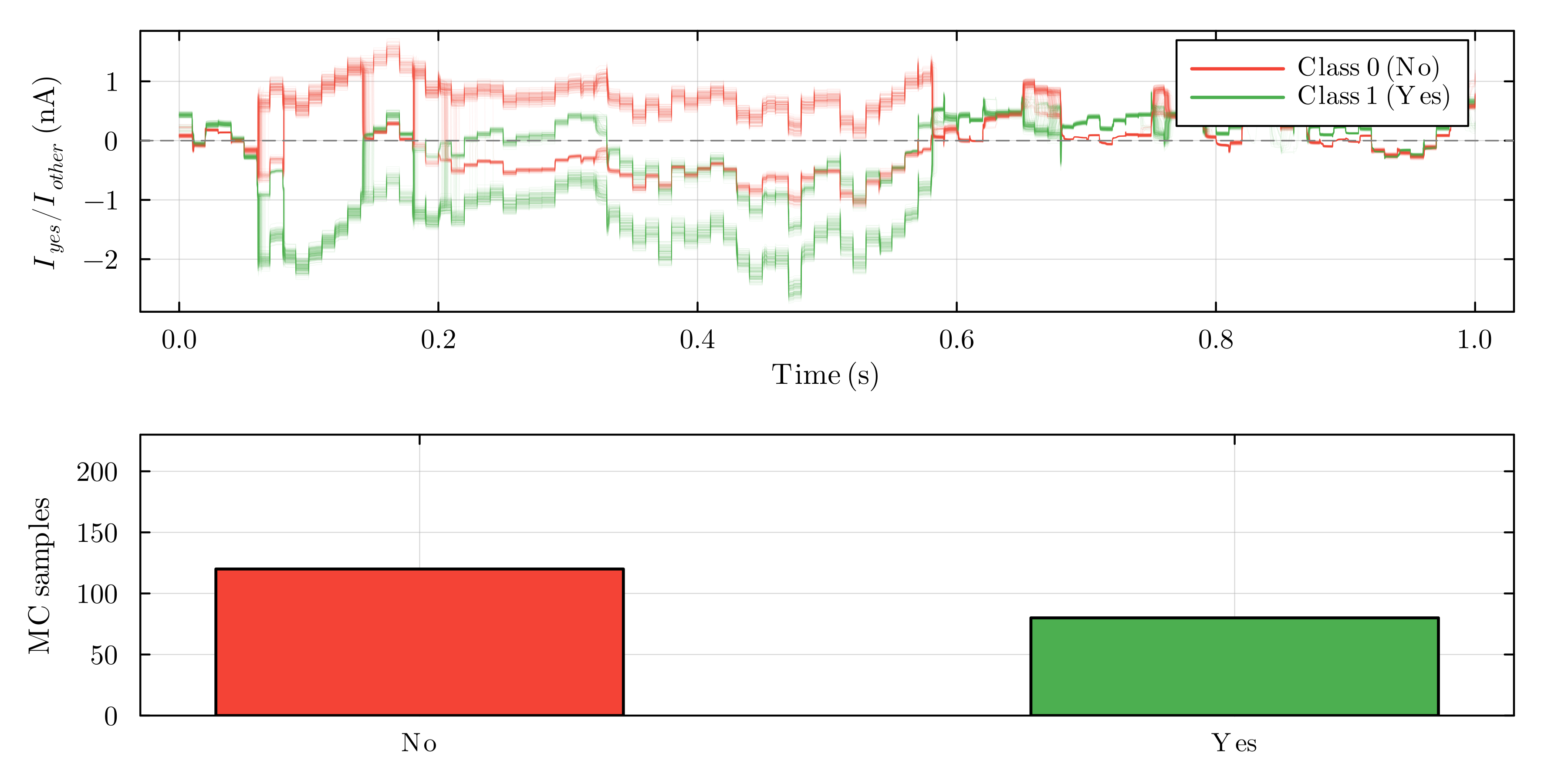}}
    \caption{
      \textbf{Monte Carlo mismatch analysis from Cadence Spectre simulation (seed 67).}
      Each panel shows output logit currents (top: $I_{\text{yes}}$ in green, $I_{\text{no}}$ in red) over the 101-frame input sequence and the corresponding prediction for each Monte Carlo sample (bottom). Analysis performed with $3\sigma$ mismatch variation on all transistors (200 samples). Spoken word: ``yes''. Nominal hardware prediction: ``background'' (already in disagreement with software under nominal conditions). Impaired sample rate: 41\%.
    }
    \label{fig:monte_carlo8}
  \end{center}
\end{figure*}
\clearpage
 
\begin{figure*}[ht!]
  \vskip 0.2in
  \begin{center}
    \centerline{\includegraphics[width=\linewidth]{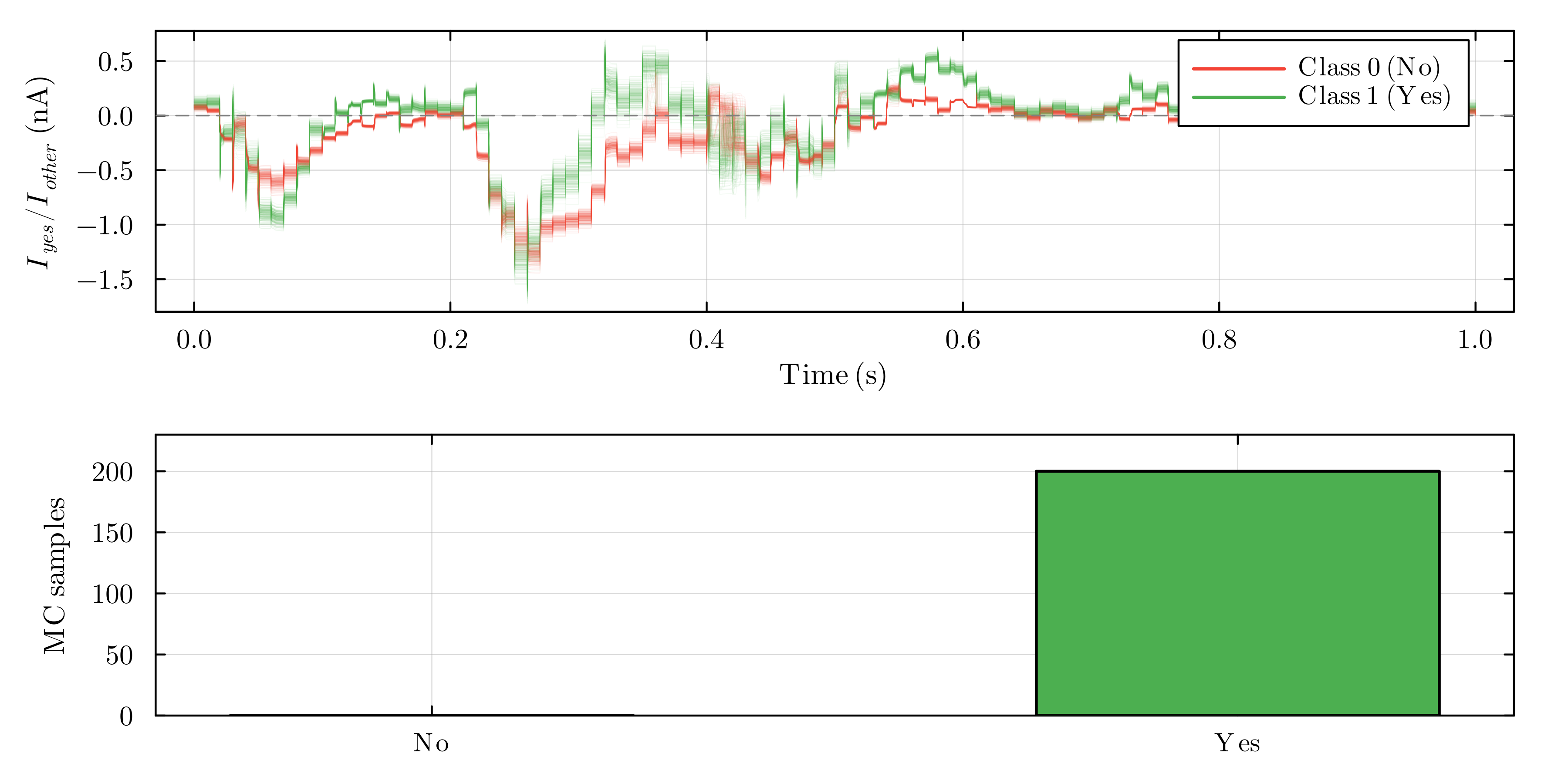}}
    \caption{
      \textbf{Monte Carlo mismatch analysis from Cadence Spectre simulation (seed 68).}
      Each panel shows output logit currents (top: $I_{\text{yes}}$ in green, $I_{\text{no}}$ in red) over the 101-frame input sequence and the corresponding prediction for each Monte Carlo sample (bottom). Analysis performed with $3\sigma$ mismatch variation on all transistors (200 samples). Spoken word: ``yes''. Nominal prediction: ``yes''. Impaired sample rate: 0\%.
    }
    \label{fig:monte_carlo9}
  \end{center}
\end{figure*}

\begin{figure*}[ht!]
  \vskip 0.2in
  \begin{center}
    \centerline{\includegraphics[width=\linewidth]{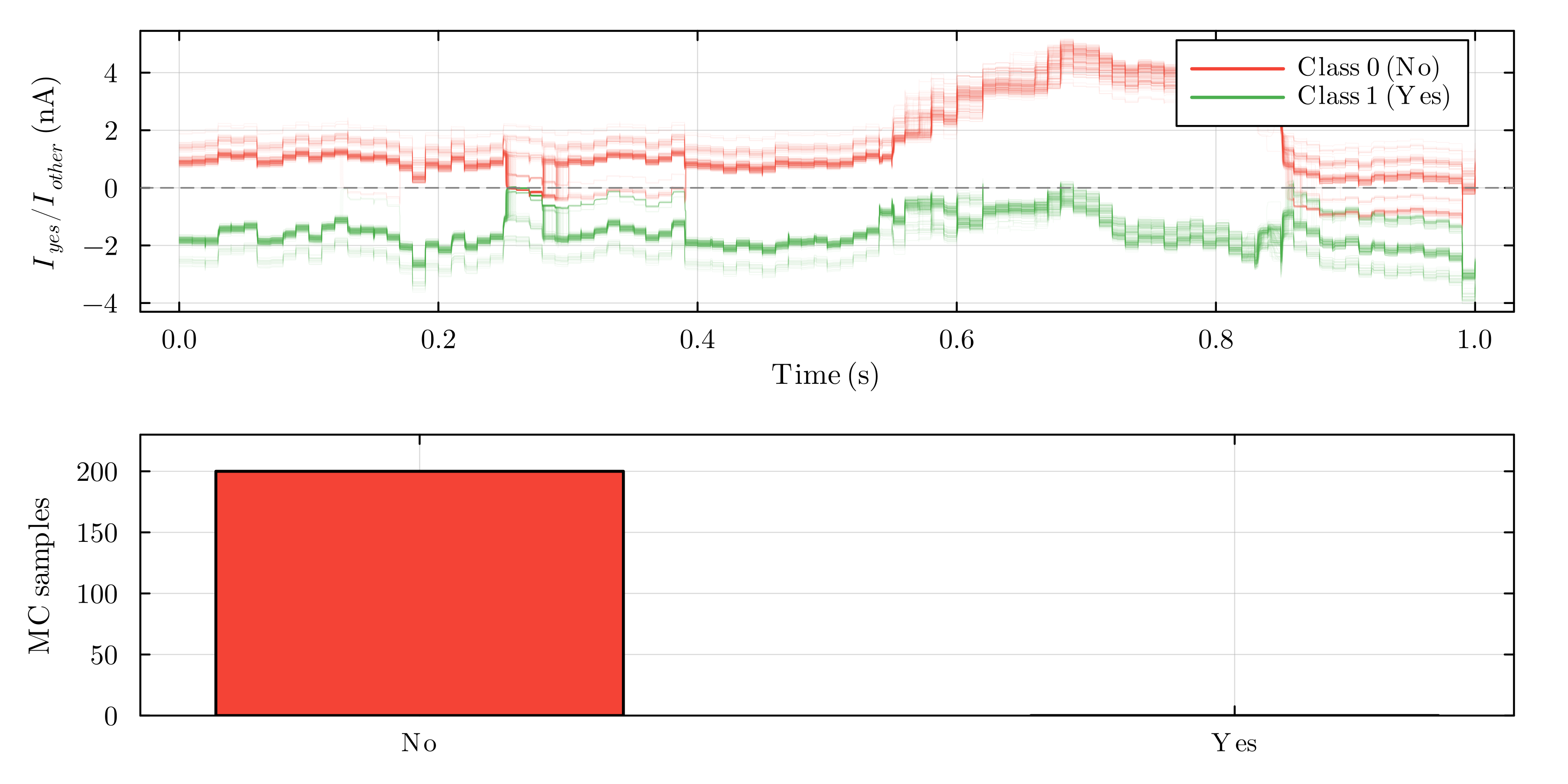}}
    \caption{
      \textbf{Monte Carlo mismatch analysis from Cadence Spectre simulation (seed 61).}
      Each panel shows output logit currents (top: $I_{\text{yes}}$ in green, $I_{\text{no}}$ in red) over the 101-frame input sequence and the corresponding prediction for each Monte Carlo sample (bottom). Analysis performed with $3\sigma$ mismatch variation on all transistors (200 samples). Spoken word: ``right''. Nominal prediction: ``background''. Impaired sample rate: 0\%.
    }
    \label{fig:monte_carlo10}
  \end{center}
\end{figure*}

\clearpage

\begin{figure*}[ht!]
  \begin{center}
    \centerline{\includegraphics[width=\linewidth]{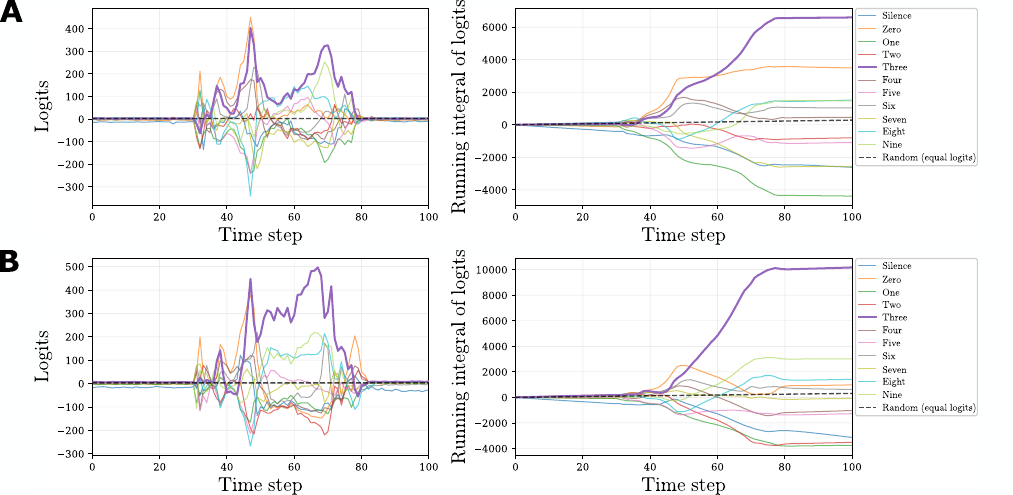}}
    \caption{
      \textbf{Multi-class KWS evaluation (11 classes), "three" spoken.}
      \textbf{A.}~$2\times 4$ network. Logit time evolution (left) and integrated logits used for the final classification decision (right). The classification is correct, but the narrow decision margins leave the prediction vulnerable to mismatch.
      \textbf{B.}~$2\times 16$ network. Logit time evolution (left) and integrated logits (right). The classification is correct and the decision margins are substantially wider, improving robustness. Networks of this size remain within sub-microwatt consumption estimates (Table~\ref{tab:power_scaling}).
    }
    \label{fig:multiclass}
  \end{center}
\end{figure*}

\begin{figure*}[ht!]
  \begin{center}
    \centerline{\includegraphics[width=\linewidth]{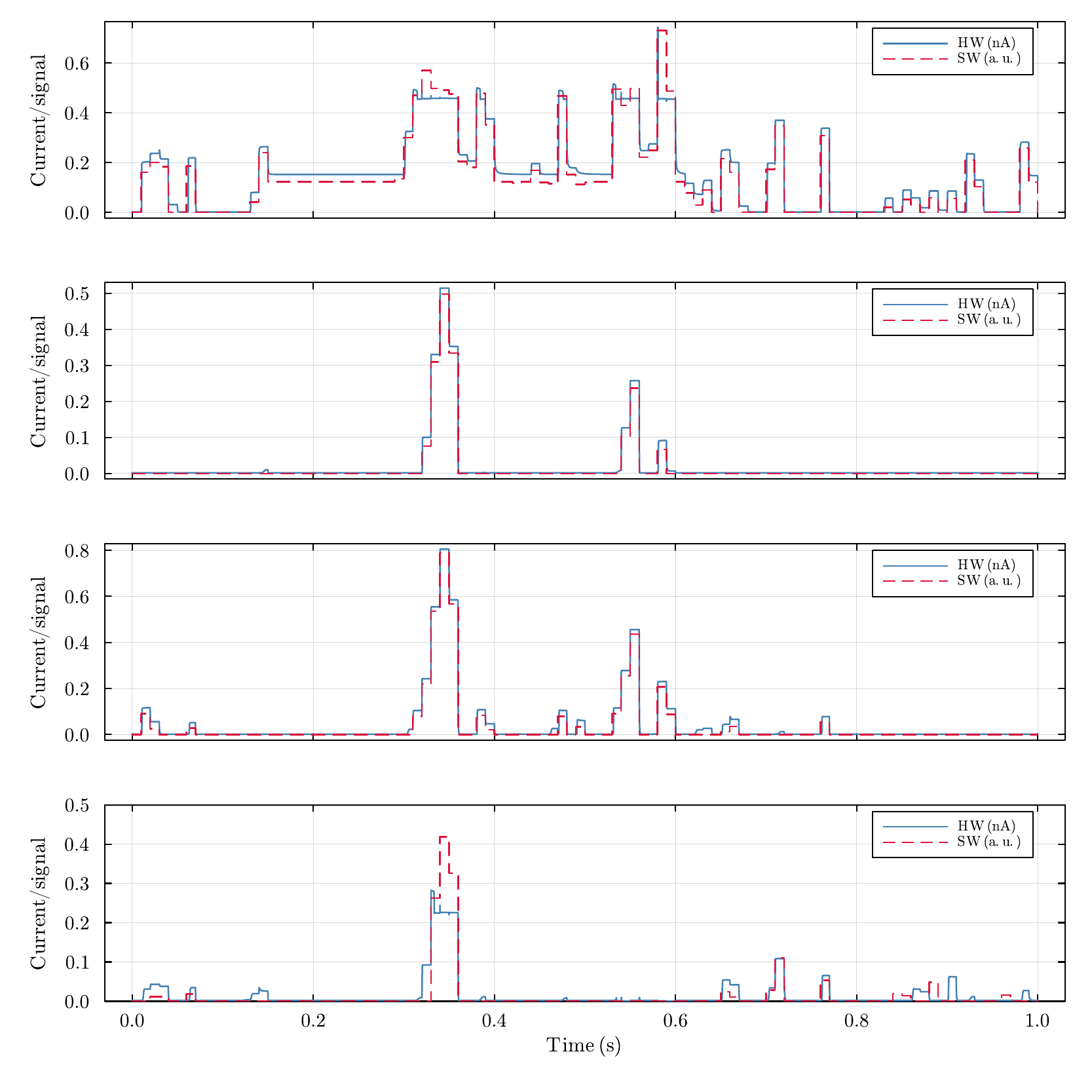}}
    \caption{
      \textbf{Intermediate signal comparison between software and hardware (layer-1 candidates, seed 51).}
      Overlay of the 4 software-predicted and Cadence-simulated candidate currents of the first recurrent layer for a representative ``yes'' inference sample.
    }
    \label{fig:signal_agreement0}
  \end{center}
\end{figure*}

\begin{figure*}[ht!]
  \begin{center}
    \centerline{\includegraphics[width=\linewidth]{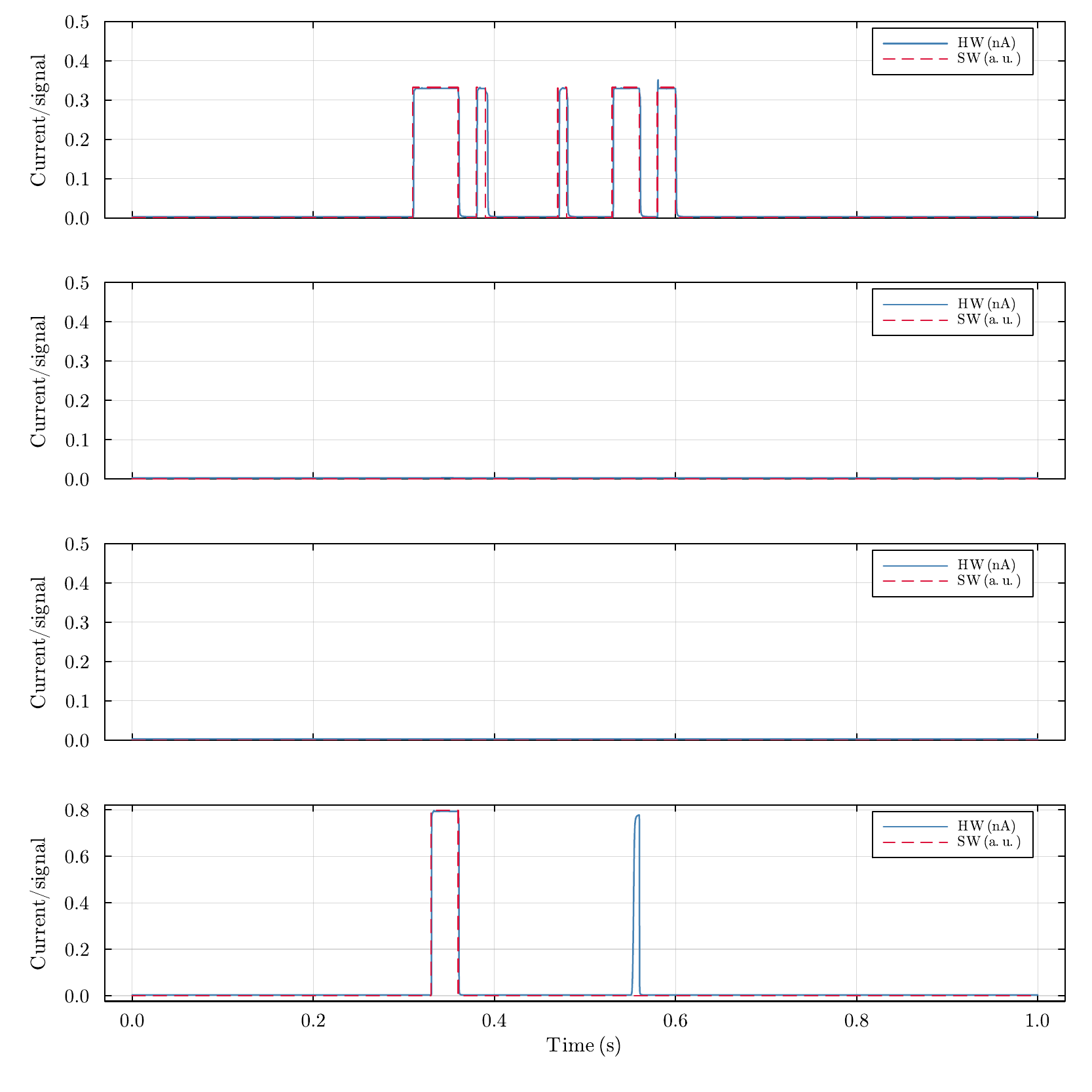}}
    \caption{
      \textbf{Intermediate signal comparison between software and hardware (layer-1 states, seed 51).}
      Overlay of the 4 software-predicted and Cadence-simulated FQ BMRU cell outputs of the first recurrent layer for a representative ``yes'' inference sample.
    }
    \label{fig:signal_agreement1}
  \end{center}
\end{figure*}
 
\begin{figure*}[ht!]
  \begin{center}
    \centerline{\includegraphics[width=\linewidth]{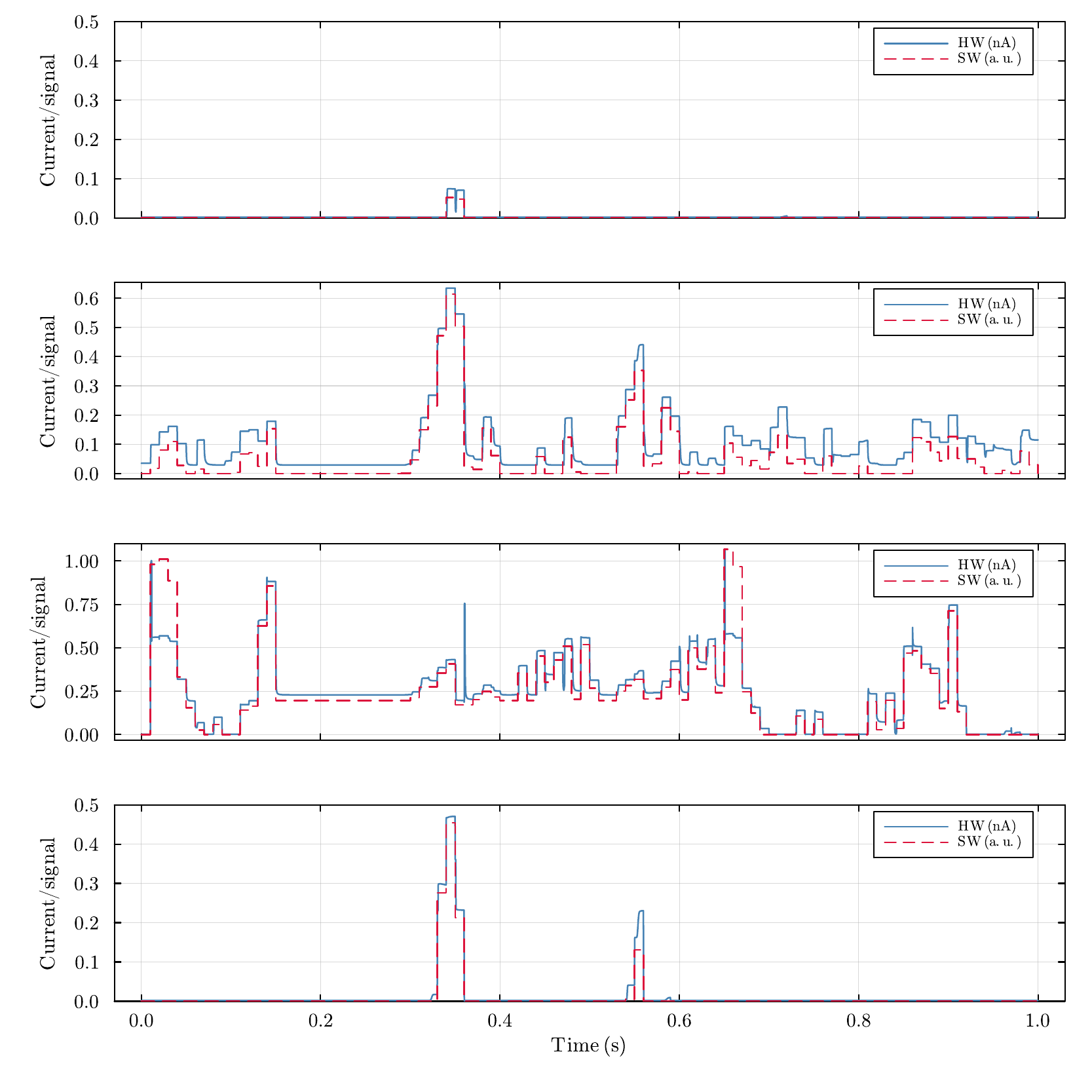}}
    \caption{
      \textbf{Intermediate signal comparison between software and hardware (layer-2 candidates, seed 51).}
      Overlay of the 4 software-predicted and Cadence-simulated candidate currents of the second recurrent layer for a representative ``yes'' inference sample.
    }
    \label{fig:signal_agreement2}
  \end{center}
\end{figure*}
 
\begin{figure*}[ht!]
  \begin{center}
    \centerline{\includegraphics[width=\linewidth]{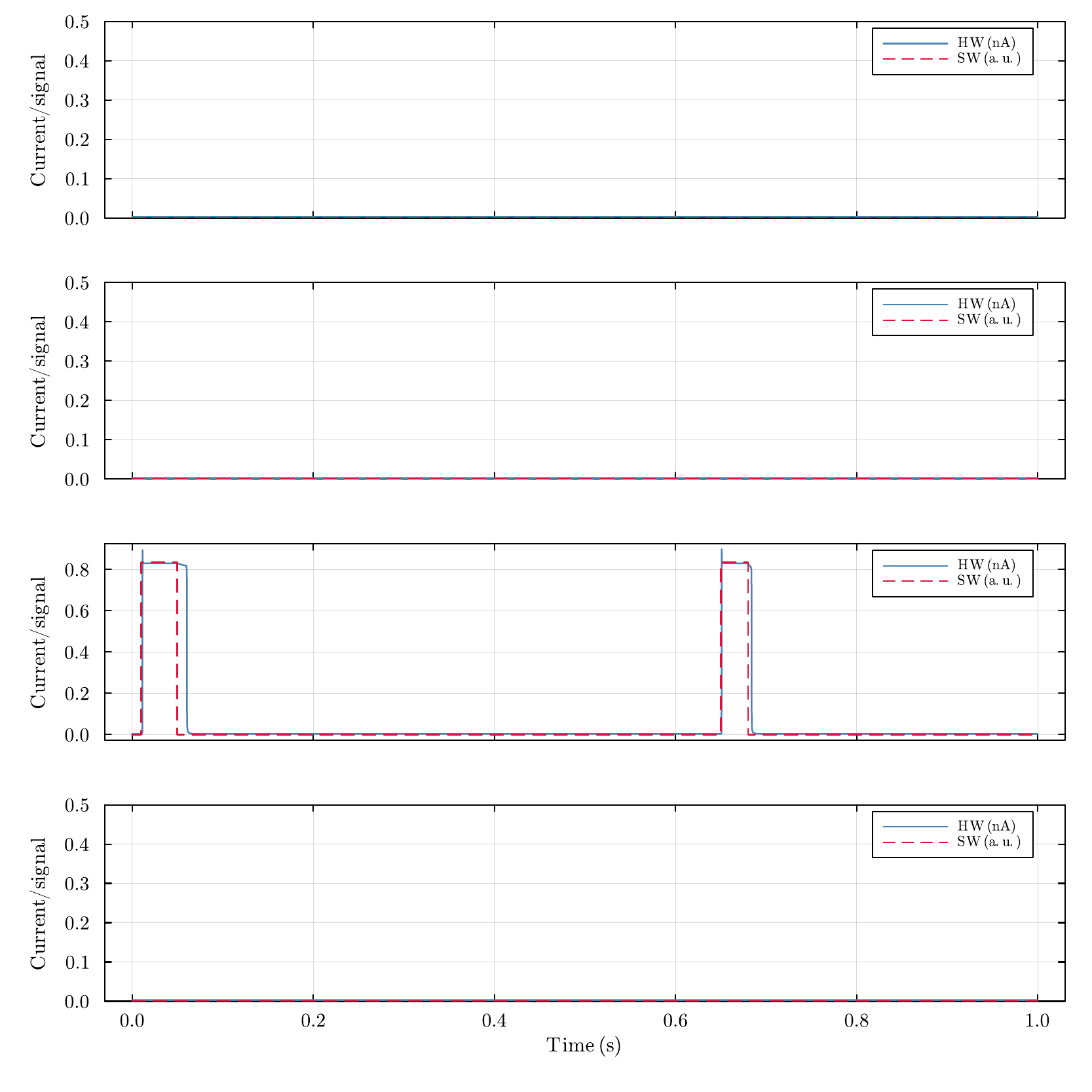}}
    \caption{
      \textbf{Intermediate signal comparison between software and hardware (layer-2 states, seed 51).}
      Overlay of the 4 software-predicted and Cadence-simulated FQ BMRU cell outputs of the second recurrent layer for a representative ``yes'' inference sample.
    }
    \label{fig:signal_agreement3}
  \end{center}
\end{figure*}
 
\begin{figure*}[ht!]
  \begin{center}
    \centerline{\includegraphics[width=\linewidth]{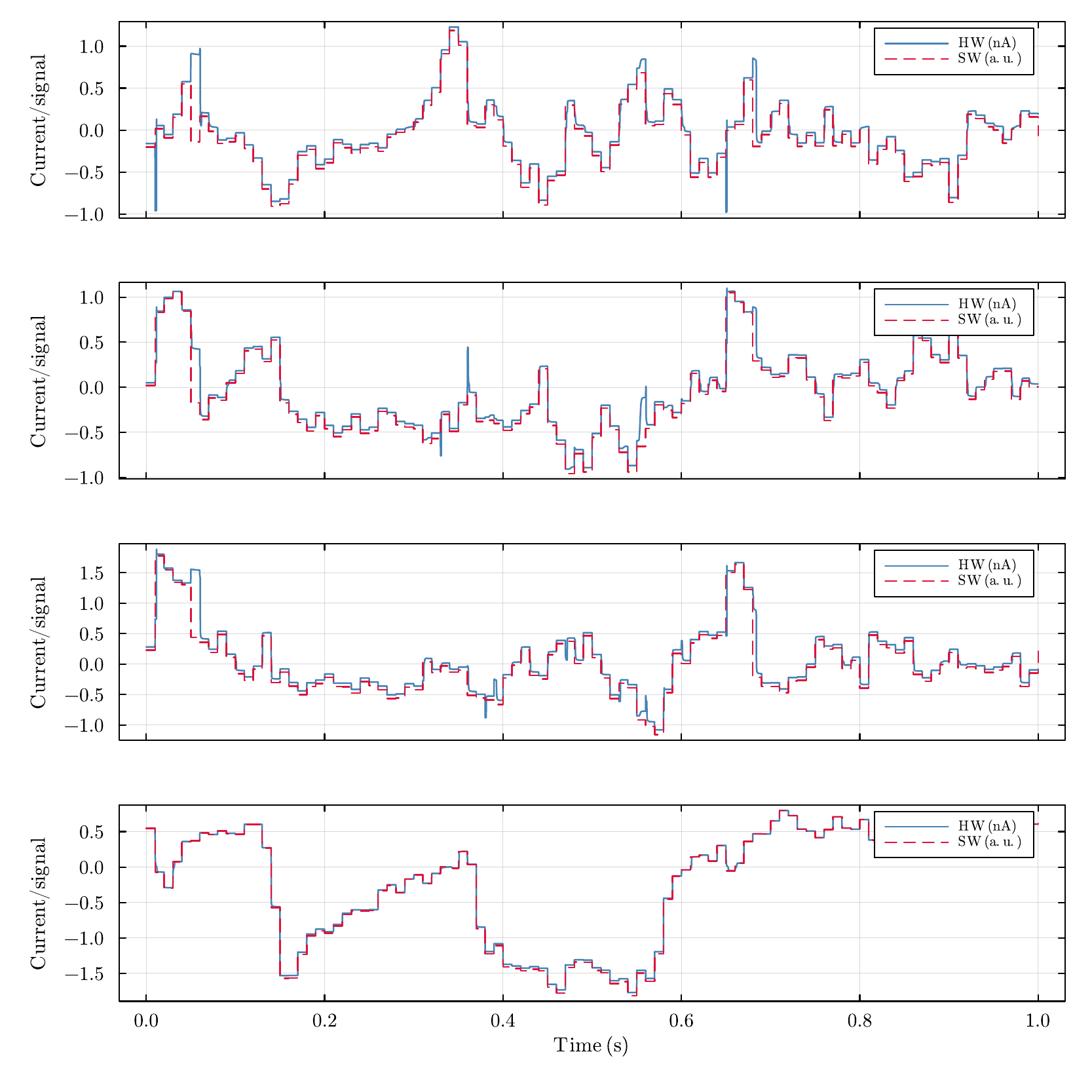}}
    \caption{
      \textbf{Intermediate signal comparison between software and hardware (layer-2 output after skip connection, seed 51).}
      Overlay of the 4 software-predicted and Cadence-simulated output signals of the second recurrent layer, after the skip connection, for a representative ``yes'' inference sample.
    }
    \label{fig:signal_agreement4}
  \end{center}
\end{figure*}
 
\begin{figure*}[ht!]
  \begin{center}
    \centerline{\includegraphics[width=\linewidth]{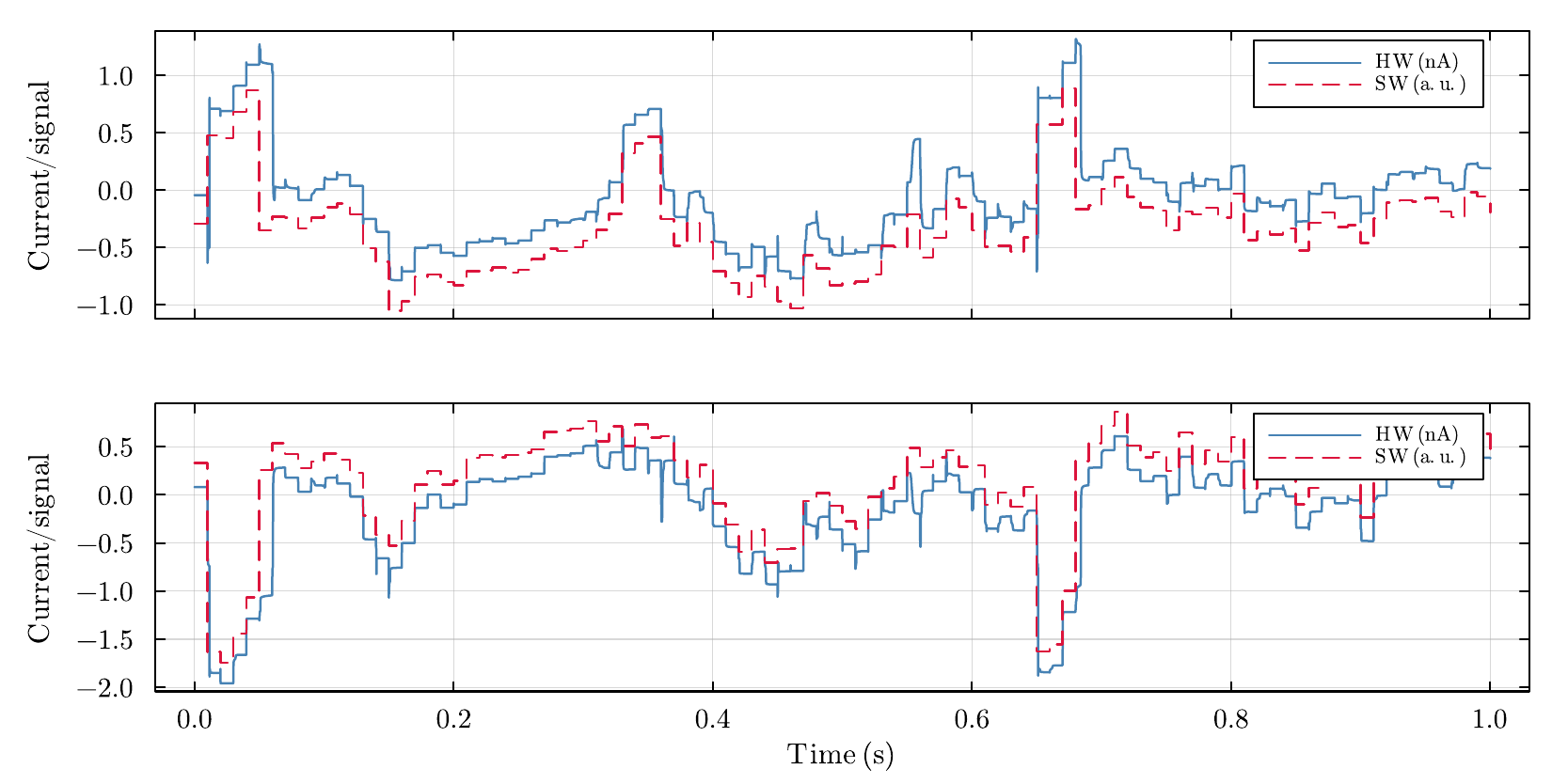}}
    \caption{
      \textbf{Intermediate signal comparison between software and hardware (output logits, seed 51).}
      Overlay of the software-predicted and Cadence-simulated output logit currents for a representative ``yes'' inference sample.
    }
    \label{fig:signal_agreement5}
  \end{center}
\end{figure*}
 
\begin{figure*}[ht!]
  \begin{center}
    \centerline{\includegraphics[width=\linewidth]{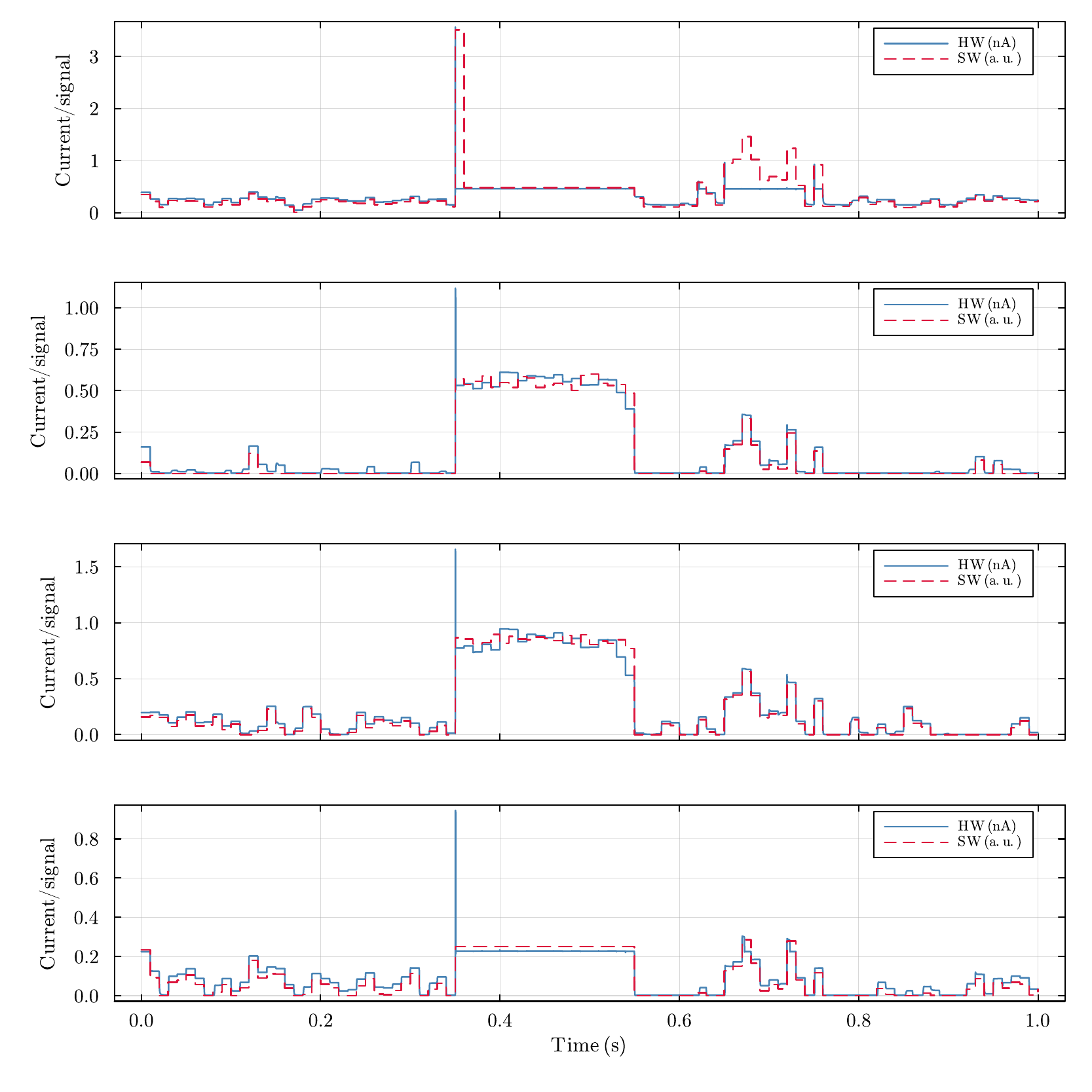}}
    \caption{
      \textbf{Intermediate signal comparison between software and hardware (layer-1 candidates, seed 66).}
      Overlay of the 4 software-predicted and Cadence-simulated candidate currents of the first recurrent layer for a representative ``background'' inference sample.
    }
    \label{fig:signal_agreement6}
  \end{center}
\end{figure*}
 
\begin{figure*}[ht!]
  \begin{center}
    \centerline{\includegraphics[width=\linewidth]{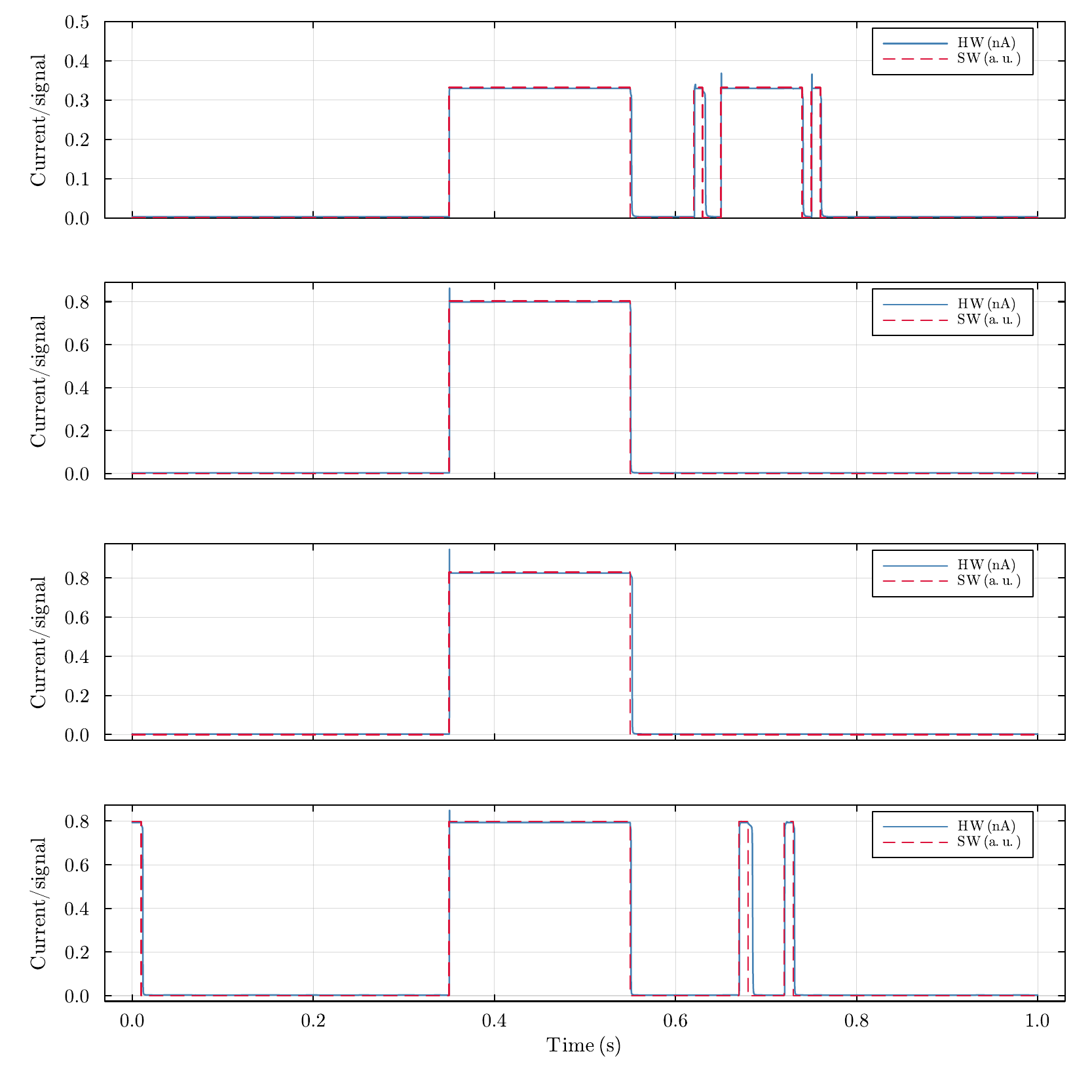}}
    \caption{
      \textbf{Intermediate signal comparison between software and hardware (layer-1 states, seed 66).}
      Overlay of the 4 software-predicted and Cadence-simulated FQ BMRU cell outputs of the first recurrent layer for a representative ``background'' inference sample.
    }
    \label{fig:signal_agreement7}
  \end{center}
\end{figure*}
 
\begin{figure*}[ht!]
  \begin{center}
    \centerline{\includegraphics[width=\linewidth]{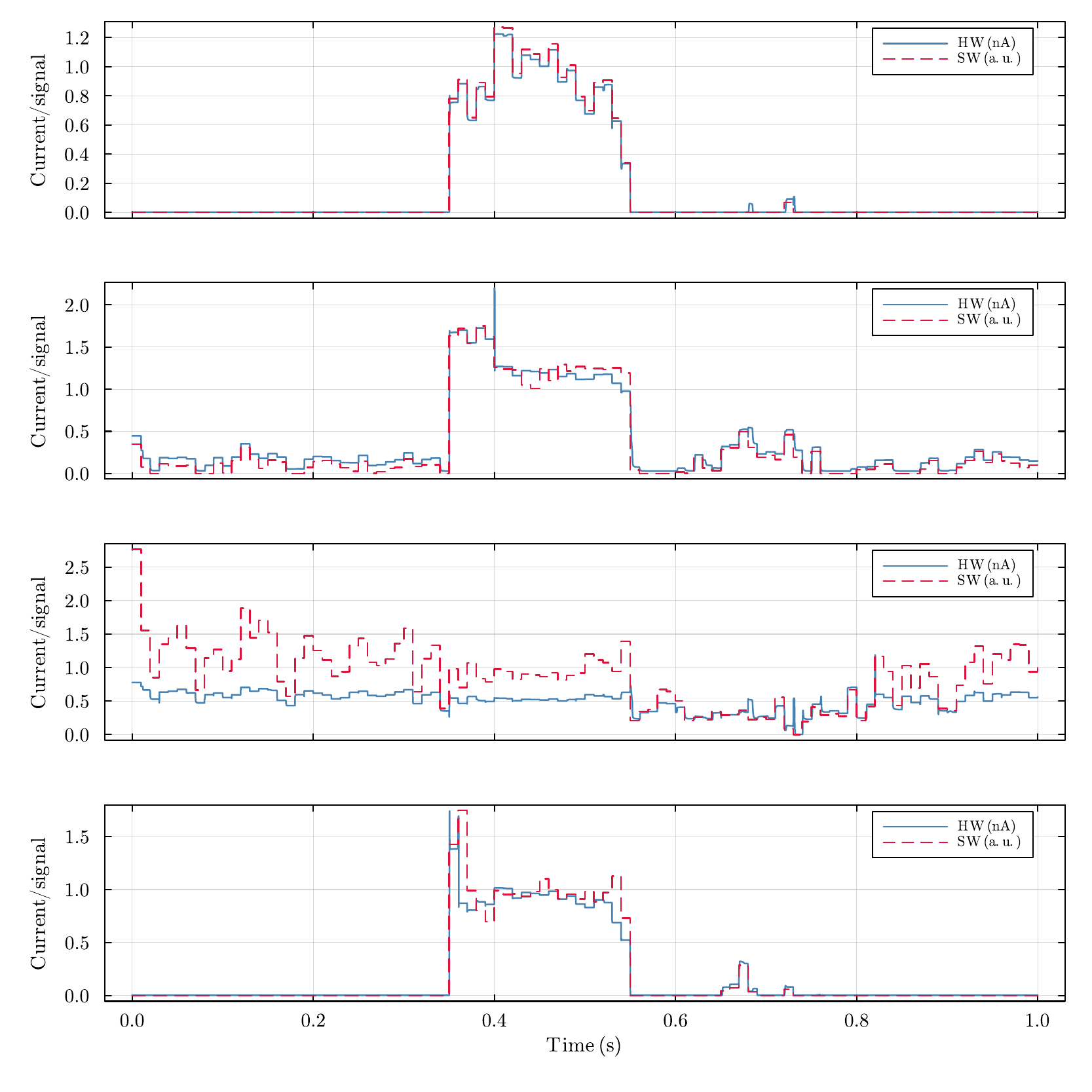}}
    \caption{
      \textbf{Intermediate signal comparison between software and hardware (layer-2 candidates, seed 66).}
      Overlay of the 4 software-predicted and Cadence-simulated candidate currents of the second recurrent layer for a representative ``background'' inference sample.
    }
    \label{fig:signal_agreement8}
  \end{center}
\end{figure*}
 
\begin{figure*}[ht!]
  \begin{center}
    \centerline{\includegraphics[width=\linewidth]{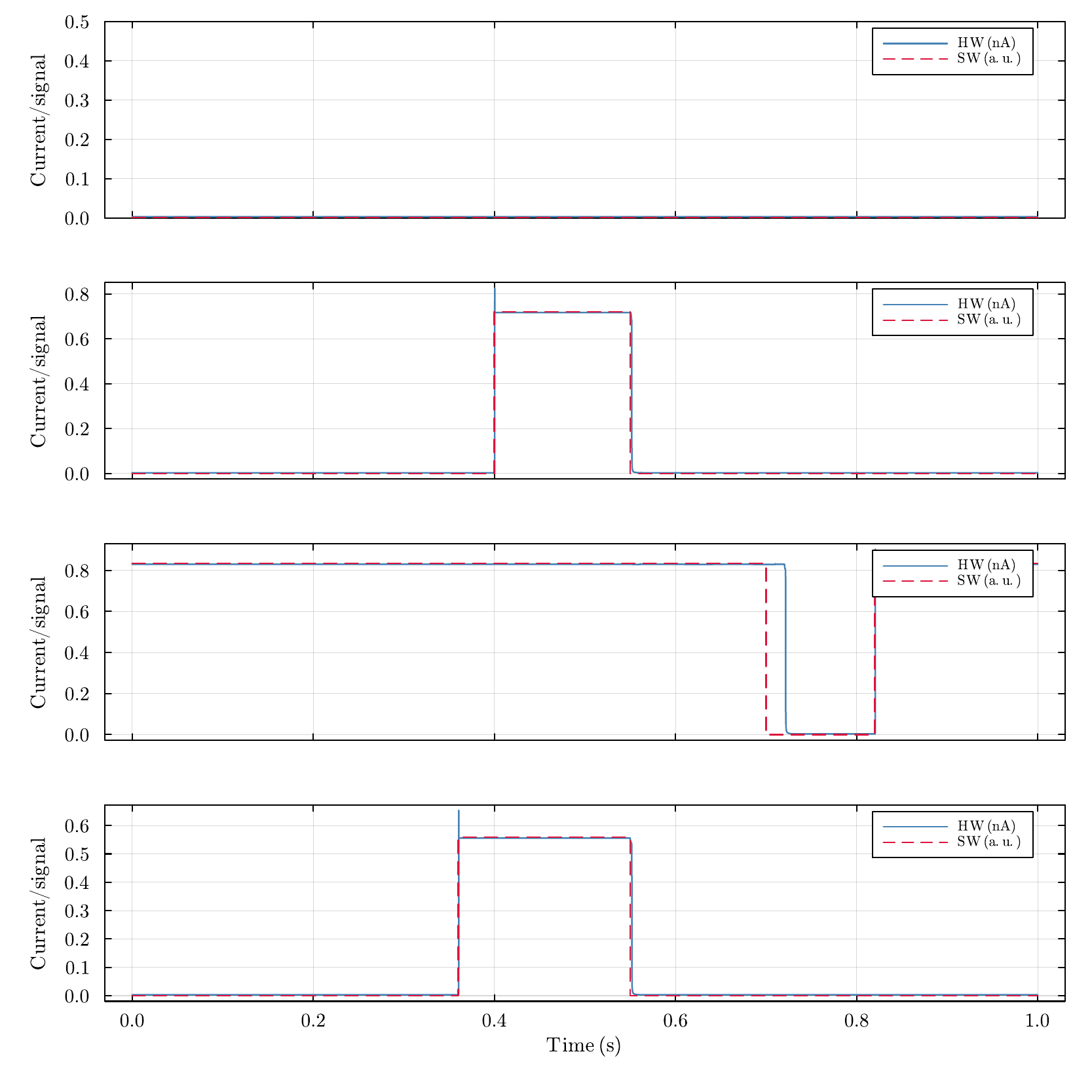}}
    \caption{
      \textbf{Intermediate signal comparison between software and hardware (layer-2 states, seed 66).}
      Overlay of the 4 software-predicted and Cadence-simulated FQ BMRU cell outputs of the second recurrent layer for a representative ``background'' inference sample.
    }
    \label{fig:signal_agreement9}
  \end{center}
\end{figure*}
 
\begin{figure*}[ht!]
  \begin{center}
    \centerline{\includegraphics[width=\linewidth]{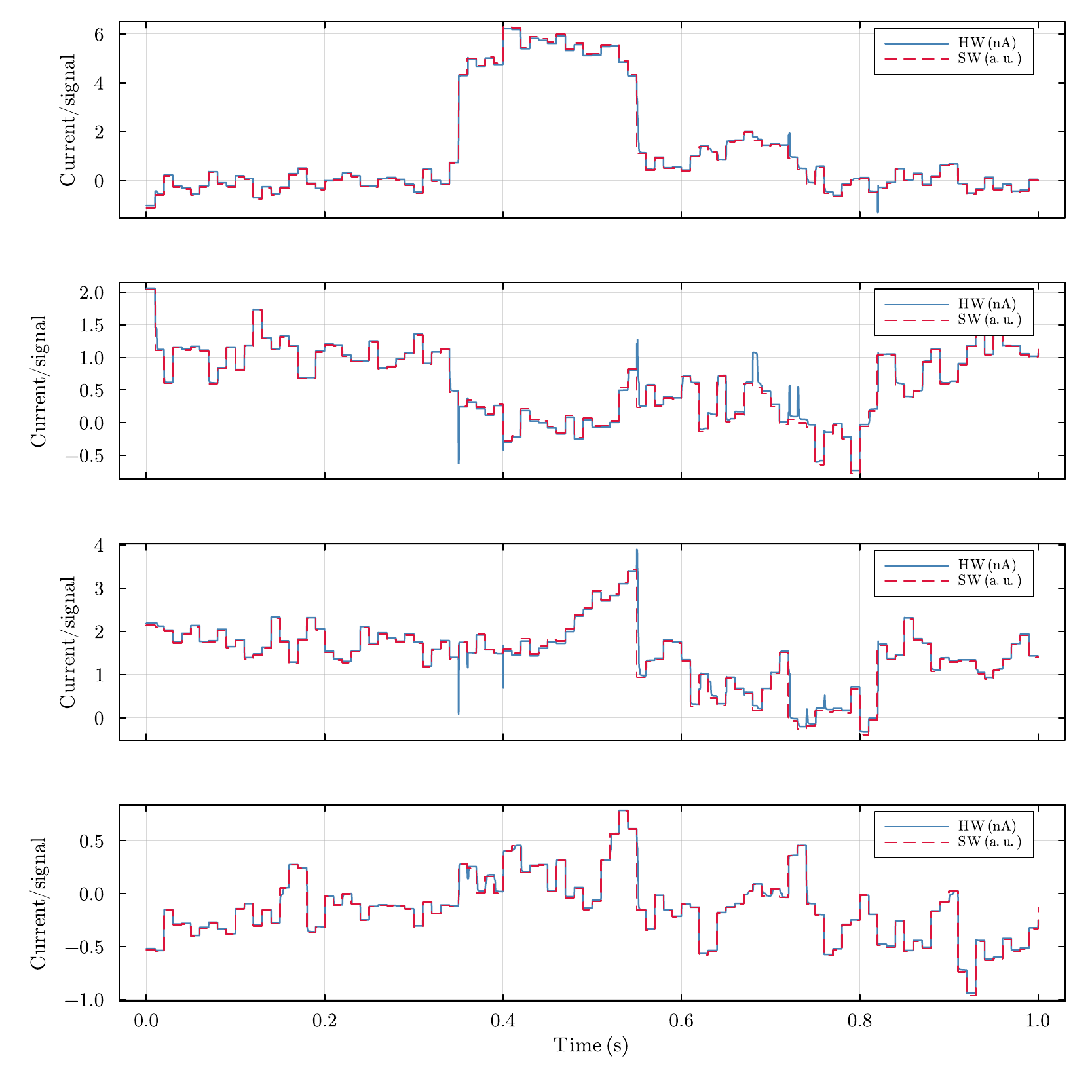}}
    \caption{
      \textbf{Intermediate signal comparison between software and hardware (layer-2 output after skip connection, seed 66).}
      Overlay of the 4 software-predicted and Cadence-simulated output signals of the second recurrent layer, after the skip connection, for a representative ``background'' inference sample.
    }
    \label{fig:signal_agreement10}
  \end{center}
\end{figure*}
 
\begin{figure*}[ht!]
  \begin{center}
    \centerline{\includegraphics[width=\linewidth]{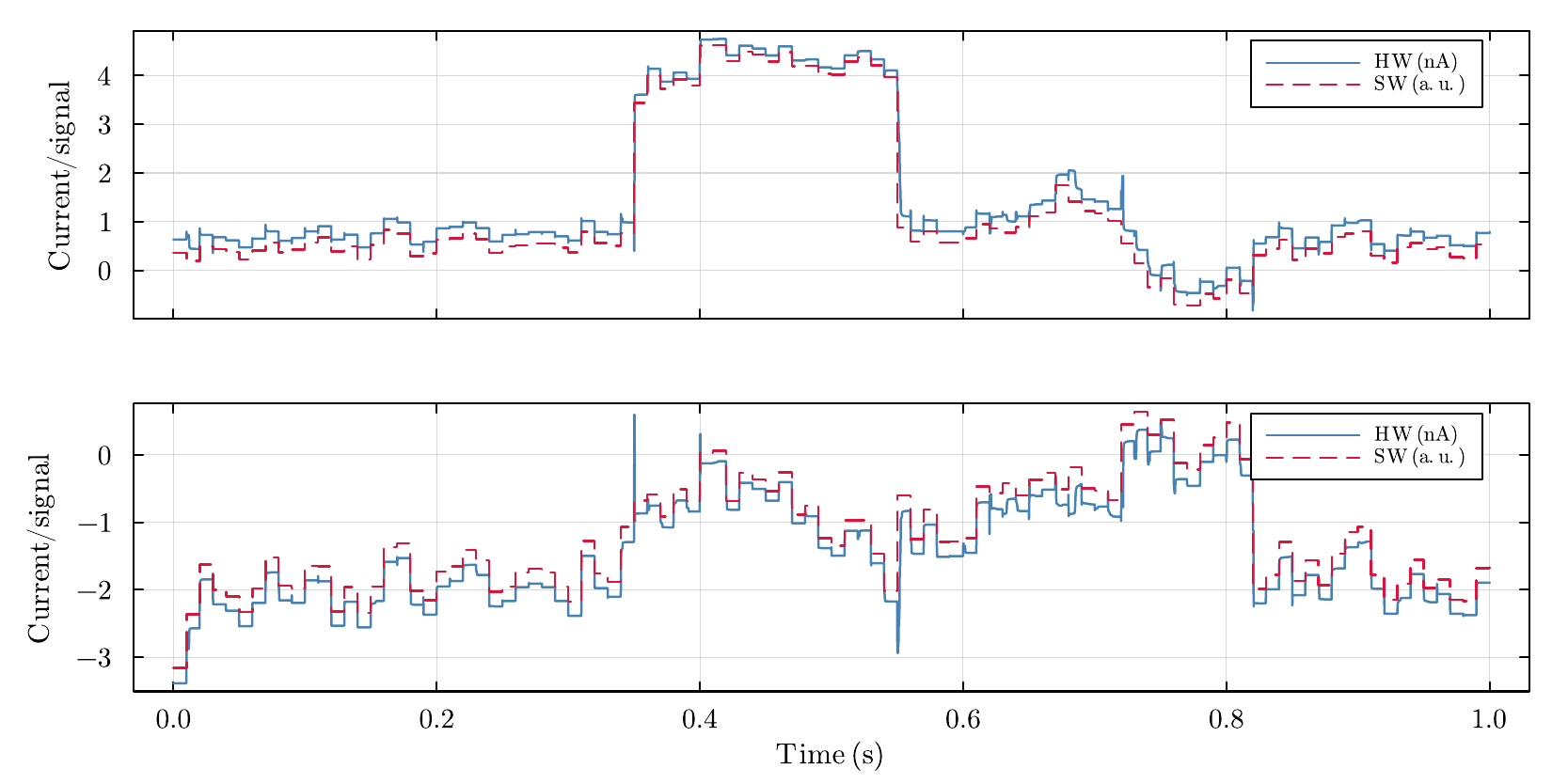}}
    \caption{
      \textbf{Intermediate signal comparison between software and hardware (output logits, seed 66).}
      Overlay of the software-predicted and Cadence-simulated output logit currents for a representative ``background'' inference sample.
    }
    \label{fig:signal_agreement11}
  \end{center}
\end{figure*}


\end{document}